\documentclass[regno]{amsart}
\usepackage[british]{babel}
\usepackage[OT2,T2A,T1]{fontenc}
\usepackage[utf8]{inputenc}		
\usepackage{eulervm}	
\usepackage[scr=dutchcal, bb=boondox]{mathalfa}
\newcommand\quotient[2]{{^{\displaystyle #1}}\Big/{_{\displaystyle #2}}}

\usepackage{titlesec}
\titleformat{\section}
	{\centering\Large\scshape\bfseries}{\thesection.}{1em}{} 
\titleformat{\subsection}
 	{\normalfont\bfseries\scshape}{\thesubsection.}{.5em}{}
\titleformat*{\paragraph}{\bfseries}

\usepackage{titletoc}

    \titlecontents{section}
    [0em] %
    {\medskip\bfseries\scshape}
    {\thecontentslabel\quad}
    {}
    {\hfill\contentspage}

    \titlecontents{subsection}[1em]{\smallskip}%
    {\thecontentslabel.\enspace}
    {}
    {\titlerule*[1.2pc]{.}\contentspage}%

    \newenvironment{acknowledgements}{%
  
  \begin{abstract}
}{%
  \end{abstract}
}

\usepackage{amsmath,amsthm,amssymb,dsfont,extarrows,amscd,amsfonts,mathtools,bm,physics}
\usepackage[v1]{subfiles}
\usepackage{hyperref}
\usepackage{mathrsfs}
\usepackage[mathscr]{euscript}
\usepackage{tensor}
\usepackage[most]{tcolorbox}
\usepackage{stmaryrd}
\usepackage{cleveref}
\usepackage{MnSymbol}
\usepackage{thmtools}
\usepackage{todonotes}
\usepackage{geometry}
\usepackage[alphabetic, initials]{amsrefs}

\DeclareFontFamily{U}{MnSymbolC}{}
\DeclareSymbolFont{MnSyC}{U}{MnSymbolC}{m}{n}
\DeclareFontShape{U}{MnSymbolC}{m}{n}{
	<-6>  MnSymbolC5
	<6-7>  MnSymbolC6
	<7-8>  MnSymbolC7
	<8-9>  MnSymbolC8
	<9-10> MnSymbolC9
	<10-12> MnSymbolC10
	<12->   MnSymbolC12}{}
\DeclareMathSymbol{\intprod}{\mathbin}{MnSyC}{'270}

\theoremstyle{plain} 
\newtheorem{thm}{Theorem} [section]

\newtheorem{lem}[thm]{Lemma} 
\newtheorem{prop}{Proposition}[section]

\theoremstyle{definition}
\newtheorem{defn}{Definition}[section]
\newtheorem{ex}{Example}[section]

\theoremstyle{remark} 
\newtheorem{oss}{Remark} [section]

\numberwithin{equation}{section}

\newcommand{\R}{\mathbb{R}}
\newcommand{\CC}{\mathbb{C}}
\newcommand{\Z}{\mathbb{Z}}

\DeclareMathOperator{\X}{\mathfrak{X}}
\DeclareMathOperator{\A}{\mathbb{A}}
\DeclareMathOperator{\Lie}{\mathscr{L}}

\newcommand{\Hol}{\mathscr{O}}
\newcommand{\Aut}[1]{\mathrm{Aut}(#1)}
\newcommand{\RS}{\mathbb{P}^{1}}
\newcommand{\Hw}{\mathscr{H}}
\newcommand{\AN}{\mathscr{A}}
\newcommand{\Fr}{\mathscr{F}}
\newcommand{\Sym}{\mathfrak{S}}

\newcommand{\inp}[2]{\langle\, #1\,,\, #2 \,\rangle}

\newcommand{\vin}{\rotatebox[origin=c]{90}{$\in$}}

\DeclareMathOperator{\id}{id}

\title{Diagonal invariants and genus-zero Hurwitz Frobenius manifolds}
\author[A.~Proserpio]{Alessandro Proserpio}
\address[A.~Proserpio]{School of Mathematics and Statistics,
University of Glasgow, G12 8QQ, Glasgow, United Kingdom.
}
\email{a.proserpio.1@research.gla.ac.uk}

\author[I. A. B.~Strachan]{Ian A. B. Strachan}
\address[I. A. B.~Strachan]{
School of Mathematics and Statistics,
University of Glasgow, G12 8QQ, Glasgow, United Kingdom.
}
\email{ian.strachan@glasgow.ac.uk}

\begin{document}

\renewcommand{\hbar}{\hslash}
\begin{abstract}
The Frobenius manifold structure on the space of rational functions with multiple simple poles is constructed. In particular, the dependence of the Saito-flat coordinates on the flat coordinates of the intersection form is studied. While some of the individual flat coordinates are complicated rational functions, they appear in the prepotential in certain combinations known as diagonal invariants, which turn out to be polynomial. Two classes are studied in more detail. These are generalisations of the Coxeter and extended-affine-Weyl orbit spaces for the group $W=W(\AN_\ell)\,$. An invariant theory is also developed.
\end{abstract}


\maketitle
\tableofcontents

\section{Introduction}
\documentclass[main.tex]{subfile}
The fact that a monic polynomial may be expressed in two different ways
\begin{align}
\lambda(w)  &= w^n + a_{n-1} w^{n-1} + \ldots + a_0 \,, \label{originalpoly} \\
\quad &= \quad \prod_{i=1}^n (w - z_i)
\end{align}
was known to the 16$^{\rm th}$ century mathematician Vi\`eta, and the map 
$\bm{ z} \mapsto \bm{ a}(\bm{ z})$, which expressed the coefficients $a_i$ as functions of the zeros of polynomial, are called Vi\`eta formulae. This result lies at the foundation of invariant theory, and hence representation theory, as the functions $a_i(\bm{ z})$ are invariant under permutations of the zeros.

The space of polynomials is therefore endowed with two natural coordinate systems, and locally these are related by the map $\bm{ z} \mapsto \bm{ a}(\bm{ z})\,.$ Globally one has an orbit space structure which reflects the fact that preimage of a point $\bm{ a}$ consist of points in the same orbit, under the action of the permutation group:

\[
\begin{array}{ccc}
\{z_i\} & \in & \CC^n \\
\downmapsto& & \downarrow \\
\{a_i(\bm{ z})\} & \in & \CC^n \slash \Sym_n
\end{array}
\]

This structure plays a foundational role in the theory of Frobenius manifolds. In the classical Saito construction one starts with a finite-dimensional Coxeter group $W$ of rank $n$ acting by reflections on a space $V$, and from the invariant polynomials on the complexification $V_{\CC} \cong \CC^n$ one forms the quotient space $\CC^n/W\,.$ From Chevalley's Theorem $\CC[\bm{ z}]^W = \CC[\bm{ a}]$ and hence the quotient space is a complex manifold. The natural invariant quadratic polynomial (a complex bilinear symmetric form, referred to as a metric) descends to the orbit space (where it is known as the intersection form) and the central component of the Saito construction is the construction of a second flat metric and a coordinate system in which the components of this metric are constant -- the so-called flat coordinates for the Saito metric. From these objects one constructs the structure of a Frobenius manifold on the orbit space $\CC^n/W\,.$

With the Coxeter group $W=W(\AN_{n-1})$ (and with the constraint $a_{n-1} = -\sum_{i=1}^n z_i=0$) the polynomial (\ref{originalpoly})
serves both as a generating function for the invariant polynomials and as a so-called superpotential from which the Frobenius manifold structure on the space of such polynomials may be derived via the calculation of Grothendieck residues. In this, $\lambda$ is not only as a generating function, it is a holomorphic map from $\mathbb{P}^1 \rightarrow \mathbb{P}
^1\,.$

Such structures are easily generalised: the space of polynomials can be replaced with the space of rational functions (which are holomorphic functions from $\mathbb{P}^1 \rightarrow \mathbb{P}^1\,$) or even to the space of meromorphic maps on a higher-genus Riemann surface, and with this one obtains Frobenius manifolds on Hurwitz spaces. Even in the simplest generalization to spaces of rational functions, the dual role of $\lambda(w)$ as a generating function for invariant structures becomes less clear, as does the nature of group action. This is illustrated in the following simple example.

\begin{ex}
Consider a rational function with two zeros and two simple poles:
\begin{align}
\lambda(x)  &= \frac{(x-a)(x-b)}{(x-c)(x-d)}\,, \label{exprod}\\
&=  1 + \frac{\alpha}{x-\gamma} + \frac{\beta}{x-\delta} \label{exsum}
\end{align}
(where we assume $a,b,c,d$ are all distinct), so
\[
\alpha = \frac{(c-a)(c-b)}{(c-d)}\,, \quad \beta = \frac{(d-a)(d-b)}{(d-c)} \,, \quad \gamma = c\,, \quad \delta = d\,.
\]
The product form (\ref{exprod}) is clearly invariant under the action $\Sym_2^{(ab)} \times \Sym_2^{(cd)}$ which permutes (separately) $a \leftrightarrow b$ and $c \leftrightarrow d\,.$ The action of the non-trivial 1-cycle $s\in\Sym_2$ induces the following transformation on the $\alpha\,,\beta\,,\gamma\,,\delta$ coordinates:
\begin{align*}
    s^{(ab)}: \quad \alpha \mapsto \alpha\,, \quad \beta \mapsto \beta\,, \quad \gamma \mapsto \gamma\,, \quad \delta \mapsto \delta\,, \\
    s^{(cd)}: \quad \alpha \mapsto \beta\,, \quad \beta \mapsto \alpha\,, \quad \gamma \mapsto \delta\,, \quad \delta \mapsto \gamma\,.
\end{align*}
Thus, unlike the simple polynomial case, there is a residual action as the natural coordinates $\alpha\,,\beta\,,\gamma\,,\delta$ are not fully invariant. Introducing variables
\[
w_1 = \alpha+\beta\,, \quad z_1 = \alpha - \beta\,, \quad w_2 = \gamma +\delta\,, \quad z_2 = \gamma - \delta\,,
\]
then these are all invariant under the action of $\Sym_2^{(ab)}$, whereas under the action of the simple transposition in $\Sym_2^{(cd)}$ it is:
\[
w_1 \mapsto w_1\,, \quad w_2 \mapsto w_2\,, \quad z_1 \mapsto -z_1\,, \quad z_2 \mapsto -z_2\,.
\]
We thus obtain the orbit space 
\[
\begin{array}{ccc}
\mathbb{C}^2 & \times & \mathbb{C}^2 \slash \Sym_2 \\
\vin & & \vin\\
(w_1,w_2) & & (z_1,z_2)
\end{array}
\]
and in terms of the original coordinates
\[
\begin{array}{cccccc}
w_1 & = & (c+d)-(a+b)\,, & w_2  & =& c+d \,,\\
w_2  & = & \displaystyle{\frac{ (c^2+d^2) - (c+d)(a+b) + 2 a b}{(c-d)}} \,, & z_2  & = & c-d\,.
\end{array}
\]
The ring of invariant functions on $\mathbb{C}^2 \slash \Sym_2$ - the simplest case of a cyclic singularity - is not freely generated:
\[
\mathbb{C}[z_1,z_2]^{\Sym_2} \cong \quotient{\CC[u_1,u_2,u_3] }{< u_1 u_2 - u_3^2 >}\,,
\]
where $u_1=z_1^2\,, u_2=z_2^2$ and $u_3=z_1 z_2\,.$ Note that the singular point is excluded as it lies in the image of the locus $c=d$ which has been excluded from the start. Thus even in the simplest case of a rational function with two poles one obtains a much richer structure when one applies Vi\`eta-type formulae.
\end{ex}

A different approach to the understanding the various actions is to note that the {\sl pair} $\{\alpha,\gamma\}$ is mapped under the actions of $\Sym_2^{(cd)}$ to the pair $\{\beta,\delta\}\,.$ Thus $\Sym_2^{(cd)}$ induces a {\sl diagonal} action on the coordinates $\{\alpha,\beta,\gamma,\delta\}\,.$ The invariants of this action - diagonal invariants - are, by the fundamental result of Weyl \cite{weyl}, the polarized power sums
\[
\mathscr{P}_{m,n} = \alpha^m \gamma^n + \beta^m \delta^n\,.
\]
In general these are rational in the $\{a,b,c,d\}$-variables but it is straightforward to see (\cref{prop:diaginvariantspoly}) that diagonal invariants $\mathscr{P}_{m=1,n}$ are polynomial functions which are invariant under the full group $\Sym_2^{(ab)} \times \Sym_2^{(cd)}$. The ring of such diagonal invariants is poorly understood - it is not freely generated for example. It is the polynomial diagonal invariants $\mathscr{P}_{m=1,n}$ that will play a central role in the construction of Frobenius manifolds given in \cref{ssec:DFmfldrational}.

\subsection{Relation to previous work}

All the Frobenius manifolds in this paper will be constructed via superpotentials. An essential component of this construction, in addition to a holomorphic map $\lambda:\mathbb{P}^1\rightarrow\mathbb{P}^1$, is the so-called Saito form, or primary differential, here denoted by $\omega\,.$ It is the properties of this form that guarantee that the metric defined by the residue formula (\ref{eq:HurwitzFrobeniusstructure}) is flat. There is no unique choice for this form, and different choices give Frobenius manifolds related to each other via Legendre transformations \cite{Dubrovin1996}. With $\omega=dw$ and
\begin{align*}
\lambda(w) &= w^{\ell+1} + \sigma_{\ell-1} w^{\ell-1} + \ldots + \sigma_0\,,\\
&= \eval{\prod_{a=1}^{\ell+1} (w-z_a) }_{\sum z_a=0}
\end{align*}
one obtains the well-known Frobenius structure on the Coxeter group orbit space $\mathbb{C}^\ell\,/\,W(\AN_\ell)\,$ \cite[Lecture 4]{Dubrovin1996}. With $\omega = -w^{-1} dw$ and
\begin{align*}
\lambda(w) &= w^{\ell-r+1} + \sigma_\ell  w^{\ell-r} + \ldots + \sigma_r + \frac{1}{w} \sigma_{r-1} + \ldots + \frac{1}{w^r} \sigma_0\,, \\
&= \frac{1}{w^r} \prod_{i=1}^{\ell+1} (w - e^{\phi_i})
\end{align*}
one obtains the Frobenius manifold on the extended-affine-Weyl orbit space $\mathbb{C}^{\ell+1}\,/\,\widetilde{W}^{(r)}(\AN_\ell)\,.$ In both these construction one also has a construction via a Chevalley-type theorem, which expresses the flat-coordinates for the Saito metric in terms of a finite set of freely generated invariant polynomial (or trigonometric polynomial) functions. Other examples of such Frobenius manifolds can e.g. be found in the works \cites{MSgenusone,ShrDouble,BertolaJacobiII,almeidaextendedJacobi}, where positive-genus cases are also taken into account.

In this paper we will study two generalisations of these example to spaces of rational functions with multiple simple zeros and poles, that is, to functions of the form:
\begin{align*}
\lambda(w) 
&= w^{\ell+1}+\sigma_{\ell-1}\,w^{\ell-1}+\dots+\sigma_{0}+\frac{\alpha_1}{w-\beta_1}+\dots+\frac{\alpha_{n_p}}{w-\beta_{n_p}}\,,\\
&= \frac{\prod_{a=1}^{n_z} (w-z_a) }{\prod_{\nu=1}^{n_p} (w-p_\nu)} \,, 
\end{align*}
and to functions of the form
\begin{align*}
\lambda(w) 
&= w^{\ell-r+1}+\sigma_{\ell}w^{\ell-r}+\dots +\sigma_{r}+\tfrac{1}{w}\sigma_{r-1}+\dots+\tfrac{1}{w^{r}}\sigma_{0}+\frac{\alpha_1}{w-\beta_1}+\dots+\frac{\alpha_{n_p}}{w-\beta_{n_p}}\,,\\
&= \frac{1}{w^r}\frac{\prod_{a=1}^{n_z} (w-z_a) }{\prod_{\nu=1}^{n_p} (w-p_\nu)} \,, 
\end{align*}
where we assume the number of zeros $(n_z)$ is greater than the number of poles $(n_p)$, and we denote $\ell+1=n_z-n_p\,.$ Fundamental in this is the action of the two natural symmetric groups $\Sym_{n_z}$ and $\Sym_{n_p}$, which permute the zeros and poles respectively.\\

In \cite{ma2023frobenius}, superpotentials of the form
\begin{align*}
    \lambda(w)&=\frac{1}{(w-p_1)^r \, w^m} \prod_{a=1}^{\ell+1} (w-z_a)
\end{align*}
were studied (generalising the earlier work \cite{zuo2020frobenius} for the $r=1$-case). A Chevalley-type theorem was constructed, enabling an orbit space description of the resulting Frobenius manifold in terms of a certain extended affine Weyl group. However, the superpotential only has the trivial action $\Sym_{n_p=1}$ on poles. The case $m=1$ was studied in \cite{Arsie_2022} (without using a superpotential). With multiple poles, the action of the group $\Sym_{n_p}$ comes into play, and the aim of this paper is to study this action.

\subsection{Summary of Results}

In this paper, the Frobenius manifold structure on Hurwitz spaces of generic rational functions is studied. This structure comes from suitably extending the standard superpotential description of the Frobenius manifolds on the orbit spaces of Weyl group of type $\AN$, and their extensions following Dubrovin and Zhang \cite{dubrovin_zhang_1998}. Thus, from the PDE-theoretical viewpoint, a family of solutions to the WDVV equations is constructed exhibiting a peculiar symmetrical behaviour under some specific action of the symmetric group on their variables. 
In particular, these solutions are constructed starting from the corresponding polynomial solution associated to some Weyl group $W(\AN_\ell)$ (or to the related Dubrovin-Zhang extended-affine-Weyl solution), the knowledge of the corresponding Saito polynomials (respectively, of the relation between the elementary symmetric polynomials and the Dubrovin-Zhang flat coordinates), and some generators of the ring of invariant polynomials with respect to the diagonal action of the symmetric group, as briefly discussed above.\\

The first step is the construction of the Saito-flat coordinates in terms of the flat-coordinates of the intersection form. As the above example shows, the dependence of the residue and pole variables (which in general turn out to be Saito-flat coordinates) on the flat coordinates of the intersection form are via complicated rational functions. However, up to simple $\log$-terms, these only enter the resulting prepotential via certain combinations knowns as diagonal invariants, and hence the prepotential is basically polynomial in the flat coordinates of the intersection form.\\

The main results of the paper are as follows:
\begin{itemize}
    \item On the Hurwitz space $\Hw_{0,\ell+n_p+1}(\ell,\bm{0})$ -- equipped with the primary form $\omega$ being the second-kind Abelian differential with a double pole at $\infty_0$ -- coordinates $(\bm{t},\bm{\alpha},\bm{\beta})$, as in \cref{lem:flatAntail}, are flat for the Frobenius pairing away from the discriminant locus. In such coordinate system, the Frobenius manifold structure is described by:
    \begin{align*}
         \Fr(\bm{t},\bm{\alpha},\bm{\beta})&=\Fr_{\AN_\ell}(\bm{t})+\tfrac12\sum_{\mu=1}^{n_p}\alpha_\mu^2\log\alpha_\mu+\sum_{\mu<\nu}\alpha_\mu\alpha_\nu\log(\beta_\mu-\beta_\nu)+\\
        &\quad+\tfrac{1}{\ell+2}\Theta_{\ell+2}(\bm{\alpha},\bm{\beta})+\tfrac{1}{\ell}\sigma_{\ell-1}(\bm{t})\,\Theta_\ell(\bm{\alpha},\bm{\beta})+\dots+\sigma_0(\bm{t})\,\Theta_1(\bm{\alpha},\bm{\beta})+f(\bm{t})\,\Theta_{0}(\bm{\alpha},\bm{\beta})\,,\\
        e&=\partial_{t_1}\,,\\
        E&= \sum_{k=1}^{\ell}\bigl(1-\tfrac{k-1}{\ell+1}\bigr)t_k\partial_{t_k}+\tfrac{\ell+2}{\ell+1}\bigl[\alpha_1\partial_{\alpha_1}+\dots+\alpha_{n_p}\partial_{\alpha_{n_p}}\bigr]+\tfrac{1}{\ell+1}\bigl(\beta_1\partial_{\beta_1}+\dots+\beta_{n_p}\partial_{\beta_{n_p}}\bigr)\,.
    \end{align*}
    In the expression for the prepotential, it is understood that $\Fr_{\AN_\ell}\in\CC[\bm{t}]$ is the polynomial solution to the WDVV equations associated to the Weyl group $W(\AN_\ell)$, $\bigl\{\Theta_\nu(\bm{x},\bm{y})\bigr\}_{\nu\in\Z_{\geq 0}}$
    are the polarised power sums that are linear in the first set of coordinates, as defined in \cref{prop:polarisedpowersums} and \cref{prop:diaginvariantspoly}, $(-1)^p\sigma_p$ is the degree-$(\ell+1-p)$ elementary symmetric polynomial, expressed as a polynomial function of the Saito polynomials $t_1,\dots, t_\ell$, while, finally, $f(\bm{t})$ is a solution to the system of PDEs given below in \cref{eq:PDEfAn}.
    
    The single simple-pole case $n_p=1$ corresponds to the structure associated to the Weyl group $W(\mathscr{B}_{\ell+2})$ found in \cite{Arsie_2022}, as already realised in \cite{MaZuoBn}.
    \item On the Hurwitz space $\Hw_{0,\ell+r+n_p+1}(\ell,r-1,\bm{0})$ -- equipped with the primary differential $\omega$ being the third-kind Abelian differential with poles at $\infty_0$ and $\infty_1$ and residues $+1$ and $-1$ respectively -- coordinates $(\bm{t},\bm{\alpha},\bm{\beta})$, as defined in \cref{lem:flatEAWdeform}, are flat for its Frobenius pairing, away from the discriminant. The Frobenius manifold structure is described, in such coordinate system, by:
    \begin{align*}
             \Fr(\bm{t},\bm{\alpha},\bm{\beta})&=\Fr_{\AN_{\ell+r}^{(r)}}(\bm{t})+\tfrac12\sum_{\mu=1}^{n_p}\alpha_\mu^2\log\alpha_\mu+\sum_{\mu<\nu}\alpha_\mu\alpha_\nu\log(e^{\beta_\mu}-e^{\beta_\nu})+\\
        &\quad+ \tfrac{1}{\ell+1}\widetilde{\Theta}_{\ell+1}(\bm{\alpha},\bm{\beta})+\tfrac{1}{\ell}\sigma_{\ell+r}(\bm{t})\,\widetilde{\Theta}_\ell(\bm{\alpha},\bm{\beta})+\dots+\sigma_{r+1}(\bm{t})\,\widetilde{\Theta}_1(\bm{\alpha},\bm{\beta})+\\
          &\quad+f(\bm{t})\,\widetilde{\Theta}_0(\bm{\alpha},\bm{\beta})+\tfrac{1}{-1}\sigma_{r-1}(\bm{t})\,\widetilde{\Theta}_{-1}(\bm{\alpha},\bm{\beta})+\dots+\tfrac{1}{-r}\sigma_0(\bm{t})\,\widetilde{\Theta}_{-r}(\bm{\alpha},\bm{\beta})+\\
          &\quad +\sigma_r(\bm{t})\,\Theta_1(\bm{\alpha},\bm{\beta})+\tfrac12\mathscr{P}_{2,1}(\bm{\alpha},\bm{\beta})\,,
          \end{align*}
    \begin{align*}
          e&=\partial_{t_{\ell+1}}\,,\\
          E&=\sum_{p=1}^{\ell+r+1}\tfrac{p}{\ell+1}t_p\partial_{t_p}+\tfrac{1}{\ell+1}\bigl[\partial_{\beta_1}+\dots+\partial_{\beta_{n_p }}\bigr]+\alpha_1\partial_{\alpha_1}+\dots+\alpha_{n_p }\partial_{\alpha_{n_p }}\,.
    \end{align*}
    Similarly, $\Fr_{\AN_{\ell+r}^{(r)}}$ is the prepotential on the space of orbits of the extended affine-Weyl group $\widetilde{W}^{(r)}(\AN_{\ell+r})$, as a function of the Dubrovin-Zhang flat coordinates $\bm{t}=(t_1,\dots, t_{\ell+r+1})$,
    \[
    \begin{aligned}
        \bigl\{\widetilde{\mathscr{P}}_{p,q}(\bm{x},\bm{y})\bigr\}_{(p,q)\in\Z_{\geq 0}\times\Z}\qquad \text{ and } \qquad  \bigl\{\mathscr{P}_{p,q}(\bm{x},\bm{y})\bigr\}_{(p,q)\in\Z_{\geq 0}^2}
    \end{aligned}
    \]
    are the exponential and ordinary polarised power sums respectively, as in \cref{prop:polarisedpowersums} and \cref{prop:exppowersums}, and 
   $\widetilde{\Theta}_k:=\widetilde{\mathscr{P}}_{1,k}$, $(-1)^p\sigma_p$ is the degree-$(\ell+r+1-p)$ elementary symmetric polynomial as a function of the Dubrovin-Zhang flat coordinates and, finally, $f(\bm{t})$ is a solution to the system of PDEs \ref{EAW:eq:systemf} -- with reference to the notation in \cref{lem:identityflatcoords}, we have here chosen $\star=\ell+1$.
   
    The single simple-pole case $n_p=1$ here corresponds, on the other hand, to the structure on the orbit spaces of the extended affine Weyl group $\widetilde{W}^{(r,r+1)}(\AN_{\ell+r})$ with two adjacent marked roots, introduced in \cite{zuo2020frobenius}.
\end{itemize}

In particular, in both cases considered in this paper, the explicit expression for the prepotential relies on knowledge of the solution associated to the Weyl and extended affine-Weyl groups of type-$\AN$ respectively. In the former case, explicit formulae are given in \cite{natanzonprepotentials}, even though the coefficients of each individual monomial in the flat coordinates consist of some somewhat involved combinatorial expressions. On the other hand, to our knowledge, no explicit formula for the solutions to the WDVV equations associated to the extended affine-Weyl group is present in the literature. Presumably, combinatorial formulae might also be found for these cases. Nevertheless, for fixed choices of $\ell$ and $r$, it is possible to work out the corresponding prepotential using e.g. mirror symmetry as in \cite{JEP_2022__9__907_0}. See also the examples in Section 5 of the second author's PhD thesis \cite{vGemstPhD}.

In the recent work \cite{rejeb2023new}, a formula for the prepotential of the Frobenius manifold structure on any Hurwitz space with any choice of admissible primary form is given in Theorem 4.1. Of course, our expression will represent a class of those general solutions. We will comment on it in explicit examples. Our work, however, despite only providing an answer in some of the cases, relies on a slightly different approach to determine the flat coordinates and prepotential. Namely, rather than directly using the formulae given in \cite[Theorem 5.1]{Dubrovin1996} which, although very general, do not relate the structure on different Hurwitz spaces, we consider embeddings $\Hw_{0,\ell+1}(\ell)\hookrightarrow\Hw_{0,\ell+n_p+1}(\ell,\bm{0})$ and, similarly, $\Hw_{0,\ell+r+1}(\ell,r-1)\hookrightarrow\Hw_{0,\ell+r+n_p+1}(\ell,r-1,\bm{0})$. Given that the structure on the lower-dimensional Hurwitz spaces are known for some choice of the primary form, and since the same primary form also determines a Frobenius manifold structure on the higher-dimensional ones, then we look for sets of flat coordinates that preserve such embeddings. In other words, if $\bm{t}$ denotes a set of flat coordinates of the Hurwitz spaces on the left-hand side of the embeddings, we look for flat coordinates on the right-hand side of the form $(\bm{t},\bm{s})$. These will not necessarily be given by the formulae in \cite{Dubrovin1996}. It, then, makes sense to relate the solution on the lower-dimensional space to the one on the higher-dimensional one. Furthermore, the dependence of the new solution on the additional set of flat coordinates $\bm{s}$ can be controlled by looking at the invariance properties of the coordinate map $(\bm{t},\bm{s})\mapsto\lambda(w,\bm{t},\bm{s})$, i.e. by looking at a (finite) group $G$ whose action on the $\bm{s}$-coordinates leaves the superpotential unchanged. The new solution will, then, necessarily depend on the $\bm{s}$-coordinates through invariant functions under such action. In particular, polynomial invariants are the easiest ones to study. Despite the necessity of the new solution to contain some non-polynomial terms -- as all polynomial solutions to the WDVV equations have been classified in \cite{hertling_2002} -- it turns out that this contribution is here easily controlled, and given by the logarithmic terms that appear above. Such universal structure of these kind of solutions was not spotted using the more general method in \cite{rejeb2023new}.
\begin{acknowledgements}
The authors thank the anonymous referee for useful comments and suggestions.
   A. P. would also like to thank EPSRC for financial support. 
\end{acknowledgements}

\section{Preliminaries}
\documentclass[main.tex]{subfile}

\subsection{Frobenius manifolds}
Following \citelist{\cite{Dubrovin1996}*{Lecture 1} \cite{hertling_2002}}, we define: 
\begin{defn}[(Dubrovin-)Frobenius manifold]
An $n$-dimensional Frobenius manifold of charge $d\in\CC$ is an $n$-dimensional complex manifold $M$ equipped with:
\begin{itemize}
    \item an $\Hol_M$-bilinear multiplication $\bullet:\X_M\otimes_{\Hol_M}\X_M\to \X_M$ on the sheaf of holomorphic tangent vectors, called \emph{Frobenius product}.
    \item a symmetric and non-degenerate bilinear form $\eta:\X_M\otimes_{\Hol_M}\X_M\to\Hol_M$, called \emph{Frobenius pairing} (or metric).
    \item two holomorphic vector fields $e$ and $E$, respectively called \emph{unity} and \emph{Euler vector field}.
\end{itemize}
such that:\begin{enumerate}
    \item [(DFM1)] At any $p\in M$, $T_pM$ equipped with the multiplication $\bullet_p$ and the non-degenerate, $\CC$-bilinear form $\eta_p$ is a Frobenius algebra.
    \item [(DFM2)] $e_p$ is the identity of $T_pM$ at any $p\in M$.
        \item [(DFM3)] for any $\hbar \in\CC$, the \emph{Dubrovin connection} ${}^\hbar\nabla:={}^\eta\nabla+\hbar\,\bullet$ is flat and torsion-free, ${}^\eta\nabla$ being the Levi-Civita connection of $\eta$.
    \item [(DFM4)] $e$ is covariantly constant with respect to $\eta$: ${}^\eta\nabla e=0$.
    \item [(DFM5)] $E$ satisfies the following relations:
    \[
    \begin{aligned}
\Lie_Ee\,&=\,-e \,,& && \Lie_E\bullet\,&=\,\bullet\, ,& && \Lie_E\eta\,&=\, (2-d)\,\eta\,,
    \end{aligned}
    \]
    $\Lie$ being the Lie derivative.
\end{enumerate}
\end{defn}
\begin{oss}
    The definition works also works in the category of smooth manifolds. One just has to replace “holomorphic” with “smooth”.
\end{oss}
\begin{oss} To unpack the third condition (DFM3), we firstly recall that $\bullet$ is equivalent to a holomorphic section of $TM\otimes T^*M^{\otimes 2}$. Using the isomorphism $TM\cong T^*M$ induced by $\eta$, one can then “lower” the first index to get a holomorphic section of $T^*M^{\otimes 3}$, which we shall denote by $c$, so $c(X,Y,Z)=\eta(X,Y\bullet Z)$.

A simple calculation shows that, for any $X,Y\in\X_M$ the torsion of the Dubrovin connection becomes
\begin{align*}
      {}^\hbar T(X,Y)&= {}^\eta T(X,Y)+\hbar (X\bullet Y-Y\bullet X)\,.
\end{align*}
Hence, ${}^\hbar\nabla$ is torsion-free identically in $\hbar$ if and only if ${}^\eta\nabla$ is torsion-free and $\bullet$ is commutative. Notice that this, along with the compatibility of $\eta$ and $\bullet$, implies that $c$ is a section of the third symmetric power of the cotangent bundle. 

Similarly, using the compatibility of $\eta$ and $c$, one obtains the following expression for the Riemann tensor of ${}^\hbar\nabla$:
\begin{align*}
       \inp{W}{{}^\hbar R(X,Y)Z}&= \inp{W}{{}^\eta R(X,Y)Z}+\hbar\bigl[({}^\eta\nabla_X c)(W,Y,Z)-({}^\eta\nabla_Y c)(W,X,Z)\bigr]+\\&\quad\quad +\hbar^2\inp{W}{X\bullet(Y\bullet Z)-(X\bullet Y)\bullet Z}\,.
\end{align*}
This is again required to vanish identically in $X,Y,Z,W\in\X_M$ and $\hbar\in\CC$, where $\inp{\vdot}{\vdot}:=\eta(\vdot,\vdot)$. 

Hence, the condition (DFM3) is equivalent to:
\begin{enumerate}
        \item [(DFM3a)] ${}^\eta\nabla$ is itself flat (and torsion-free).
        \item [(DFM3b)] $\bullet$ is associative and commutative.
        \item [(DFM3c)] $c$ is a Codazzi tensor (i.e. its covariant derivative ${}^\eta\nabla c$ is totally symmetric).
    \end{enumerate}
\end{oss}

\begin{lem}[Flow of the Euler vector field]\label{lem:EVFaffine}
    Let $M$ be a Frobenius manifold. The Euler vector field $E$ is affine with respect to $\eta$, i.e. ${}^\eta\nabla^2 E=0$.
\end{lem}
\begin{proof}
If $\xi $ is a conformal Killing vector field for the metric $\eta$, i.e. if $\Lie_\xi\eta=f\,\eta$ for some $f\in\Hol_M$, then a simple generalisation of the second Killing identity (also known as Kostant formula) shows that, for any two $X,Y\in\X_M$:
\[
({}^\eta\nabla_X{}^\eta\nabla\xi)(Y)+{}^\eta R(X,\xi)Y=\tfrac12\bigl(Y(f)\,X+X(f)\,Y-\eta(X,Y)\,{}^\eta\nabla f\bigr)\,.
\]

\noindent Now, the Euler vector field is a conformal Killing vector field for the Frobenius pairing -- whose Riemann tensor vanishes -- with constant conformal factor $f=2-d$. This proves the statement.
\end{proof}
    
\begin{oss}[Prepotential and WDVV equations] Since $\eta$ is a flat metric (i.e. its Levi-Civita connection is flat), we know from standard differential geometry that one can find coordinates $(t_1,\dots, t_{n})$ such that $\eta$ is represented by a constant matrix in such coordinate system -- and, as a consequence, the corresponding Christoffel symbols will vanish. These are called \emph{flat coordinates} for $\eta$. It is clear that any affine transition function will produce a new system of flat coordinates from a given one.
Therefore, using (DFM4), we can, without loss of generality, redefine our system of flat coordinates so that $e=\partial_{t_1}$.

Moreover, the commutativity of $\bullet$ and (DFM3c) imply that, in any system of flat coordinates, the components $c_{pqr}$ of $c$ are actually the third derivatives of a $\CC$-valued, local holomorphic function $\Fr$, called \emph{prepotential} (or free energy) of the Frobenius manifold structure. Clearly, such a function is only determined up to any at-most-quadratic polynomial in the $\bm{t}$-coordinates (i.e. any function whose third derivatives vanish identically). Associativity of $\bullet$ further constraints $\Fr$ to be a solution of the Witten-Dijkgraaf-Verlinde-Verlinde (WDVV for short) equations \cite{Dubrovin1996}:
\begin{align}
  \sum_{a,b}(\eta^\sharp)_{ab}\,\Fr_{pqa}\,\Fr_{brs}= \sum_{a,b}(\eta^\sharp)_{ab}\,\Fr_{pra}\,\Fr_{bqs}\,,
\end{align}
$\eta^\sharp$ now being the inverse of the matrix $[e(\Fr_{ab})]_{a,b=1,\dots,n}$.

Finally, as far as the Euler vector field goes, it follows from \cref{lem:EVFaffine} that its components in a system of flat coordinates must be at most linear functions of the coordinates. Moreover, (DFM5) implies that $\Lie_E\Fr=(3-d)\Fr$. Hence, by  a straightforward generalisation of Euler's Theorem for homogeneous functions, it follows that $\Fr$ must be quasi-homogeneous, up to quadratic polynomials in the coordinates $t_1,\dots, t_n$.

Conversely, given a quasi-homogeneous solution to the WDVV equations -- i.e. a solution such that, for any $c\in\CC^*$, $\Fr(c^{d_1}t_1,\dots, c^{d_n}t_n)=c^{d_F} \Fr(t_1,\dots, t_n)$ for some $d_1,\dots, d_n,d_F\in\CC$ -- one can (locally) obtain a Frobenius manifold structure by reversing the procedure we have just described. 
There is, therefore, a one-to-one correspondence as follows:
\[
\begin{aligned}
    \frac{\left\{
   \parbox{4.5cm}{\centering Quasi-homogeneous solutions to the WDVV equations}
\right\}}{\left\{\parbox{3cm}{\centering At-most-quadratic polynomials} \right\}}&&&&\longleftrightarrow &&&& \left\{\parbox{3.5cm}{\centering Frobenius manifold structures}\right\}\,.
\end{aligned}
\]
\end{oss}
\begin{oss}
    A Frobenius manifold $M$ is called \emph{semi-simple} if, for all generic $p\in M$, the Frobenius algebra $T_pM$ is semi-simple.
\end{oss}

\subsection{Almost-duality}
As discussed in \cite{Dub04}, a Frobenius manifold structure actually “almost” determines another Frobenius manifold structure on (an open subset of) the same manifold. To describe this, let us first point out that the metric-induced isomorphism $\sharp:T^*M\tilde{\longrightarrow} TM$ can be used to define a multiplication on the cotangent bundle (by multiplication of the metric-dual vector fields); with a slight abuse of notation, we shall again denote this by $\bullet$. This is justified by the fact that, by construction, $\sharp$ is an isomorphism of vector bundles which fibre-wise preserves the algebra structure.

\begin{defn}[Intersection form]
    The intersection form of the Frobenius manifold $M$ is the cometric\footnote{By \emph{cometric}, we mean a metric on the cotangent bundle of $M$.} $g^{*}$ defined by:
    \begin{align}
            g^{*}(\alpha,\beta)&:= E\intprod(\alpha\bullet\beta)\,, \qquad  \forall \alpha,\beta\in \Omega^1_M\,.
    \end{align}
\end{defn}
\begin{oss}
    Of course, any non-degenerate metric on $TM$ induces a non-degenerate metric on $T^*M$ (and vice versa).
\end{oss}
Since the following relation holds: $g^{*}(\alpha,\beta)=\eta(E\bullet \alpha^\sharp,\beta^\sharp)$, it follows that $g^{*}$ is non-degenerate wherever $E$ is invertible, i.e. where there exists an holomorphic vector field $E^{-1}$ such that $E\bullet E^{-1}=e$.
\begin{defn}[Discriminant locus]
    The discriminant locus $\Delta_M$ of a Frobenius manifold $M$ is the set:
    \begin{align*}
            \Delta_M&:= \bigl\{p\in M:\quad E_p\text{ is not invertible}\bigr\}\,\equiv\,\bigl\{p\in M:\quad g^*_p\text{ is degenerate}\bigr\}\,.
    \end{align*}
\end{defn}

\noindent It is not difficult to see that the discriminant locus is a closed subset of $M$. On $M\smallsetminus\Delta_M$, the metric $g$ on $TM$ induced by $g^*$ will therefore be given by:
\begin{align}
        g(X,Y)&= \eta(E^{-1}\bullet X,Y)\,.
\end{align}

\begin{lem}[Almost duality, \cite{Dub04}]
    Let $M$ be a Frobenius manifold with intersection form $g$ on $M\smallsetminus\Delta_M$. Let ${}^g\nabla$ denote the Levi-Civita connection of $g$ and let $\star$ be the multiplication on $T(M\smallsetminus\Delta_M)$
    given by:
    \begin{align}
                X\star Y&:= E^{-1}\bullet X\bullet Y\,.
    \end{align}
    Then:
    \begin{enumerate}
            \item $g$ and $\star$ are compatible.
                \item $E$ is the identity of $\star$.
        \item For any $\hbar\in\CC$, the connection ${}^\hbar\nabla:={}^g\nabla+\hbar\,\star$ is flat and torsion-free.
    \end{enumerate}
\end{lem}

This structure is \emph{almost-dual} to the original one, as it is not quite a fully-fledged Frobenius manifold. In fact, in general, the unit $E$ will not be covariantly constant with respect to ${}^g\nabla$, and there is no choice of an Euler vector field either. By relaxing these two requirements, one gets a structure that goes by the name of \emph{almost-flat F-manifold} \cite{Arsie_2022}.

Nevertheless, $g$ is still a flat metric, therefore we can find a set of flat coordinates $z_1,\dots, z_n$ in a neighbourhood of each point in the complement $M\smallsetminus\Delta_M$. Furthermore, the flatness of the connections ${}^g\nabla+\hbar\,\star$ ensures that, locally, the structure is again specified by a prepotential $\overset{\star}{\Fr}$ solving the WDVV equations (in the flat coordinates of $g$). However, the lack of an Euler vector field will now generically lift the quasi-homogeneity of this solution.
\begin{oss}
    Notice that the $g$-dual of $\star$ and the $\eta$-dual of $\bullet$ -- with respect to the isomorphisms $TM\cong T^*M$ induced by the corresponding metric -- coincide as multiplications on the cotangent bundle.
\end{oss}

\subsection{Mirror symmetry for Frobenius manifolds}
In the literature, there are traditionally three kinds of spaces that carry a natural Frobenius manifold structure:
\begin{enumerate}
    \item [(A)] The big quantum cohomology ring $QH^\bullet(X)$ of a complex projective variety $X$ (this construction will not play a role in this paper).
    \item [(B)] Hurwitz spaces $\Hw_{g,L}(\bm{n})$, i.e. the set of equivalence classes of meromorphic functions on some genus-$g$ closed Riemann surface $\mathscr{C}_g$ with prescribed ramification at $\infty$. The equivalence relation is here given by factorisation through automorphisms of $\mathscr{C}_g$, i.e. two meromorphic functions $\lambda,\mu:\mathscr{C}_g\to \RS$ are equivalent if there is some $f\in\Aut{\mathscr{C}_g}$ such that $\lambda=\mu\circ f$.
    \item [(C)] Orbit spaces of (extensions of) reflection groups. For this construction -- which is traditionally named after Saito -- we are just going to mention here that the main technical point is that, in order to have a complex manifold structure on the space of orbits, one need to look at algebraic relations in its invariant ring. We will therefore sometimes shorthandedly employ the name \emph{Saito Frobenius manifolds} for Frobenius manifolds constructed in this fashion.
\end{enumerate}
A \emph{mirror symmetry} is, then, any (local) Frobenius manifold isomorphism between any two of these constructions, i.e. a one-to-one correspondence between a Frobenius manifold structure constructed in one of the previous ways and another one. The map realising the equivalence is referred to as \emph{mirror map} \cite{JEP_2022__9__907_0}.

\begin{oss}
    Suppose that there exists a local Frobenius manifold isomorphism between the C-model\footnote{The notion of a \emph{C-model} is not standard, as in ordinary mirror symmetry one only has two models to compare, but we shall here employ it to mean the orbit space construction.} $M_{\mathrm{orb}}$ and the B-model $M_{\mathrm{Hurw}}$, and let $V $ be the vector space upon which the reflection group $W$ acts. It follows that, for any $v\in V$ and $g\in W$, the two meromorphic functions $\lambda_v$ and $\lambda_{g.v}$ respectively associated to the points $v$ and $g.v$ by the local isomorphism must be the same point in the Hurwitz space, i.e. they must be equivalent in the sense we have just discussed.
\end{oss}

Notice, finally, that the prepotential of any Saito-Frobenius manifold structure on the orbit space of some reflection group $W$ will necessarily be $W$-invariant up to at-most-quadratic terms, as its third derivatives need to be well-defined on the orbit space.


\subsection{Hurwitz Frobenius manifolds}
\documentclass[main.tex]{subfile}
Let $g\in\Z_{\geq 0}$ and $\bm{n}\in \Z_{\geq 0}^{\ell+1}$. As already mentioned upon, the Hurwitz space $\Hw_{g,L}(\bm{n})$ is the set of equivalence classes of degree-$L$ meromorphic functions on $\mathscr{C}_g$ having poles $\infty_0,\dots, \infty_\ell\in\mathscr{C}_g$ of order $n_0+1,\dots, n_\ell+1$ respectively, and only simple finite ramification points. For the sake of brevity, we shall call Hurwitz spaces with $g=0$ \emph{genus-0 Hurwitz spaces}. We shall further say that two equivalent meromorphic functions in a genus-0 Hurwitz space are \emph{projective equivalent}.
\begin{oss}
It is straightforward to prove that the degree of the maps is actually fixed by the choices of $\ell$ and $\bm{n}$ as follows:
    \[
    L\,\,=\,\,\ell+1+n_0+\dots+n_{\ell}\,.
    \]
    Furthermore, it follows from the Riemann-Hurwitz formula that $\infty$ cannot be the only ramification point of such a function. In particular, there will be precisely $d_{g,\bm{n}}:=2g+2\ell+n_0+\dots+n_\ell$ finite simple ramification points. This is also the dimension of $\Hw_{g,L}(\bm{n})$ as a complex, quasi-projective variety.
\end{oss}

A choice of a \emph{compatible} Abelian differential $\omega$ on $\mathscr{C}_g$ -- in the sense of \cite{Dubrovin1996}*{Lecture 5} -- determines a semi-simple Frobenius manifold on a cover $\widehat{\Hw}_{g,L}(\bm{n})$ of $\Hw_{g,L}(\bm{n})$ whose fibre over $\lambda\in \Hw_{g,L}(\bm{n})$ consists of the Torelli markings of $\mathscr{C}_g$. We shall denote the Frobenius manifold structure on such a space by $\Hw_{g,L}^\omega(\bm{n})$, and, with a slight abuse of notation, we are going to say that it comes from \emph{equipping} the Hurwitz space with a primary differential -- even though the latter is defined on the underlying compact Riemann surface rather than on the Hurwitz space itself. Notice that, in the genus-0 case, which we are ultimately interested in, the cover is trivial.\\

The Frobenius manifold structure is determined as follows. Firstly, we note that it follows from Riemann existence Theorem that the critical values $\{u_1,\dots, u_{d_{g,\bm{n}}}\}$ of $\lambda$ make up a coordinate chart in $\Hw_{g,L}(\bm{n})$ away from the closed subset where any two of them coincide. We then require that the holomorphic vector fields $\{\partial_{u_\alpha}\}_{\alpha=1}^{d_{g,\bm{n}}}$ be idempotents of the algebra, i.e.:
\[
\partial_{u_\alpha} \bullet \, \partial_{u_\beta} = \delta_{\alpha\beta} \partial_{u_\alpha}\, \qquad ({\rm no~sum})\,.
\]
The identity and Euler vector field will therefore be given by:
\begin{equation}
    \begin{aligned}
        e&\,:=\,\sum_{\alpha=1}^{d_{g,\bm{n}}}\partial_{u_\alpha} ,& && E&\,:=\,\sum_{\alpha=1}^{d_{g,\bm{n}}} u_\alpha\partial_{u_\alpha.}
    \end{aligned}
\end{equation}
The Frobenius metric and multiplication are determined by the choice of $\omega$ as follows: if $\Gamma_\lambda$ denotes the set of critical points of $\lambda$, then we set:
\begin{equation}\label{eq:HurwitzFrobeniusstructure}
    \begin{aligned}
        \eta(X,Y)&\,:=\,\sum_{x\in\Gamma_\lambda}\Res_x\bigl\{X(\lambda)Y(\lambda)\tfrac{\omega^2}{\dd\lambda}\bigr\}\, , &&& c(X,Y,Z)&\,:=\,\sum_{x\in\Gamma_\lambda}\Res_x\bigl\{X(\lambda)Y(\lambda)Z(\lambda)\tfrac{\omega^2}{\dd\lambda}\bigr\}\,,
    \end{aligned}
\end{equation}
for any triple of holomorphic vector fields $X,Y,Z$ on $\Hw_{g,L}(\bm{n})$. One can then show that the intersection pairing is given by:
\begin{equation}\label{eq:intformHurwitz}
    g(X,Y)=\sum_{x\in \Gamma_\lambda}\Res_{x}\bigl\{X(\log\lambda)Y(\log\lambda)\tfrac{\omega^2}{\dd\log\lambda}\bigr\}\,.
\end{equation}

It is customary to call a generic parametrisation of the points in the Hurwitz space by a family of meromorphic functions $\lambda$ a \emph{Landau-Ginzburg superpotential} for the Frobenius manifold, and $\omega$ its \emph{primary differential} (or primary form).\\

As an example, we consider genus-0 Hurwitz spaces $\Hw_{0,L}(\bm{n})$, which contain equivalence classes of meromorphic functions on the Riemann sphere, i.e. rational functions.
\begin{lem}[Choice of a representative in each equivalence class]\label{lem:reprequivclass}
     Consider a meromorphic function $\lambda$ on $\RS$. $\lambda$ is projective equivalent to a meromorphic function $\widetilde{\lambda}$ on $\RS$ such that:
    \begin{enumerate}
        \item $\widetilde{\lambda}$ has a pole at $\infty\,$;
        \item the polynomial quotient in the division of the numerator of $\widetilde{\lambda}$ by its denominator is monic;
        \item the weighted sum of the zeros of $\widetilde{\lambda}$ is equal to the weighted sum of its finite poles -- the weights being the multiplicities and orders respectively.
    \end{enumerate}
\end{lem}
\begin{proof}
If $\infty$ is not a pole of $\lambda$, then fix any M\"{o}bius transformation $f$ mapping a finite pole $p\in\CC$ to $\infty$ to have $(\lambda\circ f)(\infty)=\infty$. We can, then, uniquely write $\lambda\circ f$ as:
    \[
    (\lambda\circ f)(w)=\frac{P(w)}{Q(w)}\equiv H(w)+R(w)\,,
    \]
    for some polynomials $P,Q$ with $\deg P>\deg Q$, so that $H$ is the polynomial quotient of $P$ by $Q$ and $R$ is the rational reminder. Let $k:=\deg P-\deg Q\equiv\deg H\,$.
    
    Our second condition is that the leading-order coefficient $c\neq 0$ of $H$ be one. Therefore, it is enough to pick a $k$-th root of $c$ and any M\"{o}bius transformation $g$ fixing $\infty$ and such that $g(1)=c^{-1/k}$ to ensure that the polynomial part of $\lambda\circ f\circ g$ is monic.

    Finally, we notice that the weighted sum of the zeros and poles is the coefficient $s$ of the sub-leading term of $H$, up to an overall sign. Hence, it is enough to pick a parabolic M\"{o}bius transformation $h$ fixing $\infty$ and such that $h(0)=-\tfrac1k s$ to have that $\lambda\circ f\circ g\circ h$ satisfies all the required properties.
\end{proof}
\begin{oss}
Notice that there is some ambiguity in how one gets from $\lambda$ to $\widetilde{\lambda}$ in the proof of the previous Lemma. Namely, while M\"{o}bius transformations are indeed uniquely specified by the images of three points, there is no prescription on which one of the poles should be mapped to $\infty$.
In what follows, however, we will be considering genus-0 Hurwitz spaces of meromorphic functions having all simple poles except for one or two, which we can, without loss of generality, pick to be $\infty_0$ and $\infty_1$. In the former case, we shall always choose to place the non-simple pole at $\infty_0=\infty$. In the latter, on the other hand, it is going to be more convenient to pick the non-simple poles to be $\infty_0=\infty$ and $\infty_1=0$. Clearly, one could prove a slight modification of the previous Lemma with the third condition replaced by having a pole at $0$. With these prescriptions, we remove the residual ambiguity and we find that any point in any such orbit space can be uniquely represented by a meromorphic function satisfying the three conditions above (or the aforementioned slight modification). This will provide an explicit coordinate system on some open subset of the Hurwitz space.

Finally, we point out that, whenever we say that the primary differential has, for instance, a pole at $\infty$, we actually mean that its pole is at $\infty_0$ (and similarly for $0$ and $\infty_1$).
\end{oss}

\subsection{Saito Frobenius manifolds}\label{subsec:Saito}
\documentclass[main.tex]{subfile}

Here, $W$ will denote a finite subgroup of the orthogonal group of the $n$-dimensional Euclidean space $(\R^n,\inp{-}{-})$ generated by reflections, i.e. a Coxeter group. Clearly, this acts naturally on the polynomial ring $\R[x_1,\dots, x_n]$. The following result is classical:
\begin{thm}[Chevalley, \cite{humphreys_1990}]
    The ring of $W$-invariant polynomials in $\R[x_1,\dots, x_n]$ is generated, as a real algebra, by $n$ homogeneous, algebraically independent elements of positive degree.
    
    Furthermore, the degrees $d_1,\dots, d_n$ of the homogeneous, algebraically independent generators of the ring of invariants -- which we shall call \emph{basic invariants} -- are unique up to reordering.
\end{thm}
\begin{oss}
    Clearly, since the transformations in $W$ are orthogonal, then there is always an invariant polynomial of degree two -- namely $x_1^2+\dots+x_n^2$. This is also the lowest-degree invariant. The highest degree is called \emph{Coxeter number} of $W$, and denoted by $h$.

    In the remainder, we shall always order the degrees as follows:
    \[
    h=d_1>d_2\geq \dots\geq d_{n-1}>d_n=2\,.
    \]
\end{oss}
Evidently, if we now consider the complexification $\CC^n$ equipped with the symmetric, non-degenerate, $\CC$-bilinear form $\inp{-}{-}$ and the induced orthogonal action of $W$, the fact that the invariant ring is freely generated implies that the orbit space $\CC^n/W$ is a complex affine variety of dimension $n$, whose coordinate ring is the complexification $\R[x_1,\dots, x_n]^W\otimes_{\R}\CC\cong\CC[x_1,\dots, x_n]^W$ of the ring of invariants.

The natural surjection $\CC^n\twoheadrightarrow \CC^n/W$ is a local isomorphism away from the reflection hyperplanes -- which form a closed subset of $\CC^n$. Therefore, orthonormal coordinates $x_1,\dots, x_n$ on $\CC^n$ serve as local coordinates on the orbit space. It follows that the induced metric $\inp{-}{-}^*$ on the dual space descends to a cometric $g^*$. It is not difficult to see that its components with respect to any ordered system of generators for the coordinate ring are given by:
\begin{align*}
    \inp{\dd u_a}{\dd u_b}^*&= \tfrac{\partial u_a}{\partial x_1}\tfrac{\partial u_b}{\partial x_1}+\dots+\tfrac{\partial u_a}{\partial x_n}\tfrac{\partial u_b}{\partial x_n}\,, \qquad a,b=1,\dots, n\,.
\end{align*}

In particular, the determinant of $g^*$ is the square of the Jacobian $J(u_1,\dots, u_n)$, which vanishes on the image of the reflection hyperplanes \cite{humphreys_1990}*{Section 3.13}. Hence, $g^*$ defines a metric on the complement of the image of the irregular orbits.

\begin{thm}[\cite{Dubrovin1996}*{Theorem 4.1}]
There exists a unique Frobenius manifold structure on the complexified orbit space $\CC^n/W$ of $W$ such that:
\begin{itemize}
    \item $g^*$ is its intersection form.
    \item Its unity vector field is $e:=\pdv{u_1}$.
    \item Its Euler vector field is:
    \begin{align*}
    E&:= \tfrac1h \bigl[d_1u_1\tfrac{\partial}{\partial u_1}+\dots+d_nu_n\tfrac{\partial}{\partial u_n}\bigr]\,.        
    \end{align*}
\end{itemize}
Such structure is, in particular, semi-simple, its Frobenius pairing is given by the metric $\eta$ on the tangent bundle determined by the Saito cometric $\eta^*:=\Lie_e g^*$ on the orbit space, and the charge is $d=1-\tfrac2h$.    
\end{thm}
\begin{oss}
    The discriminant of such a Frobenius manifold is, as discussed, the image of the reflection hyperplane via the quotient map.
\end{oss}
\begin{oss}
Flat coordinates for the Frobenius pairing in any of these Frobenius manifolds are given by (some suitable choice of) a generating set for the coordinate ring \cite{Dubrovin1996}*{Corollary 4.3}.
\end{oss}
\begin{oss}
    As for the prepotential, a famous conjecture by Dubrovin, then proven by Hertling \cite{hertling_2002}, states that all the polynomial solutions to the WDVV equations correspond to Saito Frobenius manifold structures on orbit spaces of Coxeter groups.
\end{oss}

Now, in the case of the Weyl group $W(\AN_\ell)$ -- and actually for any Weyl group associated to simply-laced Dynkin diagrams -- there is a natural connection between Hurwitz spaces and the spaces of orbits of such groups. As a matter of fact, let $\ell\in\Z_{\geq 0}$ and consider the Hurwitz space $\Hw_{0,\ell+1}(\ell)$. Then, according to \cref{lem:reprequivclass}, in each equivalence class, we can find a representative of the form:
\begin{align}
    \lambda_\ell(w)&= w^{\ell+1}+\sigma_{\ell-1}w^{\ell-1}+\dots+\sigma_1w+\sigma_0\,,\label{eq:lambdaell}
\end{align}
for some $\sigma_0,\dots, \sigma_{\ell-1}\in\CC$. This is of course the same as the space of miniversal unfoldings of the $\AN_\ell$-singularity. 

Moreover, one can prove \cite{Dubrovin1996} that not only is this correspondence realised at the complex-manifold level, but the Saito Frobenius manifold structure on the orbit space $\CC^\ell/W(\AN_\ell)$ is isomorphic to the one on $\Hw_{0,\ell+1}(\ell)$, with the primary differential $\omega$ being the second-kind Abelian differential with a double pole at $\infty$.\footnote{Since $\deg(\omega)=-\chi(\RS)=-2$ \cite{donaldsonRiemannsurfaces}, this means that $\omega$ has no zeros.} That is to say, $\omega=\dd w$ in the above coordinate system. This is a primary differential of type 1 in Dubrovin's classification in \cite{Dubrovin1996}*{Lecture 5}. \\

In the work \cite{dubrovin_zhang_1998}, the authors constructed a generalised Saito construction for some extension of the affine-Weyl groups $\widetilde{W}^{(r)}$ corresponding to the choice of a marked simple root $\alpha_{r}$ in the Dynkin diagram\footnote{There is actually some restriction on the root one is allowed to mark, but, since we are actually only going to consider the case $W=W(\AN_\ell)$, any simple root does the job.}. In particular, if $W$ is a Weyl group, the affine-Weyl group $\widetilde{W}$ is the semi-direct product $W\ltimes \Lambda^\vee$, $\Lambda^\vee$ being the lattice of coroots generated by $\alpha^\vee_i:=\tfrac{2}{\norm{\alpha_i}^2}\alpha_i$, for $i=1,\dots, \rank W$. $W$ acts on $\Lambda^\vee$ by restriction of the natural action on the real Cartan subalgebra $\mathfrak{h}_{\R}$ of the complex Lie algebra associated to $W$. Finally, the extended affine-Weyl (EAW for short) group $\widetilde{W}^{(r)}$ is the semi-direct product $\widetilde{W}\rtimes\Z$ of $\widetilde{W}$ with the infinite-cyclic group $\Z$, whose defining action on $\mathfrak{h}_{\R}\oplus\R$ is given by:
\begin{align}
    (w,\alpha^\vee,n).(x,y)&:= \bigl(w(x)+\alpha^\vee+n\,\omega_{r}, y-n\bigr)\,,\label{eq:DZEAWaction}
\end{align}
where $\omega_{r}$ is the $r^{\mathrm{th}}$ fundamental weight of $W$: $\inp{\omega_{r}}{\alpha^\vee_s}=\delta_{rs}$. For future reference, the fundamental weights for the $\AN_\ell$ root system are given, in the simple root basis, by:
\begin{align}
    \omega_k&= \bigl(1-\tfrac{k}{\ell+1}\bigr)\bigl[\alpha_1+2\alpha_2+\dots+k\alpha_k\bigr]+\tfrac{k}{\ell+1}\bigl[(\ell-k)\alpha_{k+1}+(\ell-k-1)\alpha_{k+2}+\dots+\alpha_\ell\bigr]\,. \label{eq:fundweight}
\end{align}

One, then, considers the complexification of the representation space -- which is isomorphic to $\CC^{\rank W+1}$ -- and proves a Chevalley-like Theorem for a subring of the $\widetilde{W}^{(r)}$-invariant polynomial ring, thus leading to a Frobenius manifold structure on a sector of the orbit space $\CC^{\rank W+1}/\widetilde{W}^{(r)}$.\\ A key point in the construction is also fixing the metric on the extended represention space $\mathfrak{h}_{\R}\oplus \R$ so that the Frobenius pairing one gets at the end is actually flat. Flat coordinates will still be given by some choice of algebraically independent generators of the subring we are considering, as described in \cite{dubrovin_zhang_1998}*{Corollary 2.5}.

As for the prepotential, this will now still be a polynomial in the flat coordinates $\bm{t}$, except that it now also depends exponentially on one distinguished coordinate \cite{dubrovin_zhang_1998}*{Lemma 2.6}, denoted by $t_\bullet$: $\Fr\in\CC[\bm{t},e^{t_\bullet}]$. This of course is a reflection of the fact that we have made a choice of a root in the diagram.\\

Again in \cite{dubrovin_zhang_1998}*{Section 3}, the authors provide a B-model for the EAW orbit space construction in the cases $W=W(\AN_{\ell})$. In particular, the B-model for the  Frobenius manifold on the orbit space of $\widetilde{W}^{(r)}(\AN_\ell)$ is the Hurwitz space $\Hw_{0,\ell+1}(\ell-r,r-1)$ with primary differential being the third-kind Abelian differential with simple poles at $0$ and at $\infty$ and residues $-1$ and $+1$ respectively\footnote{As anticipated, we actually mean that the primary differential has simple poles at the non-simple poles of the meromorphic function, and we are picking a representative so that these two points coincide with $0$ and $\infty$. }. This is a primary differential of the third type according to Dubrovin's classification in \cite{Dubrovin1996}*{Lecture 5}. Any equivalence class in the Hurwitz space may now be represented as the following LG superpotential:
\begin{align}
    \lambda_{\ell,r}(w)&:= w^{\ell-r+1}+\sigma_{\ell}w^{\ell-r}+\dots +\sigma_{r}+\tfrac{1}{w}\sigma_{r-1}+\dots+\tfrac{1}{w^{r}}\sigma_{0}\,, \label{eq:EAWsuperpotential}
\end{align}
the primary differential being, in such coordinate system, $\omega=-\tfrac{1}{w}\dd{w}$.\footnote{Actually, the LG superpotential provided in \cite{dubrovin_zhang_1998} has a pole of order $r$ at $\infty$ and a pole of order $\ell-r+1$ at $0$. Of course, the two are projective equivalent, as they factor through the inversion map, which swaps the poles and residues of the primary differential. In the orbit space, this corresponds to a reflection about the centre of the Dynkin diagram. It was already noted in the original article that any symmetry $T$ of the Dynkin diagram yields an isomorphism between the Frobenius manifold structures associated to the marked nodes $\alpha$ and $T(\alpha)$.}

In particular, we point out the following result, which will be useful later on:
\begin{lem}
    Consider the Frobenius manifold structure on the Hurwitz space $\Hw_{0,\ell+1}(\ell-r,r-1)$, with primary differential being the third-kind Abelian differential with simple poles at $0$ and at $\infty$, and residues $-1$ and $+1$ respectively.\\
    One can choose flat coordinates $\bm{t}$ on it so that the distinguished coordinate $t_\bullet$ is chosen in such a way that $\sigma_0$ is $e^{rt_\bullet+i(\ell+1)\pi}$.
\end{lem}
\begin{proof}
    According to \cite{dubrovin_zhang_1998}*{Lemma 3.1}, such a structure is isomorphic to the one on the orbit space $\CC^{\ell+1}/\widetilde{W}^{(r)}(\AN_{\ell})$, the isomorphism coming from relating coordinates $\{\varphi_1,\dots,\varphi_{\ell+1}\}$ on the Hurwitz space described by the following factorisation of the superpotential:
    \[
    \lambda_{\ell,r}(w)=\tfrac{1}{w^r}(w-e^{\varphi_1})\dots (w-e^{\varphi_{\ell+1}})\,,
    \]
    with coordinates $\{x_1,\dots, x_{\ell+1}\}$ on the orbit space with respect to the simple coroot basis, completed by $1$ to a basis of the $(\ell+1)$-dimensional representation space of $\widetilde{W}^{(r)}(\AN_\ell)$.
    
Comparing the previous expression with \cref{eq:EAWsuperpotential} clearly gives $\sigma_0=(-1)^{\ell+1}e^{\varphi_1+\dots+\varphi_{\ell+1}}$.    Now, the sum of the $\bm{\varphi}$-coordinates is, according to the explicit expression of the isomorphism in \cite{dubrovin_zhang_1998}*{Lemma 3.1}, equal to $i2\pi r x_{\ell+1}=rt_{\bullet}$, where $t_\bullet:=i2\pi x_{\ell+1}$ is the last flat coordinate in the system presented in \cite{dubrovin_zhang_1998}*{Corollary 2.5}.
\end{proof}

In a more recent paper \cite{zuo2020frobenius}, this construction is generalised to the EAW groups $\widetilde{W}^{(r,r+1)}(\AN_\ell)$ with two adjacent marked simple roots in the Dynkin diagram\footnote{Later further generalised to the case where the roots are not adjacent in \cite{ma2023frobenius}.}. These are given by the semi-direct product $\widetilde{W}(\AN_\ell)\rtimes\Z^2$, with defining action on $\mathfrak{h}_{\R}\oplus \R^2$ given by:
\begin{align}
    (w,\alpha^\vee, n,m).(x,y,z)&:= \bigl(w(x)+\alpha^\vee+n\,\omega_r+m\,\omega_{r+1},\, y-n,\,z-m\bigr)\,. \label{eq:MZEAW}
\end{align}
One can easily deduce what the action of an EAW group with two arbitrary marked simple roots will look like.\\

In \cite{zuo2020frobenius}*{Section 5}, the author again provides a B-model for the Frobenius manifold structure on each of the orbit spaces $\CC^{\ell+2}/\widetilde{W}^{(r,r+1)}(\AN_{\ell})$. This is, now, given by the same third-kind, Abelian primary differential $\omega$ with simple poles at $0$ and $\infty$, whose residues are $-1$ and $+1$ respectively, on the Hurwitz space $\Hw_{0,\ell+2}(\ell-r-1,r-1,0)$. We, then, know that there is a meromorphic function of the following form in each equivalence class:
\begin{align}
    \lambda_{\ell,r,1}(w)&:= w^{\ell-r}+\sigma_{\ell-1}\,w^{\ell-r-1}+\dots+\sigma_{r}+\tfrac1w \sigma_{r-1}+\dots+\tfrac{1}{w^r}\sigma_0+\tfrac{1}{w-\zeta}\,\rho\,.
\end{align}

    \subsection{Legendre transformations}
    As a final notion we would like to recall, following \citelist{\cite{Dubrovin1996}*{Appendix B} \cite{Strachan_2017}}, we define:
\begin{defn}[Legendre field]
    A Legendre field $\delta\in\X_M$ on a Frobenius manifold $M$ is an invertible holomorphic vector field such that, for any $X,Y\in \X_M$:
    \begin{align}
            X\bullet {}^\eta\nabla_Y\delta&= Y\bullet{}^\eta\nabla_X\delta\,.
    \end{align}
\end{defn}
\begin{oss}
    In particular, this is true whenever $\delta$ is invertible and covariantly constant. This is the case considered in \cite{Dubrovin1996}, and we will mostly stick to that here. The previous condition is a generalisation of that, as it is equivalent to ${}^\eta\nabla_X\delta=X\bullet {}^\eta\nabla_e\delta$ for all $X\in\X_M$.
\end{oss}

\begin{lem}[\cite{Strachan_2017}*{Lemmas 2.2, 2.7}]
Let $M$ be a Frobenius manifold and $\delta$ be a Legendre field.\\The metric $\widetilde{\eta}$, defined by
\begin{align}
    \widetilde{\eta}(X,Y)&:= \eta(\delta\bullet X,\delta\bullet Y)\,,
\end{align}
is still a Frobenius pairing on $M$. In particular, its (flat) Levi-Civita connection is given by:
\begin{align*}
    {}^{\tilde{\eta}}\nabla_XY&= \delta^{-1}\bullet {}^\eta\nabla_X(\delta\bullet Y)\,.
\end{align*}

If, furthermore, $\Lie_E\delta=\kappa\,\delta$ for some $\kappa\in\CC$, then $E$ is again an Euler vector field for $\widetilde{\eta}$ with charge $d-2(\kappa+1)$.
\end{lem}
It, then, follows that $\widetilde{\eta}$ will have its own set of flat coordinates in a neighbourhood of any point of $M$, which we denote by $\widetilde{t}_1,\dots, \widetilde{t}_n$. In these coordinates, the Legendre-transformed Frobenius manifold structure will be given by a prepotential $\widetilde{\Fr}$, again a quasi-homogeneous solution to the WDVV equations. Therefore, from the PDE-theoretic point of view, Legendre transformations can be thought of as “oriented” symmetries of the WDVV equations. The following Proposition spells out how to work out the new flat coordinates and free energy in terms of the old ones.
\begin{prop}[Coordinate expression for Legendre transformations, \cite{Dubrovin1996}]
Let $M$ be a Frobenius manifold, $\delta$ be a Legendre field, $t_1,\dots, t_n$ and $\widetilde{t}_1,\dots, \widetilde{t}_n$ be a set of flat coordinates for $\eta$ and $\widetilde{\eta}$ respectively defined in a neighbourhood of some point $p\in M$.\\Then:
\begin{equation}\label{eq::legtrlocal}
   \begin{aligned}
       \begin{aligned}
        \widetilde{\eta}\bigl(\partial_{\tilde{t}_a},\,\partial_{\tilde{t}_b}\bigr)&=\eta\bigl(\partial_{ t_a},\, \partial_{ t_b}\bigr)\,,\\\pdv[2]{\Fr}{t_a}{t_b}&=\pdv[2]{\widetilde{\Fr}}{\widetilde{t}_a}{\widetilde{t}_b} \, ,  \end{aligned}& &&& a,b&=1,\dots, n\,;
   \end{aligned} 
\end{equation}
where $\Fr$ and $\widetilde{\Fr}$ are the local prepotentials of the Frobenius manifold structures with pairing $\eta$ and $\widetilde{\eta}$ respectively, in the given flat coordinate systems.
\end{prop}

It is important to point out how Legendre transformations work in the Hurwitz Frobenius manifold picture. Namely, as described in \cite{Dubrovin1996}*{Lecture 5}, a Legendre transformation does not modify the underlying Hurwitz space but will, on the other hand, pick out a different primary differential. The way this is done can be easily described in the case where $\delta$ is the coordinate vector field in the direction of one of the flat coordinates for the original Frobenius pairing. As stated in \cite{Dubrovin1996}*{Equation 5.55}, the Legendre-transformed primary form will be given by $\widetilde{\omega}:=-\delta(\lambda)\,\omega$.

\section{Unfolding of the $\AN_\ell$-singularity with tail of simple poles}
\subsection{Generalities on rational functions and Vi\`eta-like identities}
\documentclass[main.tex]{subfile}

Let $\ell\in\Z_{\geq 0}$ and consider the space of monic polynomials of degree $\ell+1$ in the unknown $x$, with complex coefficients\footnote{Any algebraically closed field would also work, but we are here ultimately interested in complex structures.} and with the topology coming from the natural identification with the affine space $\A^{\ell+1}$. A natural choice for a system of coordinates on the open subset of polynomials having only simple roots is given by the roots themselves, the coordinate map being:
\begin{align*}
    (z_0,\dots, z_{\ell})\quad&\mapsto\quad  (x-z_0)\dots (x-z_{\ell}).
\end{align*}
However, this map is clearly not one-to-one as, for any $\sigma\in\Sym_{\ell+1}$, the points $\bm{z}:=(z_0,\dots, z_{\ell})$ and $\sigma.\bm{z}:=(z_{\sigma(0)},\dots, z_{\sigma(\ell)})$ produce the same polynomial. As a matter of fact, the subset of generic polynomials of degree $\ell+1$ is rather identified with the space of orbits for the action of the symmetric group $\Sym_{\ell+1}$ on $\A^{\ell+1}\smallsetminus \Delta$, where $\Delta$ is the union of the diagonals:
\begin{align*}
      \Delta&:= \bigcup_{0\leq a<b\leq \ell}\bigl\{(z_0,\dots, z_{\ell})\in\A^{\ell+1}:\quad z_a=z_b\bigr\}.
\end{align*}

Now, in general, quotienting out an open subset of some complex vector space by the action of a finite group will give rise to an orbifold. However, the point stabilisers are evidently trivial, therefore the quotient space can actually be endowed with a unique natural complex manifold structure. Coordinates $(z_0,\dots, z_{\ell})$ do not factor through the quotient map, as discussed, but will “only” define a \emph{local} coordinate system on the orbit space. 

Another way to see this is by noticing that the quotient of $\A^{\ell+1}\smallsetminus \Delta$ by $\Sym_{\ell+1}$ is the spectrum of the ring of invariants $\CC[z_0,\dots, z_{\ell+1}]^{\Sym_{\ell+1}}$. According to a famous result due to Chevalley, such a ring is freely generated by $\ell+1$ algebraically independent homogeneous polynomials, called \emph{basic invariants}. Hence, $\mathrm{Spec}\bigl(\CC[z_0,\dots, z_{\ell}]^{\Sym_{\ell+1}}\bigr)$ has a natural complex manifold structure of dimension $\ell+1$, as anticipated.\\

 The relation between the basic invariants and the zeros of the polynomial is a classical problem firstly solved by Vi\`eta: the coefficient of the monomial $x^{\ell-k}$, for $k=0,\dots, \ell$, is, up to a sign, the homogeneous elementary symmetric polynomial of degree $k+1$ in the roots: $a_k\in\CC[z_0,\dots, z_{\ell}]^{\Sym_{\ell}}$. In particular, the coordinates $(a_0,\dots, a_\ell)$ do factor through the quotient map, and do therefore define global coordinates on the orbit space.
 
This has a clear interpretation in terms of reflection groups, as the Weyl group $W(\AN_\ell)$ is the symmetric group $\Sym_{\ell+1}$ acting by permutations of the vectors in the standard basis of $\R^{\ell+1}$. Such an action is essential relative to the orthogonal complement $H_\ell$ of the space of its fixed points, which is the line spanned by $(1,\dots, 1)$. According to the Vi\`eta's identities, polynomials corresponding to points in $H_\ell$ will have $a_\ell=-(z_0+\dots+z_\ell)=0$. In other words, monic polynomials of degree $\ell+1$, modulo translations of the unknown, are in one to one correspondence with points in the complexification of the hyperplane $H_\ell$, upon which $W(\AN_\ell)$ acts fixing no points but the origin. The former is of course the same as the space of miniversal unfoldings of the $\AN_\ell$--singularity \cite{arnold1985singularities}.\\

It is natural to now ask a similar question for rational functions. More precisely, let us set $n_z>n_p\in\Z_{\geq 0}$ and consider the space of rational functions which are ratios of two coprime monic polynomials of degree $n_z$ and $n_p$ respectively, with complex coefficients:
\begin{align}
    \lambda(x)&= \frac{x^{n_z}+a_{n_z-1}x^{n_z-1}+\dots+a_0}{x^{n_p}+b_{n_p-1}x^{n_p-1}+\dots+b_0}. \label{eq:rationalquotient}
\end{align}
Let $\{z_1,\dots, z_{n_z}\}$ and $\{p_1,\dots,p_{n_p}\}$ be the zeros of the numerator and of the denominator respectively, so
\begin{align}
    \lambda(x)&= \frac{\displaystyle{(x-z_1)\dots(x-z_{n_z})}}{\displaystyle{(x-p_1)\dots(x-p_{n_p})}}. \label{eq:rationalquotientfactored}
\end{align}
Then, away from the diagonals where any two of the zeros of either polynomial coincide, we can use the 
set $\{\bm{ a},\bm{b}\}$ as coordinates on such a space. Once again, however, the map defined on $\A^{n_z+n_p}$ minus the diagonals is not one-to-one, as, for any $\sigma\in\Sym_{n_z}$ and $\pi\in\Sym_{n_p}$, the points $(\bm{z},\bm{p})$ and $(\sigma.\bm{z},\pi.\bm{p})$ correspond to the same rational function. Therefore, the space actually looks like the quotient of $\A^{n_z+n_p}\smallsetminus\Delta$ by the action of the group $\Sym_{n_z}\times\Sym_{n_p}$ permuting the subspaces $\A^{n_z}$ and $\A^{n_p}$ independently. This is, however, still a manifold, as the action is again free and properly discontinuous.\\
    
In order to relate this to the purely polynomial case, since the degree of the numerator is by construction higher than the degree of the denominator, we perform polynomial division with rational reminder, and can write the rational function as follows:
\begin{align}
    \lambda(x)&= x^{n_z-n_p}+\sigma_{n_z-n_p-1}\,x^{n_Z-n_p-1}+\dots+\sigma_{0}+\frac{\alpha_1}{x-\beta_1}+\dots+\frac{\alpha_{n_p}}{x-\beta_{n_p}}. \label{eq:rationaldivision}
\end{align}
Comparing the two expressions  in \cref{eq:rationalquotientfactored,eq:rationaldivision} yields polynomial relations of the form:
\begin{align*}
    \sigma_p &=  \sigma_p(\bm{z},\bm{ p})\,, \qquad p=0\,,\ldots\,,n_z-n_p-1\,,\\
    \alpha_\mu &=  \alpha_\mu(\bm{z},\bm{ p})\,, \qquad \mu=1\,,\ldots\,,n_p\,,\\
    \beta_\mu &=  \beta_\mu(\bm{ z},\bm{ p})\,, \qquad \mu=1\,,\ldots\,,n_p\,.
\end{align*}
Using the methods developed in \cite{AFS} (in particular, Theorem 4.5), one may then easily show that the Jacobian of the transformation is 
\begin{align}
        J(\bm{z},\bm{p}) &= \det\frac{\partial(\bm{\sigma},\bm{\alpha},\bm{\beta})}{\partial (\bm{z},\bm{p})}  = \kappa\,\, \frac{ \displaystyle{\prod_{1\leq a<b\leq n_z} (z_a-z_b)}}{\displaystyle{\prod_{1\leq \mu<\nu\leq n_p} (p_\mu-p_\nu) }}\,,
\label{eq:jacobian}
\end{align}
where $\kappa\in\CC^*$ is some inessential constant. We therefore obtain a well-defined change of variables away from the diagonals. However, these coordinates are still not one-to-one on the whole affine space:
\begin{prop}[Vi\`eta-like property for rational functions]\label{prop:rational:vietaidentities}
    The coefficients $\sigma_0,\dots, \sigma_{n_z-n_p-1}$ of the quotient in the polynomial division of the numerator by the denominator in the rational function \cref{eq:rationalquotient}, defined as in \cref{eq:rationaldivision}, are invariant polynomials under the action of the group $\Sym_{n_z}\times\Sym_{n_p}$ permuting the zeros and poles respectively:
    \begin{align*}
           \sigma_0,\dots, \sigma_{n_z-n_p-1}&\in\CC[z_1,\dots, z_{n_z}]^{\Sym_{n_z}}\otimes \CC[p_1,\dots, p_{n_p}]^{\Sym_{n_p}}.
    \end{align*}
    
    On the other hand, the coefficients $\alpha_1,\dots, \alpha_{n_p}$ in the remainder are symmetric polynomials in the zeros, but are rational in the poles. Moreover, there is a residual, non-trivial action of the symmetric group $\Sym_{n_p}$ by diagonal permutations of the coordinates $(\bm{\alpha},\bm{\beta})$. Explicitly, for any $\pi\in\Sym_{n_p}$, the points $(\bm{\sigma},\bm{\alpha},\bm{\beta})$ and $(\bm{\sigma},\pi.\bm{\alpha},\pi.\bm{\beta})$ in $\A^{n_z+n_p}$ determine the same rational function.
\end{prop}
\begin{proof}
Clearly, the $\beta$-coordinates are the same as the poles of the denominator in \cref{eq:rationalquotient}, up to reshuffling. Without loss of generality, we take $\beta_\mu=p_\mu$ for $\mu=1,\dots, n_p$.

As for the parameters $\alpha_1,\dots, \alpha_{n_p}$, they are the residues of the rational function $\lambda$ at the corresponding simple pole, therefore:
\begin{equation}\label{eq:residueszerospoles}
\begin{aligned}
        \alpha_\mu&=\lim_{x\to p_\mu}(x-p_\mu)\,\lambda(x)=\frac{p_\mu^{n_z}+a_{n_z-1} p_\mu^{n_z-1}+\dots+a_0}{\prod_{\nu\neq \mu}(p_\mu-p_\nu)}, &&&\mu&=1,\dots, n_p.
\end{aligned}
\end{equation}
Now, since $a_0,\dots, a_{n_z-1}\in\CC[z_1,\dots, z_n]^{\Sym_{n_z}}$, it follows that the action of an element of $\Sym_{n_z}$ leaves $\alpha_\mu$ invariant as well. More interestingly, the transposition in $\Sym_{n_p}$ swapping $p_\alpha$ with $p_\beta$, simultaneously swaps $\beta_\mu$ with $\beta_\nu$ and $\alpha_\mu$ with $\alpha_\nu$, as one can easily work out from the limit expression above.\\

On the other hand, regarding the coefficients of the polynomial quotient, from the two equivalent expressions for $\lambda$ from \cref{eq:rationalquotient} and \cref{eq:rationaldivision}, we get the following polynomial equality:
    \[
    \begin{aligned}
        x^{n_z}+a_{n_z-1}x^{n_z-1}+\dots+a_0&=(x^{n_p}+b_{n_p-1}x^{n_p-1}+\dots+b_0)(x^{n_z-n_p}+\dots+\sigma_{0})+\\
        &\quad +(x-\beta_1)\dots(x-\beta_{n_p})\bigl(\tfrac{\alpha_1}{x-\beta_1}+\dots+\tfrac{\alpha_{n_p}}{x-\beta_{n_p}}\bigr).
    \end{aligned}
    \]
    The highest power coefficients of the first polynomial on the right-hand side are:
    \[
\begin{matrix}
      [x^{n_z}]:&&1,\\
        [x^{n_z-1}]:&& \sigma_{n_z-n_p-1}+b_{n_p-1},\\
        [x^{n_z-2}]:&&\sigma_{n_z-n_p-2}+\sigma_{n_z-n_p-1}b_{n_p-1}+b_{n_p-2},\\
        \vdots&&\vdots\\
        [x^{1}] :&&b_0\sigma_1+b_1\sigma_0,\\
        [x^{0}]: && \sigma_{0}b_0.
\end{matrix}      
    \]
    In general, for $k=0,\dots, n_z$, it is clear that the coefficient of the monomial $x^{k}$ will be the sum of all the monomials of the form $\sigma_{k_1}b_{k_2}$ for all $k_1=0,\dots, n_z-n_p$ and $k_2=0,\dots, n_p-1$ such that $(k_1,k_2)$ is an ordered partition of $k$, with the prescription that $b_{n_p}=\sigma_{n_z-n_p}:=1$.
    
    Similarly, the degree of the second polynomial on the right-hand side is $n_p-1$ and, for $k=1,\dots, n_p$, the coefficient of the monomial $x^{m-k}$ is:
    \[
    \begin{aligned}
            [x^{m-k}]&: &&&(-1)^k\sum_{\mu=1}^{n_p}\alpha_\mu\varpi_{k}(p_0,\dots, \widehat{p}_\mu,\dots, p_{n_p})=\sum_{\mu=1}^{n_p}\alpha_\mu\pdv{b_{n_p-k}}{p_\mu},
    \end{aligned}
    \]$\varpi_k\in\CC[\xi_1,\dots, \xi_\ell]^{\Sym_\ell}$ denoting the basic degree--$k$ $\Sym_\ell$--invariant polynomial, where it is understood that $\varpi_0:=1$.
    
    Consequently, the system of equations for the unknowns $\sigma_0,\dots, \sigma_{n_z-n_p-1}$ one gets by equating the coefficients of the monomials $x^{n_z-1},\dots, x^{n_p}$ to the corresponding ones on the left-hand side of the equation above (i.e. to $a_{n_p},\dots, a_{n_z-1}$) is a square, upper-triangular linear system with ones on the diagonal. Therefore, it admits a unique polynomial solution.
\end{proof}
It follows that, while polynomial functions defined on the space of polynomials we were considering earlier could simply be thought of as functions of the coordinates $a_0,\dots, a_{n_z-1}$ -- i.e. symmetric polynomials in the zeros -- in the rational case we are lead to consider an enlargement of such a ring encoding polynomials in the coordinates $\alpha_1,\dots, \alpha_{n_p},\beta_1,\dots, \beta_{n_p}$ that are invariant under the action we have described in the previous Proposition. Such an action is \emph{diagonal} in the sense that it comes from combining the action of $\Sym_{n_p}\times \Sym_{n_p}$ on $\CC[\bm{\alpha},\bm{\beta}]$ permuting the $\bm{\alpha}$ and $\bm{\beta}$ generators independently, with the diagonal embedding of $\Sym_{n_p}\overset{\Delta}{\hookrightarrow} \Sym_{n_p}\times \Sym_{n_p}$, $\pi\mapsto(\pi,\pi)$. 

It turns out that this is a classical problem, and there is a well-known generating system for the ring of invariants:
\begin{thm}[\cite{weyl}]\label{prop:polarisedpowersums}
    The ring of diagonal invariants $\CC[\bm{x},\bm{y}]^{\Sym_n}$ is generated by the \emph{polarised (power) sums}:
    \begin{align}
            \mathscr{P}_{p,q}(\bm{x},\bm{y})&:= x_1^py_1^q+\dots+x_n^py_n^q,\qquad p,q\in\Z_{\geq 0}.
    \end{align}
\end{thm}
\begin{oss}
   The polarised sums are quite manifestly not algebraically independent. As a matter of fact, the algebraic relations between them are to this date “not well understood” \cite{Gordon_2003}. 
\end{oss}
\begin{oss}
    A key remark here is that the diagonal action of the symmetric group is {\bfseries not} an reflection group action. As a matter of fact, the diagonal subgroup does not contain any reflection at all, as the determinant of any such transformation is clearly positive.
\end{oss}
Owing to the intrinsically rational nature of the residue coordinates, we will in general expect the polarised power sums to be rational functions of the poles of $\lambda$. However, it turns out that, in some important cases for our discussion, these contributions will “magically” cancel out, giving rise again to purely polynomial functions of both the zeros and poles.

In order to see this, we recall the following classical result:
\begin{lem}[Ring of alternating polynomials, \cite{Cauchy_2009}]\label{lem:ringaltpolynomials}
   Let $\mathfrak{A}_n \leq \Sym_n $ denote the alternating group on a set of $n $ elements and $K$ be a field of characteristic $\mathrm{char}(K)\neq 2$. Then, the ring of alternating polynomials is the following quadratic extension of the ring of symmetric polynomials:
   \begin{align*}
   K[x_1,\dots, x_n ]^{\mathfrak{A}_n }&\cong \quotient{K[x_1,\dots,x_n ]^{\Sym_n }[V_n ]}{< V_n ^2-\Delta >}\,,       
   \end{align*}
   where $V_n $ is the Vandermonde determinant:
   \begin{align*}
     V_n  (x_1,\dots, x_n )&:= \prod_{1\leq a<b\leq n }(x_a-x_b),  
   \end{align*}
   and $\Delta\in K[x_1,\dots, x_n]^{\Sym_n}$ is the discriminant.
\end{lem} 

\begin{prop}[Linear diagonal invariants are polynomial]\label{prop:diaginvariantspoly}
    For any $r\in\Z_{\geq 0}$, the diagonal invariants: 
    \begin{align*}
     \Theta_r(\bm{\alpha},\bm{\beta}):=\mathscr{P}_{1,r}(\bm{\alpha},\bm{\beta})\equiv \alpha_1\beta_1^r+\dots +\alpha_{n_p}\beta_{n_p}^r,       
    \end{align*}
     which are linear in the residue coefficients, are symmetric polynomials of the zeros and poles:
    \begin{align*}
    \Theta_r&\in\CC[z_1,\dots, z_{n_z}]^{\Sym_{n_z}}\otimes\CC[p_1,\dots, p_{n_p}]^{\Sym_{n_p}}.   
    \end{align*}
    
    Furthermore, $\Theta_r$ is a homogeneous polynomial of degree $n_z-n_p+1+r$ in $z_1,\dots, z_{n_Z}$, $p_1,\dots, p_{n_p}$.
\end{prop}
\begin{proof}
   According to \cref{eq:residueszerospoles}, the residues only depend on the zeros through the coefficients of the powers of the unknown in the numerator. Hence, $\Theta_r$ is manifestly a symmetric polynomial in the zeros.
   
   As for the poles, we notice that, using \cref{eq:residueszerospoles}, we can write:
    \[
    \Theta_r(\bm{z},\bm{p})=\frac{1}{V_{n_p}(\bm{p})}\sum_{\kappa=1}^{n_p}(-1)^{\kappa-1}p_\kappa^r\, N(p_\kappa)\prod_{\substack{\mu<\nu\\\mu,\nu\neq \kappa}}(p_\mu-p_\nu),
    \]
    where $N(x):=x^{n_z}+a_{n_z-1}x^{n_z-1}+\dots+a_0$ is the numerator of $\lambda$.
    
    Now, the numerator in the right-hand side of the equation above is a polynomial in the poles, and it is clearly divisible by the Vandermonde determinant, for it vanishes whenever $p_\mu=p_\nu$ for any $\mu<\nu=1,\dots, n_p$. To see this, we firstly notice that, when we set $p_\mu=p_\nu$, any summand vanishes except for the $\mu^\mathrm{th}$ and the $\nu^{\mathrm{th}}$ ones. These two terms differ only in the signs in front and in some of the terms appearing in the product of the differences of the poles. In particular, the $\nu^{\mathrm{th}}$ summand will contain all the factors $(p_1-p_\mu)\dots (p_{\mu-1}-p_\mu)(p_\mu-p_{\mu+1})\dots(p_\mu-p_{n_p})$ which are missing in the $\mu^\mathrm{th}$ one, and vice versa. However, when we set $p_\mu=p_\nu$, the factors are actually the same, except that some of them will appear with a different sign in the two cases. In particular, in the $\nu^{\mathrm{th}}$ summand, the $\nu-\mu-1$ factors $p_\mu-p_{\mu+1},\dots, p_\mu-p_{\nu-1}$ must be flipped, when we set $p_\mu=p_\nu$, to match the ones in the $\mu^{\mathrm{th}}$ one. As a consequence, the $\nu^{\mathrm{th}}$ summand will only differ from the $\mu^{\mathrm{th}}$ one by an overall sign given by $(-1)^{\mu-\nu-(\nu-\mu-1)}=-1$. Hence, the two cancel out and the numerator vanishes whenever $p_\mu-p_\nu=0$. As a consequence, $\Theta_r$ is actually a polynomial in the poles as well.
    
    In order to now prove that $\Theta_r$ is also a symmetric polynomial in the pole variables, according to the previous lemma, it suffices to show that $V_{n_p}(\bm{p})\,\Theta_r(\bm{z},\bm{p})$ is alternating under the action of $\Sym_{n_p}$. To this end, it is of course enough to prove that it picks up a minus sign when acted upon by a transposition. Fix, then, $\mu<\nu$ and consider the action of the transposition swapping $\mu$ and $\nu$. An argument analougous to the one we have just presented to prove that the numerator of the above expression for $\Theta_r$ is divisible by the Vandermonde determinant shows that, when acted upon by such transposition, the summands in the denominator are mapped onto their opposite except for the $\mu^{\mathrm{th}}$ and the $\nu^{\mathrm{th}}$ ones, which are sent to minus one another. \\
    
As for the homogeneity, this is clearly equivalent to showing that the residues are homogeneous of degree $n_z-n_p+1$ in the same variables. This is easily done using \cref{eq:residueszerospoles}:
\[
\begin{aligned}
    \alpha_\mu(c\bm{z},c \bm{p})&=\lim_{x\to c p_\mu}(x-c p_\mu)\,\lambda(x,c \bm{z},c\bm{p})=\lim_{y\to p_\mu}(c y-c p_\mu)\,\lambda(c y,c \bm{z},c\bm{p})\\&=c^{1+n_z-n_p}\alpha_\mu(\bm{z},\bm{p}),
\end{aligned}
\]
for any $c\in\CC^*$.
\end{proof}
\begin{oss}
        Notice that, in the case $n_p=2$, the fact that the diagonal invariants $\Theta_r$ are symmetric polynomials in the pole variables follows by the elementary remark that the polynomials $x^n-y^n$, for $n\in\Z_{\geq0}$, are divisible by $x-y$, and that the quotient is invariant under the exchange of the two unknowns -- e.g. we recall the notorious identity $x^3-y^3=(x-y)(x^2+xy+y^2)$.
        As a matter of fact, if the two poles are at $x$ and $y$:
        \[
        \begin{aligned}
                    \Theta_r(\bm{a},x,y)&=\tfrac{1}{x-y}N(\bm{a},x)\,x^r+\tfrac{1}{y-x}N(\bm{a},y)\,y^r\\&=\tfrac{1}{x-y}\bigl[a_0(x^r-y^r)+a_1(x^{r+1}-y^{r+1})+\dots+a_n(x^{r+n}-y^{r+n})\bigr],
        \end{aligned}
        \]
        with $N(\bm{a},w)=a_n w^n+\dots+a_0$.
        
        In fact, one could think of \cref{lem:ringaltpolynomials} as a generalisation of this elementary property to arbitrarily-many variables.
\end{oss} 

In summary, we have seen that there is a natural action of the symmetric group by reflections on the zeros of a polynomial, and that any generic polynomial of fixed degree corresponds to a whole such orbit, rather than to a single point. In particular, the coefficients of each monomial are the basic symmetric polynomials in the zeros, and therefore serve as global coordinates on the space of orbits. Hence, polynomial functions on it can be represented as polynomials in these coefficients.

As opposed to this, we considered the case of a generic rational function made up of the ratio of two polynomials of fixed degrees. Here, any such function corresponds to an orbit by the action of the product of two symmetric groups, permuting the zeros and the poles independently. If we further assume the degree of the numerator is higher than the degree of the denominator, one can perform polynomial division yielding a polynomial quotient and a remainder, which will in general be rational. Now, whereas the coefficients of the monomials in the polynomial quotient will be invariant polynomials under both permutation groups under consideration, the coefficients in the remainder will generically be rational functions in the poles and will exhibit a non-trivial residual action of the symmetric group permuting the poles. As a consequence, polynomial (and rational) functions on the space of such rational functions are going to be polynomial (or rational) in the generators of the invariant ring under this residual action -- as well as in the coefficients of the monomials in the polynomial quotient.

\subsection{Frobenius manifold structure on spaces of generic rational functions}\label{ssec:DFmfldrational}
\documentclass[main.tex]{subfile}

In order to relate these results to the theory of Frobenius manifolds, we once again start from the fact that, owing to \cref{lem:reprequivclass}, the genus-zero Hurwitz space $\Hw_{0,\ell+1}(\ell)$ is -- for $\ell\in \Z_{\geq 1}$ -- precisely the space of monic polynomials of degree $\ell+1$ whose roots sum up to zero -- or, equivalently, the space of miniversal unfoldings of the singularity of type $\AN_\ell$ -- which we have just considered. 

Now, according to the previous discussion, points in such a space will be in one-to-one correspondence with points in the space of orbits of the Weyl group $W(\AN_\ell)$. It is known that both such spaces can be equipped with a natural Frobenius manifold structure. Namely, we have the Saito construction on the Weyl group orbit space, whereas we equip the Hurwitz space with the structure coming from the choice of the second-kind Abelian differential $\omega$ on $\RS$ having a double pole at $\infty$. As we have recalled in the preliminary section, the two spaces are famously also isomorphic as Frobenius manifolds.

Moreover, flat coordinates for both metrics are constructed as follows:
\begin{itemize}
\item As for the intersection form, we recall that, in the Saito construction, this is the metric on Euclidean space descending to the orbit space away from the image of the irregular orbits. Its flat coordinates are, therefore, given by orthogonal coordinates for the Euclidean metric. As discussed, there are only local coordinates. As the Weyl group $W(\AN_\ell)$ acts on $\R^{\ell+1}$ by permutations of the coordinates $z_0,\dots, z_\ell$ with respect to the canonical basis, then these do form a set of orthonormal coordinates. If we restrict to the hyperplane where the coordinates sum up to zero -- i.e. where the action of the group is essential -- and equip it with the induced inner product, then the coordinates $z_1,\dots, z_\ell$ are no longer orthonormal, but still orthogonal. There is then a natural map to the space of polynomials of degree $\ell+1$ whose roots sum up to zero, namely:
\begin{align*}
  (z_1,\dots, z_\ell) \quad &\mapsto\quad (x+z_1+\dots+z_\ell)(x-z_1)\dots(x-z_\ell). 
\end{align*}
This is a local coordinate system on a connected region of the space of such polynomials minus the diagonals, and it can be proven \cite{Dubrovin1996}*{Lemma 4.5} that the intersection form computed via \cref{eq:intformHurwitz} in such coordinates coincides with the one coming from the restriction of the Euclidean metric to the hyperplane. As a consequence, such a map is a local isometry for the intersection form. This gives a Frobenius manifold-theoretic interpretation of the first part of the preceding discussion on spaces of polynomials.
\item Flat coordinates $t_1,\dots, t_\ell$ for the Frobenius metric can be constructed by inversion at infinity of the relation $z^{\ell+1}=\lambda_\ell(w)$ \cites{Dubrovin1996, Hitchin1997}. In particular, the coefficients of the first $\ell$ negative powers of $z$ in the Laurent-series expansion
\begin{align}\label{eq:inversioninfty}
    w &= z-\tfrac{1}{\ell+1}\bigl[\tfrac1z t_\ell+\dots+\tfrac{1}{z^\ell}t_1\bigr]+\order{\tfrac{1}{z^{\ell+1}}}
\end{align}
are flat coordinates for the Frobenius metric on $\Hw_{0,\ell+1}^\omega(\ell)$.

A key remark here is that, by replacing $w$ in \cref{eq:lambdaell} by \cref{eq:inversioninfty} and then comparing the powers of $z$, one gets an upper triangular system for $t_1,\dots, t_\ell$ in terms of $\sigma_0,\dots, \sigma_{\ell-1}$, with the identity on the diagonal. It follows that the flat coordinates are polynomial in the coefficients of the powers of $w$, and are therefore symmetric polynomials in the zeros of $\lambda_\ell$: 
\begin{align*}
 t_1,\dots, t_\ell&\in\CC[\sigma_0,\dots, \sigma_{\ell-1}]\cong\,\quotient{\CC[z_0,\dots, z_\ell]^{\Sym_{\ell+1}}}{<z_0+\dots+z_\ell>}\,.
\end{align*}
 These are called \emph{Saito polynomials}, and they are uniquely associated to any of the Hurwitz spaces $\Hw_{0,\ell+1}^\omega(\ell)$ -- or, equivalently, to any of the Dynkin diagrams of type $\AN$. In particular, it is clear that the only linear term in $\sigma_0,\dots, \sigma_{\ell-1}$ appearing in $t_a$ is $\sigma_{a-1}$ for $a=1,\dots, \ell$. The overall normalisation is chosen so that the coefficient of this term in each $t_a$ is precisely $1$. This also proves that the Saito polynomials are algebraically independent, hence they generate the same ring.
\end{itemize}

Hence, flat coordinates for the two flat metrics coming from the Frobenius manifold structure are related to the parametrisations of polynomials that we discussed in the previous subsection.

\begin{oss}
    For the time being, we are ruling out the case $\ell=0$. The reason is twofold. Firstly, it does not correspond to any Weyl group. Secondly, while the Hurwitz space $\Hw_{0,1}(0)$ is still well-defined for $\ell=0$, the differential $\omega$ is not a primary differential of the first type on it, as both $\omega$ and $\dd\lambda$ have a double pole at $\infty_0$. We shall discuss this case in more details in \cref{ex:l=0}.
\end{oss}

Now, in the rational case, one can write a generic rational function, as in \cref{eq:rationalquotient}, as a sum of the quotient of the numerator by the denominator, with a rational reminder, as in \cref{eq:rationaldivision}. This will uniquely determine a point in some Hurwitz space, according to \cref{lem:reprequivclass}, of equivalence classes of meromorphic functions on the Riemann sphere having a pole of order $\ell+1$ at $\infty$ and a bunch of finite simple poles. One can endow such a space with the very same second-kind Abelian differential having a double pole at $\infty$.

It is natural to then consider the Saito polynomials $t_1,\dots, t_\ell$ coming only from the coefficients of the polynomial contribution. The question is whether these can be completed to a system of flat coordinates for the Frobenius metric on the “larger” Hurwitz space. This is similar to what was done in \cite{Ferguson_2008}*{Lemma 3} for logarithmic deformations. Here, we present a proof of the result for arbitrary rational “deformations”. 

\begin{lem}[Thermodynamic identity]\label{lem:thermo}
Let $\bm{\nu}\in \Z_{\geq 0}^{k}$ and $\nu:=\nu_1+\dots+\nu_k$. Any equivalence class in the Hurwitz space $\Hw_{0,\ell+1+k+\nu}(\ell,\bm{\nu})$ can be uniquely represented as $\lambda=\lambda_\ell+\xi$ for some rational function $\xi$ having $k$ finite poles of order $\nu_1+1,\dots,\nu_k+1$. 
Let, finally, $\bm{t}=(t_1,\dots, t_\ell)$ denote the Saito polynomials on $\Hw_{0,\ell+1}(\ell)$, endowed with the second-kind  Abelian differential $\omega$ having a double pole at $\infty$. Then, there exist functions $\bm{s}:=(s_1,\dots, s_{2k+\nu})$ on an open subset of $\Hw^\omega_{0,\ell+1+k+\nu}(\ell,\bm{\nu})$ such that $(\bm{t},\bm{s})$ is a system of flat coordinates for its Frobenius metric.
\end{lem}
\begin{proof}
    Following \cite{Dubrovin1996}*{Lemma 4.7}, the result is equivalent to $\pdv{\lambda}{t_a}\dd{w}=\bigl[-z^{a-1}+\order{\tfrac{1}{z}}\bigr]\dd{z}$ for any $a=1,\dots, \ell$. This is easily shown by differentiating with respect to $t_a$ the equation:
\[    \lambda(w(z,\bm{t}),\bm{t},\bm{s})=\lambda_\ell(w(z,\bm{t}),\bm{t})+\xi(w(z,\bm{t}),\bm{s})=z^{\ell+1}+    \xi(w(z,\bm{t}),\bm{s}).
\]
Indeed, one gets:
    \[
    \begin{aligned}
\pdv{w}{t_a}\dd{\lambda}+\pdv{\lambda}{t_a}\dd{w}&=\pdv{w}{t_a}\dd{\xi}.
    \end{aligned}
    \]
    Hence:
    \[
    \begin{aligned}
        \pdv{\lambda}{t_a}\dd{w}&=\pdv{w}{t_a}\dd{(\xi-\lambda)}=-\pdv{w}{t_a}\dd{\lambda_\ell}\\
        &=-\bigl[\tfrac{1}{\ell+1} z^{a-(\ell+1)}+\order{\tfrac{1}{z^{\ell+1}}}\bigr](\ell+1)z^\ell\dd{z}=\bigl[-z^{a-1}+\order{\tfrac{1}{z}}\bigr]\dd{z},
    \end{aligned}
    \]
    hence the result.
    \end{proof}

In particular, for the case under consideration, it is not difficult to explicitly construct the coordinates that complete Saito polynomials to a system of flat coordinates on the extended Hurwitz space:

\begin{prop}\label{lem:flatAntail}
Let $\ell,n_{p}\in\Z_{\geq 0}$ and consider the Hurwitz space $\Hw_{0,\ell+n_p+1}(\ell,\underbrace{0,\dots, 0}_{n_p \text{ times}})$ equipped with a primary differential $\omega$ being the second-kind Abelian differential with a double pole at $\infty$.\\
If $\bm{t}:=(t_1,\dots, t_{\ell} )$ denote the Saito polynomials for $\AN_\ell$, then $(\bm{t},\bm{\alpha},\bm{\beta})$, defined as in \cref{eq:rationaldivision}, are flat coordinates for the Frobenius pairing on $\Hw^\omega_{0,\ell+n_p+1}(\ell,\bm{0})$.
\end{prop}
\begin{proof}
The first statement follows from explicit computation of the metric components. The result is:
\[
\eta=\eta_{\AN_\ell}\oplus\begin{bmatrix}
    0&\id_{n_p}\\
    \id_{n_p}& 0
\end{bmatrix}\,.
\]
which proves the statement. In particular, we compute the Frobenius metric in the coordinates $(\bm{\sigma},\bm{\alpha},\bm{\beta})$. For instance, we provide an explicit computation of the components in the leftmost diagonal block:
\[
\begin{aligned}
    \eta_{\sigma_p\sigma_q}&=\sum_{x\in\Gamma_\lambda}\Res_x\biggl\{\frac{w^{p+q}}{\lambda_\ell'(w)-\sum_\mu\tfrac{\alpha_\mu}{(w-\beta_\mu)^2}}\dd{w}\biggr\}\\
    &=-\Res_\infty\biggl\{\frac{w^{p+q}}{\lambda_\ell'(w)-\sum_\mu\tfrac{\alpha_\mu}{(w-\beta_\mu)^2}}\dd{w}\biggr\}\\
    &=\Res_0\biggl\{\frac{1}{\lambda_\ell'(\tfrac{1}{w})-w^2\sum_\mu\tfrac{\alpha_\mu}{(1-w\beta_\mu)^2}}\frac{\dd{w}}{w^{p+q+2}}\biggr\}\,.
\end{aligned}
\]
Now, $\lambda_\ell'\bigl(\tfrac1w\bigr)$ is a rational function with a pole of order $\ell$ at $0$. It follows that $Q(w):=w^\ell\lambda_\ell'\bigl(\tfrac1w\bigr)$ is holomorphic at $0$:
\[
\begin{aligned}
    \eta_{\sigma_p\sigma_q}&=\Res_0\biggl\{\frac{w^{\ell-(p+q+2)}}{Q(w)-w^{\ell+2}\sum_\mu\tfrac{\alpha_\mu}{(1-w\beta_\mu)^2}}\dd{w}\biggr\}\\
    &=\Res_0\biggl\{\frac{w^{\ell-(p+q+2)}}{Q(w)}\biggl[1+\frac{w^{\ell+2}}{Q(w)}\sum_{\mu=1}^{n_p}\frac{\alpha_\mu}{(1-w\beta_\mu)^2}+\order{w^{2(\ell+2)}}\biggr]\dd{w}\biggr\}\,.
\end{aligned}
\]
Since $p,q<\ell$, $2\ell-(p+q)>0$, hence all the terms in the series expansion except for the first one are holomorphic at zero. As a consequence:
\[
\begin{aligned}
    \eta_{\sigma_p\sigma_q}&=-\Res_\infty\biggl\{w^{p+q}\frac{\dd{w}}{\lambda_\ell'(w)}\biggr\}=(\eta_{\AN_\ell})_{\sigma_p\sigma_q}\,.
\end{aligned}
\]

\end{proof}

\documentclass[main.tex]{subfile}

Now that we possess an explicit system of flat coordinates for the Frobenius metric on any space of rational functions coming from “attaching” a tail of simple poles to a surface singularity of type $\AN_\ell$, it is natural to consider what the prepotential in such coordinate system looks like. In particular, since these flat coordinates come from completion of the Saito polynomials for $\AN_\ell$, we expect such free energy to exhibit the very same symmetric behaviour that the corresponding parametrisation of the rational function possess. Namely, the anticipation is that it is a polynomial of $\bm{t}$ and invariant under the diagonal action of the symmetric group permuting the poles. In other words, we expect some polarised power sums to appear in its expression.
\begin{thm}[Prepotential and diagonal invariants]\label{thm:prepotAN}
In the coordinates $(\bm{t},\bm{\alpha},\bm{\beta})$ from \cref{lem:flatAntail}, the prepotential of the Frobenius manifold $\Hw^\omega_{0,\ell+1+n_p }(\ell,\bm{0})$ can be written as:
\begin{align}
    \Fr(\bm{t},\bm{\alpha},\bm{\beta})&= \Fr_{\AN_\ell}(\bm{t})+\Fr_{\mathrm{tail}}(\bm{\alpha},\bm{\beta})+\Fr_{\mathrm{int.}}(\bm{t},\bm{\alpha},\bm{\beta}),
\end{align}
where:\begin{itemize}
    \item $\Fr_{\AN_\ell}(\bm{t})\in\CC[\bm{t}]$ is the prepotential for the Frobenius manifold structure on the orbit space $\CC^\ell/W(\AN_\ell)$.
    \item $\Fr_{\mathrm{tail}}(\bm{\alpha},\bm{\beta})$ encodes the $\mathrm{diag}\,\Sym_{n_p }$-invariant, non-polynomial contribution to the free energy:
    \begin{align}
        \Fr_{\mathrm{tail}}(\bm{\alpha},\bm{\beta})&= \tfrac12\sum_{\mu=1}^{n_p }\alpha_\mu^2\log\alpha_\mu+\sum_{1\leq\mu<\nu\leq n_p }\alpha_\mu\alpha_\nu\log(\beta_\mu-\beta_\nu).
    \end{align}
    \item $\Fr_{\mathrm{int.}}(\bm{t},\bm{\alpha},\bm{\beta})\in\bigoplus_{k=0}^{\ell+2}\CC[\bm{t}]\,\Theta_k(\bm{\alpha},\bm{\beta})\subseteq\CC[\bm{t}][\bm{\alpha},\bm{\beta}]^{\mathfrak{S}_{n_p }}$ is an invariant polynomial in $\bm{\alpha},\bm{\beta}$ with respect to the diagonal action of $\mathfrak{S}_{n_p }$, with coefficients in $\CC[\bm{t}]$, lying in the subspace generated by a finite subset of the linear diagonal invariants from \cref{prop:diaginvariantspoly}.\\
    Explicitly:
    \begin{align}
    \begin{split}
             \Fr_{\mathrm{int.}}(\bm{t},\bm{\alpha},\bm{\beta})&= \tfrac{1}{\ell+2}\Theta_{\ell+2}(\bm{\alpha},\bm{\beta})+\tfrac{1}{\ell}\sigma_{\ell-1}(\bm{t})\,\Theta_\ell(\bm{\alpha},\bm{\beta})+\dots\\&\qquad\dots+\sigma_0(\bm{t})\,\Theta_1(\bm{\alpha},\bm{\beta})+f(\bm{t})\,\Theta_{0}(\bm{\alpha},\bm{\beta}).    
    \end{split}
\end{align}
Here, $f(\bm{t})\in\CC[\bm{t}]/\CC[\bm{t}]^{(\leq 1)}$ is a solution to the system of PDEs:
\begin{align}\label{eq:PDEfAn}
    \pdv[2]{f}{t_a}{t_b}&= \sum_{p,q=0}^{\ell-1}\pdv{\sigma_p}{t_a}\pdv{\sigma_q}{t_b}\Res_0\biggl\{ w^{\ell-(p+q+1)}\frac{\dd{w}}{Q(w)}\biggr\},\qquad \qquad 1\leq a\leq b\leq  \ell,
\end{align}
$Q$ being the polynomial $Q(w):=w^\ell\lambda_\ell'\bigl(\tfrac{1}{w}\bigr)\in\CC[\bm{t}][w]$ controlling the singular behaviour of $\lambda'_\ell$ at $\infty$, where $'\equiv\partial_w$.
\end{itemize} 
\end{thm}
\begin{oss}
    $f(\bm{t})$ is indeed a polynomial in $\bm{t}$, since $\sigma_0,\dots, \sigma_{\ell-1}\in\CC[\bm{t}]$ and $Q^{(k)}(0)\in\CC[\bm{\sigma}]\cong\CC[\bm{t}]$ for any $k\in\Z_{\geq 0}$, with $Q(0)=\ell+1$ being a constant.
\end{oss}
\begin{oss}
    Notice that the prepotential only depends on the polarised power sums that are linear in the residues, therefore precisely the ones that are homogeneous symmetric polynomials in the zeros and poles, as highlighted in \cref{prop:diaginvariantspoly}.
\end{oss}
\begin{oss}
    Finally, it is interesting to point out that, in the case of a single finite simple pole, the Hurwitz space $\Hw_{0,\ell+2}^\omega(\ell,0)$ is a B-model for the Frobenius manifold structure associated to the Weyl group $W(\mathscr{B}_{\ell+2})$ from \cite{Arsie_2022}, as described in \cite{MaZuoBn}. In particular, the corresponding prepotentials will be solutions to the WDVV equations associated to the constrainted KP hierarchy \cite{LIU2015177}. A natural question would be to ask whether a similar integrable-hierarchy-theoretic description can be found for an higher number of (simple) poles. 
\end{oss}

\begin{proof}
    We show that such a function arises from integration of the three-point function components in the given system of flat coordinates for the Frobenius pairing. Actually, since the Saito polynomials will generically be symmetric polynomials in the zero of the polynomial quotient $\lambda_\ell$ of the numerator by the denominator, it is actually easier to work in the coordinates $\{\bm{\sigma},\bm{\alpha},\bm{\beta}\}$ and apply the usual rules for the tensor components in different coordinate systems.

    To begin with, we compute, for $p,q,k=0,\dots, \ell-1$:
    \[
    \begin{aligned}
        c_{\sigma_p\sigma_q\sigma_k}&=\sum_{x\in\Gamma_\lambda}\Res_x\biggl\{\frac{w^{p+q+k}}{\lambda_\ell'(w)-\sum_\mu\tfrac{\alpha_\mu}{(w-\beta_\mu)^2}}\dd{w}\biggr\}=-\Res_\infty\biggl\{\frac{w^{p+q+k}}{\lambda_\ell'(w)-\sum_\mu\tfrac{\alpha_\mu}{(w-\beta_\mu)^2}}\dd{w}\biggr\}\\
        &=\Res_0\biggl\{\frac{w^{\ell-(p+q+k+2)}}{Q(w)-w^{\ell+2}\sum_\mu\tfrac{\alpha_\mu}{(1-w\beta_\mu)^2}}\dd{w}\biggr\}\\
        &=\Res_0\biggl\{\frac{w^{\ell-(p+q+k+2)}}{Q(w)}\biggl[1+\tfrac{w^{\ell+2}}{Q(w)}\sum_\mu\frac{\alpha_\mu}{(1-w\beta_\mu)^2}+\order{w^{2(\ell+2)}}\biggr]\dd{w}\biggr\}\,.
    \end{aligned}
    \]
    Now, since $p,q,k<\ell\,$, it follows that $\ell-(p+q+r+2)+2(\ell+2)>2>0\,$. As a consequence, the additional terms in the series expansion are holomorphic at zero. Hence, we are left with:
    \[
    \begin{aligned}
        c_{\sigma_p\sigma_q\sigma_k}&=-\Res_\infty\biggl\{w^{p+q+k}\frac{\dd{w}}{\lambda_\ell'(w)}\biggr\}+\sum_{\mu=1}^{n_p}\alpha_\mu\sum_{n\geq 0}(n+1)\beta_\mu^n\Res_0\biggl\{\frac{1}{Q(w)^2}w^{2\ell+n-(p+q+r)}\dd{w}\biggr\}\,.
    \end{aligned}
    \]
    This proves that, for some coefficients $\gamma_{npqk}\in\CC[\bm{t}]$ such that, for any fixed choice of $p,q,k$, $\gamma_{npqk}$ is non-zero only for finitely many $n\in\Z_{\geq 0}$, it is:
        \[c_{\sigma_p\sigma_q\sigma_k}=c^{\AN_\ell}_{\sigma_p\sigma_q\sigma_k}+\sum_{n\geq 0}\gamma_{npqk}(\bm{t})\,\Theta_n(\bm{\alpha},\bm{\beta})\,,
        \]
        where $c^{\AN_\ell}_{\sigma_p\sigma_q\sigma_\ell}$ is the corresponding three-point function component on $\Hw_{0,\ell+1}(\ell)$ in the $\bm{\sigma}$ coordinates. This gives the following PDEs for $\Fr$:
        \[
        \sum_{a,b,c=1}^\ell\pdv{t_a}{\sigma_p}\pdv{t_b}{\sigma_q}\pdv{t_c}{\sigma_k}\frac{\partial^3(\Fr-\Fr_{\AN_\ell})}{\partial t_a \partial t_b\partial t_c}=\sum_{n\geq 0}\gamma_{npqk}(\bm{t})\,\Theta_n(\bm{\alpha},\bm{\beta})\,.
        \]
       Since $\Fr_{\AN_\ell}$ and its derivatives are independent of $\bm{\alpha},\bm{\beta}$, it follows that the equations are solved by $\Fr=\Fr_{\AN_\ell}+\Phi$ for some $\Phi(\bm{t},\bm{\alpha},\bm{\beta})$ whose third derivatives in the $\bm{t}$ variables are polynomials. In particular, $\Phi$ will encode all the $\bm{\alpha},\bm{\beta}$ dependence of $\Fr$.
       
        We start fixing $\Phi$ by looking at:
        \[ \Phi_{\alpha_\mu\alpha_\nu\alpha_\rho}=\Fr_{\alpha_\mu\alpha_\nu\alpha_\rho}=c_{\alpha_\mu\alpha_\nu\alpha_\rho}=\delta_{\mu\nu}\delta_{\nu\rho}\tfrac{1}{\alpha_\mu}\,.
        \]
        The three-point function components here are computed with similar computations as the ones before, using the derivatives of $\lambda$ with respect to the $\alpha$-coordinates this time. In particular, the only difference in these cases is that there might be non-residueless singularities at the finite poles of $\lambda$, rather than just at $\infty$. Integrating up yields $\Phi(\bm{t},\bm{\alpha},\bm{\beta})=\tfrac12\sum_{\mu=1}^{n_p }\alpha_\mu^2\log\alpha_\mu+\Psi(\bm{t},\bm{\alpha},\bm{\beta})$, for some $\Psi$ depending at most quadratically on the $\bm{\alpha}$ variables.
        
        We start probing the $\bm{\beta}$-dependence by looking at the following three-point function components:
        \[
        \Phi_{\alpha_\mu\alpha_\nu\beta_\rho}=c_{\alpha_\mu\alpha_\nu\beta_\rho}=\tfrac{1}{\beta_\mu-\beta_\nu}(\delta_{\mu\rho}-\delta_{\nu\rho}).
        \]
        This fixes $\Psi$ as follows: $\Psi(\bm{t},\bm{\alpha},\bm{\beta})=\tfrac12\sum_{\mu\neq\nu}\alpha_\mu\alpha_\nu\log(\beta_\mu-\beta_\nu)+\Omega(\bm{t},\bm{\alpha},\bm{\beta})$. Here, $\Omega$ cannot depend on any non-zero power of any monomial of the form $\alpha_\mu\alpha_\nu\beta_\mu$ for any choice of $\mu, \nu=1,\dots, {n_p }$.
        
        Next, we consider:
        \[
        \Psi_{\alpha_\mu\beta_\nu\beta_\rho}=c_{\alpha_\mu\beta_\nu\beta_\rho}=\begin{cases}
        \alpha_\rho\tfrac{1}{(\beta_\mu-\beta_\rho)^2} & \mu=\nu\neq \rho\\
        \alpha_\nu\tfrac{1}{(\beta_\mu-\beta_\nu)^2} & \mu=\rho\neq \nu\\
        \lambda_\ell'(\beta_\mu)-\sum_{\gamma\neq\mu}\alpha_\gamma\tfrac{1}{(\beta_\mu-\beta_\gamma)^2 } & \mu=\nu=\rho\\
        -\alpha_\nu\tfrac{1}{(\beta_\mu-\beta_\nu)^2} &\mu\neq \nu=\rho\\
        0 & \text{otherwise}
    \end{cases}.
        \]
        Clearly, all the linear terms in $\bm{\alpha}$ come from the logarithmic terms in $\beta_\mu-\beta_\nu$ that we have just computed, therefore we should simply impose that $\Omega$ satisfy:
        \[
        \Omega_{\alpha_\mu\beta_\mu\beta_\mu}=\lambda_\ell'(\beta_\mu)=(\ell+1)\beta_\mu^\ell+(\ell-1)\sigma_{\ell-1}\beta_\mu^{\ell-2}+\dots+2\sigma_2\beta_\mu+\sigma_1,
        \]
        while $\Omega_{\alpha_\mu\beta_\nu\beta_\rho}$ vanishes for any other choice of $\mu,\nu$ and $\rho$. This gives:
        \[
        \begin{aligned}
        \Omega(\bm{t},\bm{\alpha},\bm{\beta})&=\tfrac{1}{\ell+2}\Theta_{\ell+2}(\bm{\alpha},\bm{\beta})+\tfrac{1}{\ell}\sigma_{\ell-1}(\bm{t})\Theta_\ell(\bm{\alpha},\bm{\beta})+\dots+\tfrac12\sigma_1(\bm{t})\Theta_{2}(\bm{\alpha},\bm{\beta})+\\
        &\quad+c_1(\bm{t})\alpha_1\beta_1+\dots+c_{n_p}(\bm{t})\alpha_{n_p}\beta_{n_p}+f_1(\bm{t})\alpha_1+\dots+f_{n_p}(\bm{t})\alpha_{n_p},
        \end{aligned}
        \]
        for some unknown polynomials $c_1,\dots, c_{n_p}\in\CC[\bm{t}]/\CC$ and $f_1,\dots, f_{n_p}\in\CC[\bm{t}]/\CC[\bm{t}]^{(\leq 1)}$.\\
        We fix $c_1,\dots, c_{n_p}$ by looking at the following three-point function components:
        \[
        \begin{aligned}
            \Omega_{t_a\alpha_\mu\beta_\nu}&=\Fr_{t_a\alpha_\mu\beta_\nu}=\sum_{p=0}^{\ell-1}\pdv{\sigma_p}{t_a}\,c_{\sigma_p\alpha_\mu\beta_\nu}=\delta_{\mu\nu}\biggl[\pdv{\sigma_0}{t_a}+\pdv{\sigma_1}{t_a}\beta_\mu+\dots+\pdv{\sigma_{\ell-1}}{t_a}\beta_\mu^{\ell-1}\biggr].
        \end{aligned}
        \]
On the other hand, the same components can now be computed by explicit differentiation of the expression for $\Omega$ we have derived:
\[
        \Omega_{t_a\alpha_\mu\beta_\nu}=\delta_{\mu\nu}\biggl[\pdv{c_\mu}{t_a}+\pdv{\sigma_1}{t_a}\beta_\mu+\dots+\pdv{\sigma_{\ell-1}}{t_a}\beta_\mu^{\ell-1}\biggr].
        \]
Hence, it actually is $c_1=\dots=c_{n_p}=\sigma_0$ in $\CC[\bm{t}]/\CC$.

The only freedom we are left with is, now, the choice of the functions $f_1,\dots, f_{n_p}$. If we show that they are also all equal to the same polynomial function $f$ in the Saito polynomials, then the statement follows. We, therefore, again compare the result of the calculation of the following three-point function components:
        \[
        \begin{aligned}
            \Omega_{t_at_b\alpha_\mu}&=\sum_{p,q=0}^{\ell-1}\pdv{\sigma_p}{t_a}\pdv{\sigma_q}{t_b}\,c_{\sigma_p\sigma_q\alpha_\mu}=\sum_{n\geq 0}\beta_\mu^n\sum_{p,q=0}^{\ell-1}\pdv{\sigma_p}{t_a}\pdv{\sigma_q}{t_b}\Res_0\biggl\{w^{\ell+n-(p+q+1)}\frac{\dd w}{Q(w)}\biggr\},
        \end{aligned}
        \]
with explicit differentiation of $\Omega$:
\[
        \Omega_{t_at_b\alpha_\mu}=\pdv[2]{f_\mu}{t_a}{t_b}+\pdv[2]{\sigma_0}{t_a}{t_b}\beta_\mu+\tfrac12\pdv[2]{\sigma_1}{t_a}{t_b}\beta_\mu^2+\dots+\tfrac{1}{\ell}\pdv[2]{\sigma_{\ell-1}}{t_a}{t_b}\beta_\mu^\ell.
        \]
        This gives a system of PDEs for the elementary symmetric polynomials as functions of the Saito polynomials and, more remarkably, the following system of second-order PDEs for the unknown functions $f_1,\dots, f_{n
        _p}$:
        \[
        \pdv[2]{f_\mu}{t_a}{t_b}=\sum_{p,q=0}^\ell\pdv{\sigma_p}{t_a}\pdv{\sigma_q}{t_b}\Res_0\biggl\{w^{\ell-(p+q-1)}\frac{\dd w}{Q(w)}\biggr\}.
        \]
        In particular, each $f_\mu$ satisfies the same system of second-order PDEs -- whose coefficients are polynomial in the Saito polynomials. Hence, any two functions $f_\mu$ and $f_\nu$ will at most differ by a linear function of the Saito polynomials. This proves that, in $\CC[\bm{t}]/\CC[\bm{t}]^{(\leq 1)}$, we can take $f_1=\dots=f_{n_p}=:f$.
        \end{proof}

The unity and Euler vector fields may easily be expressed in the system of flat coordinates:
\begin{prop}
    In the coordinate systems defined in \cref{eq:rationaldivision} and \cref{lem:flatAntail}, the unity and Euler vector field for the Frobenius manifold structure on the Hurwitz space $\Hw^\omega_{0,\ell+1+n_p}(\ell,\bm{0})$ are:
   \begin{align}
               e&=\partial_{\sigma_0}=\partial_{t_1},\\
        \begin{split}
             E&=  \sum_{k=0}^{\ell-1}\bigl(1-\tfrac{k}{\ell+1}\bigr)\sigma_k\partial_{\sigma_k}+\tfrac{\ell+2}{\ell+1}\bigl[\alpha_1\partial_{\alpha_1}+\dots+\alpha_{n_p}\partial_{\alpha_{n_p}}\bigr]+\tfrac{1}{\ell+1}\bigl[\beta_1\partial_{\beta_1}+\dots+\beta_{n_p}\partial_{\beta_{n_p}} \bigr]\\
             &=\sum_{k=1}^{\ell}\bigl(1-\tfrac{k-1}{\ell+1}\bigr)t_k\partial_{t_k}+\tfrac{\ell+2}{\ell+1}\bigl[\alpha_1\partial_{\alpha_1}+\dots+\alpha_{n_p}\partial_{\alpha_{n_p}}\bigr]+\tfrac{1}{\ell+1}\bigl[\beta_1\partial_{\beta_1}+\dots+\beta_{n_p}\partial_{\beta_{n_p}} \bigr]\,,
        \end{split} 
   \end{align}
   where the relation $\ell+1=n_z-n_p$ is understood.
\end{prop}
\begin{proof}
    As for the identity, it is characterised by the property $\Lie_e\lambda=1$ \cite{ma2023frobenius}. It is obvious that $\partial_{\sigma_0}$ satisfies such requirement. In the flat coordinate system, on the other hand, we recall that the transition functions are found by solving an upper-triangular linear system with ones on the diagonal. As a consequence, $t_1$ is going to be the only flat coordinate depending on $\sigma_0$ at all, and the explicit coordinate transformation is going to be of the following form: $$t_1=\sigma_0+P(t_2,\dots, t_n)=\sigma_0+\widetilde{P}(\sigma_1,\dots,\sigma_{n-1}),$$ for some polynomial $P$ -- where $\widetilde{P}(\sigma_1,\dots, \sigma_{n-1}):=P(t_2(\bm{\sigma}),\dots, t_n(\bm{\sigma}))$. This proves that $\partial_{\sigma_0}=\partial_{t_1}$.\\

    As for the Euler vector field, it follows from Euler's Theorem for homogeneous functions that the superpotential $\lambda$ satisfies:
    \[
    \bigl(\Lie_E\lambda\bigr)(w)=(\ell+1)\,\lambda(w)-w\,\lambda'(w),
    \]
    up to an overall normalisation factor \cite{ma2023frobenius}. 
    Clearly, in the coordinates $(z_1,\dots, z_{n_z},p_1,\dots, p_{n_p})$, the vector field $z_1\partial_{z_1}+\dots+z_{n_z}\partial_{n_z}+p_1\partial_{p_1}+\dots+p_{n_p}\partial_{n_p}$ satisfies the above relation.
    
    The overall normalisation is found by requiring that the component in the direction of the identity be one. This is obviously the same as $E(\sigma_0)=\sigma_0$. Now, using the relations between the coefficients in \cref{eq:rationaldivision} and those in \cref{eq:rationalquotient} as in the proof of \cref{prop:rational:vietaidentities}, one finds, in particular:
    \[
    a_{n_p}=\sigma_0+b_{n_p-1}\sigma_1+\dots+b_0\sigma_{n_p}.
    \]
    Hence, as a homogeneous polynomial in the zeros $z_1,\dots, z_{n_z}$ and the poles $p_1,\dots, p_{n_p}$ of $\lambda$, $\sigma_0$ has the same degree of $a_{n_p}$, which is $n_z-n_p=\ell+1$. As a consequence, the overall normalisation of the Euler vector field has to be $\tfrac{1}{\ell+1}$. With a similar argument, one shows that the degree of $\sigma_k$ as a homogeneous polynomial of the zeros and poles is the same as the one of $a_{n_p+k}$, for $k=0,\dots, n-1$, and that each of the $\alpha$-s has the same degree as $a_{n_p-1}$. This proves the first expression for the Euler vector field. As for the second one, one just need to again remember that $t_k$ has the same homogeneity degree of $\sigma_{k-1}$, as previously discussed in this very proof.
\end{proof}
\begin{oss}
    Notice that, since $\Fr$ is quasi-homogeneous with respect to $E$ with degree $3-d=2\tfrac{\ell+2}{\ell+1}$, then the polynomial $f(\bm{t})$ defined by \cref{eq:PDEfAn} must also be a quasi-homogeneous polynomial in the Saito coordinates, with respect to the Euler vector field $E_{\AN_\ell}$, of degree $d_f=\tfrac{\ell+2}{\ell+1}$, i.e.:
    \[
    \begin{aligned}
  f(c\,t_1, c^{\ell/(\ell+1)}t_2,\dots, c^{2/(\ell+1)}t_\ell)=c^{\tfrac{\ell+2}{\ell+1}}f(\bm{t})\,,\qquad \forall c\in\CC^*\,.
    \end{aligned}
    \]
\end{oss}

\subsection{Examples}
We provide some explicit expressions of the prepotential for some choices of $\ell$ and ${n_p }$, for which explicit expressions for the polynomials $\bm{\sigma}(\bm{t})$ and $f(\bm{t})$ are provided.
\begin{ex}[The case $\ell=0$]\label{ex:l=0}
As we have said, on the Hurwitz spaces $\Hw_{0,m+1}(0,\dots, 0)$, the differential $\omega$ we are considering is not a primary differential of the first type in the sense of Dubrovin \cite{Dubrovin1996}. However, one can notice the following, starting from the lowest-dimensional case, i.e. $m=1$. We are, therefore, considering the space of meromorphic functions on $\RS$ with two simple poles, which we can represent, up to projective equivalence, as follows:
\[
\lambda(z)=z+\frac{\alpha}{z-\beta}\,.
\]
A primary differential of the second type on it -- in fact, the only one -- is given by $\phi:=-\tfrac{\alpha}{(z-\beta)^2}\dd{z}$. Now, if we, rather, consider a different representative in each equivalence class, namely the one defined by the M\"{o}bius transformation $w:=\tfrac{\alpha}{z-\beta}+\beta$, then $\phi=\dd{w}$ while the superpotential is still:
\[
\lambda(w)=w+\frac{\alpha}{w-\beta}\,.
\]
Hence, according to our results from \cref{lem:flatAntail} and \cref{thm:prepotAN}, $(\alpha,\beta)$ is a system of flat coordinates and the prepotential is:
\[
\Fr_0(\alpha,\beta)=\tfrac12\alpha^2\log\alpha+\tfrac12\alpha\beta^2\,,
\]
where it is understood that $\Fr_{\AN_0}:=0\,$. This is, in fact, the solution associated to the Weyl group $W(\mathscr{B}_2)$ in the work \cite{Arsie_2022}, and it is known that $\Hw_{0,2}(0,0)$ with primary differential $\phi$ is a B-model for it, see \cite{MaZuoBn} -- in particular Example II.3. This is also the solution $\mathcal{F}_{\phi_{2}}$ for $m=1$ given in \cite[Theorem 6.3]{rejeb2023new}, upon identifying $\beta=y_{1,1}$ and $\alpha=y_{1,2}$. \\

For $m>1$, we have the superpotential:
\[
\lambda(z)=z+\frac{a_1}{z-b_1}+\dots+\frac{a_m}{z-b_m}\,,
\]
and we once again pick the second-type primary differential $\phi_1:=-\tfrac{a_1}{(z-b_1)^2}\dd{z}\,$. We, then, change our representative in the equivalence class on $\Hw_{0,m+1}(\bm{0})$ by $x-b_1:=\tfrac{a_1}{z-b_1}\,$, so that $\phi_1=\dd{x}$ and:
\[
\begin{aligned}
    \lambda(x)&=x+\frac{a_1}{x-b_1}+\sum_{\nu>1}\frac{a_\nu}{b_1-b_\nu}\frac{x-b_1}{x-b_1+\tfrac{a_1}{b_1-b_\nu}}\\
    &=x+\frac{a_1}{x-b_1}+\sum_{\nu>1}\biggl[\frac{a_\nu}{b_1-b_\nu}-\frac{a_1 a_\nu}{(b_1-b_\nu)^2}\frac{1}{x-b_1+\tfrac{a_1}{b_1-b_\nu}}\biggr]\,.
\end{aligned}
\]
Therefore, if we finally set $x+\sum_{\nu>1}a_\nu\,/\,(b_1-b_\nu)=:w$, then the primary differential is still going to be $\phi_1=\dd{w}$, whereas the superpotential will be put into our standard form:
\[
\lambda(w)=w+\frac{\alpha_1}{w-\beta_1}+\dots+\frac{\alpha_m}{w-\beta_m}\,,
\]
with:
\[
\begin{aligned}
    \alpha_1&:=a_1\,, &&& \beta_1&:=b_1+\sum_{\nu>1}\frac{a_\nu}{b_1-b_\nu}\,,\\
    \alpha_\mu&:=-\frac{a_\mu a_1}{(b_1-b_\mu)^2}\,, &&& \beta_\mu&=b_1-\frac{a_1}{b_1-b_\mu}+\sum_{\nu>1}\frac{a_\nu}{b_1-b_\nu}\,,&&&\mu>1\,.
\end{aligned}
\]
The coordinates $(\bm{\alpha},\bm{\beta})$ are flat and the corresponding prepotential is, according to \cref{thm:prepotAN}:
\[
\Fr_m(\bm{\alpha},\bm{\beta})=\tfrac12\sum_{\mu=1}^m\alpha_\mu^2\log\alpha_\mu+\sum_{\mu<\nu}\alpha_\mu\alpha_\nu\log(\beta_\mu-\beta_\nu)+\tfrac12\sum_{\mu=1}^m\alpha_\mu\beta_\mu^2\,.
\]
This should be the same as the solution(s) $\mathcal{F}_{\phi_{2m+1-j}}$ given in \cite[Theorem 6.3]{rejeb2023new}. In order to see this, we need to perform inversion on their representative for the superpotential, as defined in \cite[Eq. 6.1]{rejeb2023new}:
\[
\begin{aligned}
    \lambda&=\tfrac1z+\sum_{\mu=1}^m\frac{r_\mu}{z-p_\mu}=w+\sum_{\mu=1}^m\frac{-r_\mu\,/\,p_\mu^2}{w-\bigl(\tfrac{1}{p_\mu}-\sum_\nu \tfrac{r_\nu}{p_\nu}\bigr)}\,.
\end{aligned}
\]
This gives $a_\mu-=-r_\mu p_\mu^{-2}$ and $b_\mu=p_\mu^{-1}-\sum_\nu r_\nu p_\nu^{-1}\,$, for $\mu=1,\dots, m$. As a consequence, the flat coordinates given in \cite[Proposition 6.3 ii)]{rejeb2023new} for $k=1$ are related to our flat coordinates by the following:
\[
\begin{aligned}
y_{1,2m}&=-\frac{r_1}{p_1^2}-\sum_{\mu>1}\frac{r_1r_\mu}{p_1^2p_\mu^2}\frac{1}{(p_1^{-1}-p_\mu^{-1})^2}=a_1+\sum_{\mu>1}\frac{-a_1a_\mu}{(b_1-b_\mu)^2}=\sum_{\mu=1}^m\alpha_\mu\,,\\
y_{1,2m+1-\mu}&=\frac{r_1r_\mu}{p_1^2p_\mu^2}\frac{1}{(p_1^{-1}-p_\mu^{-1})^2}=-\alpha_\mu\,, &&&\mu>1\,,\\
    y_{1,1}&=\frac{1}{p_1}-\sum_\mu\frac{r_\mu}{p_\mu}-\sum_{\mu>1}\frac{-r_1}{p_1^2}\frac{1}{p_1^{-1}-p_\mu^{-1}}=b_1-\sum_{\mu>1}\frac{a_\mu}{b_\mu-b_1}=\beta_1\,,\\
    y_{1,\mu}&=-\frac{r_1}{p_1^2}\frac{1}{p_1^{-1}-p_\mu^{-1}}=\frac{a_1}{b_1-b_\mu}=\beta_1-\beta_\mu\,,&&&\mu>1\,.
\end{aligned}
\]
The coordinate transformation is, as expected, linear and the terms in the prepotentials match:
\[
\begin{aligned}
\tfrac12 y_{1}^2y_{2m}+y_{1}\sum_{\mu>1}y_{\mu}y_{2m+1-\mu}-\tfrac12\sum_{\mu>1}y_{\mu}^2y_{2m+1-\mu}=\tfrac12\sum_{\mu=1}^m\alpha_\mu\beta_\mu^2\,,\\
    \tfrac12\bigl(\sum_{\mu}y_{\mu}\bigr)^2\log\sum_\mu y_{\mu}+\tfrac12\sum_{\mu>1}y_{2m+1-\mu}^2\log y_{2m+1-\mu}=\tfrac12\sum_\mu \alpha_\mu^2\log\alpha_\mu\,,\\
   \tfrac12\sum_{\mu\neq \nu>1}y_{2m+1-\mu}y_{2m+1-\nu}\log(y_{\nu}-y_{\mu}) -\sum_{\substack{\mu\neq\nu\\\nu>1}}y_{2m+1-\mu}y_{2m+1-\nu}\log y_{\nu}&=\sum_{\mu<\nu}\alpha_\mu\alpha_\nu\log(\beta_\mu-\beta_\nu)\,,
\end{aligned}
\]
where we omitted $1,$ in all the indices on the left-hand side.

As a final remark, all the Abelian differentials $\phi_\mu:=-\tfrac{a_\mu}{(z-b_\mu)^2}\dd{z}$ for $\mu=1,\dots, m$ are primary differentials of the second type on $\Hw_{0,m+1}(\bm{0})$. However, they are all mapped into one another by permutations of the simple poles, therefore they determine isomorphic Frobenius manifold structures. In fact, the associated solution(s) in \cite[Theorem 6.3]{rejeb2023new} are all the same, and equal to $\Fr_m$ as given above, up to relabelling the coordinates. Clearly, with an obvious adjustment in the M\"{o}bius transformations introduced above, one could put the superpotential and primary differential in our standard form for any choice of $\phi_\mu$.     
\end{ex}

\begin{ex}[$\AN_1$-singularity with two simple poles]
    Consider the Hurwitz space $\Hw_{0,4}(1,0,0)$ -- i.e. we set $\ell=1$ and $n_p =2$ -- of rational functions of the form:
    \begin{align*}
    \lambda(w)&= w^2+t+\frac{\alpha_1}{w-\beta_1}+\frac{\alpha_2}{w-\beta_2},        
    \end{align*}
    with primary differential $\dd w$.
    
    It is well-known that $t$ is itself a flat coordinate for the $W(\AN_1)$ orbit space Frobenius manifold, and that the prepotential is $\Fr_{\AN_1}(t):=\tfrac{1}{12}t^3$.
    
    The Frobenius metric on $\Hw_{0,4}^{\dd w}(1,0,0)$ is, therefore, in the flat coordinates $\{t,\alpha_1,\alpha_2,\beta_1,\beta_2\}$:
    \begin{align*}
          \eta&= \begin{bmatrix}
        \tfrac12 & 0 & 0 &0 &0\\
        0 & 0 & 0 &1&0\\
        0&0&0&0&1\\
        0&1&0&0&0\\
        0&0&1&0&0
    \end{bmatrix},  
    \end{align*}
    whereas the prepotential is:
    \begin{align*}
        \Fr(t,\alpha_1,\alpha_2,\beta_1,\beta_2)&= \tfrac{1}{12}t ^3+\tfrac12 \alpha_1^2\log\alpha_1+\tfrac12 \alpha_2^2\log\alpha_2+\alpha _1\alpha_2\log(\beta_1-\beta_2)+\\ &\qquad +\tfrac13 \bigl[\alpha_1\beta_1^3+\alpha_2\beta_2^3\bigr] +t\bigl[\alpha_1\beta_1+\alpha_2\beta_2\bigr].        
    \end{align*}
   As a matter of fact, in this case it is $f=0$, as the one PDE one has from \cref{eq:PDEfAn} in this case is simply $f_{tt}=0$.\\
   
   On the other hand, the identity and Euler vector field are given by the following expressions:
   \begin{align*}
       e&= \partial_t,\\ E&= t\partial_t+\tfrac32\bigl(\alpha_1\partial_{\alpha_1}+\alpha_2\partial_{\alpha_2}\bigr)+\tfrac12\bigl(\beta_1\partial_{\beta_1}+\beta_2\partial_{\beta_2}\bigr)\,.
   \end{align*}
        Notice that the polynomials in $\bm{t}$ which serve as coefficients for the linear diagonal invariants $\Theta_0,\dots, \Theta_{\ell+2}$ in $\Fr_{\mathrm{int.}}$ only depend on the degree of the purely-polynomial contribution to the superpotential. Therefore, if, in the present case, one were to consider higher-degree deformations of the $\AN_1$-singularity, the additional summands in the corresponding interaction term would merely be due to the fact that the diagonal invariants in a higher number of variables contain a higher number of terms.\\
        
    In particular, we highlight that, in the single-pole case, the prepotential reduces to:
    \begin{align*}
    \Fr(t,\alpha,\beta)&= \tfrac{1}{12}t^3+\tfrac12 \alpha^2\log\alpha+\tfrac13 \alpha\beta^3+t\alpha\beta.        
    \end{align*}
    This is exactly the free-energy from \cite{MaZuoBn}*{Example II.4}, if we let $t=:t_2$, $\alpha=:t_3$ and $2\beta=:t_1$, or, equivalently, the prepotential for the Frobenius manifold structure associated to the Weyl group $W(\mathscr{B}_3)$ in the work \cite{Arsie_2022}. 
\end{ex}
\begin{ex}[$\AN_3$-singularity with two simple poles]
    Consider the Hurwitz space $\Hw_{0,6}(3,0,0)$ of rational functions of the form:
    \begin{align*}
    \lambda(w)&= w^4+\sigma_2w^2+\sigma_1w+\sigma_0+\frac{\alpha_1}{w-\beta_1}+\frac{\alpha_2}{w-\beta_2}\,,        
    \end{align*}
    with primary differential $\dd w$.
    
   The invariant polynomials $\{\sigma_0,\sigma_1,\sigma_2\}$ are given in terms of flat coordinates $\{t_1,t_2,t_3\}$ for the Frobenius manifold structure on the orbit space $\CC^3/W(\AN_3)$ as follows \cite{Dubrovin1999}*{Example 1.4}:
    \[
    \begin{aligned}
        \sigma_0&=t_1+\tfrac 18 t_3^2\,,&&&
        \sigma_1&=t_2\,,&&&
        \sigma_2&=t_3\,.
    \end{aligned}
    \]
    The free energy is, on the other hand:
    \begin{align*}
    \Fr_{\AN_3}(t_1,t_2,t_3) \quad &:=\quad \tfrac12 t_1t_2^2+\tfrac12 t_1^2t_3-\tfrac{1}{16}t_2^2t_3^2+\tfrac{1}{960}t_3^5\,.        
    \end{align*}
    
    As for $f(\bm{t})$, the system \cref{eq:PDEfAn} is $f_{t_2t_3}=\tfrac14$, while all the other second derivatives vanish. As a consequence, we may take $f(\bm{t})=\tfrac14 t_2t_3$. The Frobenius structure is, therefore, defined by:    
    \begin{align*}
        \Fr(\bm{t},\bm{\alpha},\bm{\beta})&= \Fr_{\AN_3}(\bm{t})+\tfrac12 \alpha_1^2\log\alpha_1+\tfrac12\alpha_2^2\log\alpha_2+\alpha_1\alpha_2\log(\beta_1-\beta_2)+\\
        &\qquad+\tfrac15 \bigl[\alpha_1\beta_1^5+\alpha_2\beta_2^5\bigr]+\tfrac13 t_3\bigl[\alpha_1\beta_1^3+\alpha_2\beta_2^3\bigr]+\tfrac12 t_2\bigl[\alpha_1\beta_1^2+\alpha_2\beta_2^2\bigr]+\\
        &\qquad +\bigl(t_1+\tfrac18 t_3^2\bigr)\bigl[\alpha_1\beta_1+\alpha_2\beta_2\bigr]+\tfrac14 t_2t_3\bigl[\alpha_1+\alpha_2\bigr]\,,\\
        e&= \partial_{\sigma_0}=\partial_{t_1}\,,\\
        E&= \sigma_0\partial_{\sigma_0}+\tfrac34 \sigma_1\partial_{\sigma_1}+\tfrac12\sigma_2\partial_{\sigma_2}+\tfrac54\bigl(\alpha_1\partial_{\alpha_1}+\alpha_2\partial_{\alpha_2}\bigr)+\tfrac14\bigl(\beta_1\partial_{\beta_1}+\beta_2\partial_{\beta_2}\bigr)\\
        &= t_1\partial_{t_1}+\tfrac34t_2\partial_{t_2}+\tfrac12 t_3\partial_{t_3}+\tfrac54\bigl(\alpha_1\partial_{\alpha_1}+\alpha_2\partial_{\alpha_2}\bigr)+\tfrac14\bigl(\beta_1\partial_{\beta_1}+\beta_2\partial_{\beta_2}\bigr)    \,.    
    \end{align*}
\end{ex}

\begin{ex}[$\AN_4$-singularity with two simple poles]
    Consider the Frobenius manifold structure on the Hurwitz space $\Hw_{0,7}(4,0,0)$ of equivalence classes of meromorphic functions on the Riemann sphere of the form:
    \begin{align*}
    \lambda(w)&= w^5+\sigma_3w^3+\sigma_2w^2+\sigma_1 w+\sigma_0+\frac{\alpha_1}{w-\beta_1}+\frac{\alpha_2}{w-\beta_2}\,,        
    \end{align*}
    with primary differential $\dd w$.
    
    The flat coordinates $\{t_1,t_2,t_3,t_4\}$ for the Frobenius manifold structure on the parameter space of miniversal deformations of the $\AN_4$-singularity and the symmetric polynomials $\{\sigma_0,\sigma_1,\sigma_2,\sigma_3\}$ in the roots are related as follows:
    \[
    \begin{aligned}
        \sigma_3&=t_4\,,&&& \sigma_2&=t_3\,,&&&\sigma_1&=t_2+\tfrac15 t_4^2\,,&&&\sigma_0&=t_1+\tfrac15 t_3t_4\,.
    \end{aligned}
    \]
    The corresponding prepotential is \cite{Dubrovin1996}*{Example 4.4}:
   \begin{align*}
         \Fr_{\AN_4}(\bm{t})&= \tfrac{1}{2}t_1^2t_4+t_1t_2t_3+\tfrac12 t_2^3+\tfrac13 t_3^4+6t_2t_3^2t_4+9t_2^2t_4^2+24t_3^2t_4^3+\tfrac{216}{5}t_4^6\,.   
   \end{align*}
   
    As for $f(\bm{t})$, the system \cref{eq:PDEfAn} reads:
    \[
    \begin{aligned}
        f_{t_3t_3}&=\tfrac15\,,&&& f_{t_2t_4}&=\tfrac15\,,&&& f_{t_4t_4}&=\tfrac{1}{25}t_4\,,
    \end{aligned}
    \]
    while all the other second derivatives vanish.
    Therefore, we can take $f(\bm{t})=\tfrac{1}{10}t_3^2+\tfrac15 t_2t_4+\tfrac{1}{150}t_4^3$.\\

    It follows that the Frobenius manifold structure on $\Hw_{0,7}(4,0,0)$ is defined by the following:
   \begin{align*}
     \Fr(\bm{t},\bm{\alpha},\bm{\beta})&= \Fr_{\AN_4}(\bm{t})+\tfrac12 \alpha_1^2\log\alpha_1+\tfrac12\alpha_2^2\log\alpha_2+\alpha_1\alpha_2\log(\beta_1-\beta_2)+\\&\qquad+\tfrac16\bigl[\alpha_1\beta_1^6+\alpha_2\beta_2^6\bigr]+\tfrac14 t_4\bigl[\alpha_1\beta_1^4+\alpha_2\beta_2^4\bigr] +\tfrac13 t_3\bigl[\alpha_1\beta_1^3+\alpha_2\beta_2^3\bigr]+\\&\qquad +\tfrac12 \bigl(t_2+\tfrac15 t_4^2\bigr)\bigl[\alpha_1\beta_1^2+\alpha_2\beta_2^2\bigr]+\bigl(t_1+\tfrac15 t_3t_4\bigr)\bigl[\alpha_1\beta_1+\alpha_2\beta_2\bigr]+\\
        &\qquad +\bigl(\tfrac{1}{10}t_3^2+\tfrac15 t_2t_4+\tfrac{1}{150}t_4^3\bigr)\bigl[\alpha_1+\alpha_2\bigr]\,,\\
        e&= \partial_{\sigma_0}=\partial_{t_1}\,,\\
        E&= \sigma_0\partial_{\sigma_0}+\tfrac45 \sigma_1\partial_{\sigma_1}+\tfrac35\sigma_2\partial_{\sigma_2}+\tfrac25\sigma_3\partial_{\sigma_3}+\tfrac65 \bigl(\alpha_1\partial_{\alpha_1}+\alpha_2\partial_{\alpha_2}\bigr)+\tfrac15\bigl(\beta_1\partial_{\beta_1}+\beta_2\partial_{\beta_2}\bigr)\\
        &= t_1\partial_{t_1}+\tfrac45 t_2\partial_{t_2}+\tfrac35 t_3\partial_{t_3}+\tfrac25 t_4\partial_{t_4}+\tfrac65 \bigl(\alpha_1\partial_{\alpha_1}+\alpha_2\partial_{\alpha_2}\bigr)+\tfrac15\bigl(\beta_1\partial_{\beta_1}+\beta_2\partial_{\beta_2}\bigr)\,.      
   \end{align*}
\end{ex}

\section{Dubrovin-Zhang Frobenius manifolds with tail of simple poles}
\documentclass[main.tex]{subfile}

\subsection{Trigonometric rational functions and flat coordinates}
As outlined upon in \cref{subsec:Saito}, there exists a unique Frobenius manifold structure on the space of orbits of the extension $\widetilde{W}^{(r)}(\AN_{\ell+r})$ of the affine-Weyl group $\widetilde{W}(\AN_{\ell+r})$ corresponding to the choice of the $r^{\mathrm{th}}$ root in the diagram, as defined by \cref{eq:DZEAWaction}. Furthermore, such a structure admits a B-model given by the Hurwitz space $\Hw_{0,\ell+r+1}(\ell,r-1)$, whose points are rational functions of the form $\lambda_{\ell+r,r}$ as defined in \cref{eq:EAWsuperpotential}, equipped with a primary differential $\omega$ being the third-kind Abelian differential $\omega$ having simple poles at $0$ and $\infty$ with residues $-1$ and $+1$ respectively -- i.e., in the coordinate system on $\RS$ where each point in the Hurwitz space can be uniquely written as in \cref{eq:EAWsuperpotential}, $\omega=-\tfrac1w \dd{w}$.
Hence the problem is very similar to that discussed above: we have a space of functions -- which, in this case, happen to be rational from the very beginning, even though the position of the only pole is prescribed -- exhibiting a natural Frobenius manifold structure coming from the action of a reflection group. \\

Now, flat coordinates for the intersection form on the orbit space are again given by coordinates with respect to the canonical basis of $\A^{\ell+1}$ -- which are only local coordinates away from the diagonals. On the open subset of Hurwitz
space consisting of rational functions whose zeros are pairwise distinct, however, flat coordinates for the intersection form are given as follows \cite{dubrovin_zhang_1998}*{Lemma 3.1}:
\begin{align*}
     \A^{\ell+1}\smallsetminus\Delta\ni (\varphi_1,\dots, \varphi_{\ell+r}) \quad &\mapsto\quad  \tfrac{1}{w^r}(w-e^{\varphi_1})\dots (w-e^{\varphi_{\ell+r}}).
\end{align*}
There is, indeed, a group “bigger” than $\Sym_{\ell+r}$ acting on $ \A^{\ell+1}\smallsetminus\Delta$ such that any two points on the same orbit determine the same rational function. Namely, this is the group $\widetilde{\Sym}_{\ell+r}:=\Sym_{\ell+r}\times\Z^{\ell+r}$, the latter acting by individual shifts of integral multiples of $i2\pi$ on each coordinate. This was expected as the monodromy group of such a Frobenius manifold is an extension of the Weyl group $W(\AN_\ell)$. We will probe the exact relation between the two in \cref{EAW:sec:monodromy}. Notice that, despite the group not being finite anymore, it is still finitely generated and the point stabilisers are again trivial. Therefore, the quotient can nonetheless be endowed with a complex manifold structure.\\

As for the flat coordinates $t_1,\dots, t_{\ell+r+1}$ for the Frobenius pairing, according to \cite{dubrovin_zhang_1998}*{eq. 3.25}, they are also given, as functions of the EAW invariant polynomials $y_1,\dots, y_{\ell+r+1}$ \cite{dubrovin_zhang_1998}*{eq. 1.10}, by the solutions to a lower-triangular, linear system with ones on the diagonal. Moreover, they can also be computed as the coefficients of some appropriate Laurent series expansions around $\infty$ and $0$, as highlighted in the proof of \cite{dubrovin_zhang_1998}*{Proposition 3.1}. In particular, this ensures that:
\begin{align*}
   t_1,\dots, t_{\ell+r+1}&\in \CC[y_1,\dots, y_{\ell+r+1}]. 
\end{align*}
Notice that, as discussed in \cite{dubrovin_zhang_1998}*{Theorem 1.1}, this is not quite the ring of all polynomial invariants with respect to the extended affine-Weyl group action -- but only consists of those satisfying some boundedness condition at infinity. As a consequence, its spectrum will be a proper subset of the extended affine-Weyl group orbit space.\\

It is again natural to ask a similar question with polynomials replaced by rational functions. In other words, we now set $n_z,n_p,r\in\Z_{\geq 0}$ such that $n_z>n_p+r$ (and $r>0$) and consider the space of rational functions:
\begin{align}\label{EAW:eq:lambdaextended}
    \lambda(w)&= \frac{1}{w^r}\,\frac{w^{n_z}+a_{n_z-1}w^{n_z-1}+\dots+a_0}{w^{n_p}+b_{n_p-1}w^{n_p-1}+\dots+b_0},
\end{align}
with complex coefficients. If the zeros of the denominator -- except for the one at $0$ -- are pairwise distinct, points in this space are going to be in a one-to-one correspondence with classes in $\Hw_{0,n_z }(n_z -(n_p +r)-1,r-1,\underbrace{0,\dots, 0}_{n_p  \text{ times.}})$. \\

We endow such a Hurwitz space with the very same third-kind Abelian differential $\omega$ having poles at $0$ and $\infty$ with residues $-1$ and $+1$ respectively. This is still an admissible primary differential, therefore it induces a Frobenius manifold structure on it in the usual way.

It turns out that flat coordinates for the intersection form \cref{eq:intformHurwitz} are the obvious generalisation of the ones we have in the Dubrovin-Zhang case:
\begin{lem}\label{lem:EAW:flatcoordsIntForm}
    We consider the Hurwitz space $\Hw^\omega_{0,n_z }(n_z -(n_p +r)-1,r-1,\bm{0})$. Coordinates $\{\varphi_1,\dots, \varphi_{n_z },\psi_1,\dots, \psi_{n_p }\}$, defined by:
   \begin{align}\label{EAW:eq:factorisation}
       \lambda(w)&= \frac{1}{w^r}\frac{(w-e^{\varphi_1})\dots (w-e^{\varphi_{n_z }})}{(w-e^{\psi_1})\dots (w-e^{\psi_{n_p }})},
   \end{align}
  are flat for its intersection form away from the diagonals where any two of its zeros coincide\footnote{If either two of the poles or a zero and a pole coincide, \cref{EAW:eq:factorisation} does not even describe a point in the Hurwitz space currently under consideration.}.
\end{lem}
\begin{proof}
    This is an easy generalisation of some of the computations in the proofs of \citelist{\cite{dubrovin_zhang_1998}*{Theorem 3.1} \cite{zuo2020frobenius}*{Theorem 5.1} \cite{ma2023frobenius}*{Proposition 2.4}}.
\end{proof}
Therefore, even in this case one has an action of the finitely-generated group $\widetilde{\Sym}_{n_z}\times\widetilde{\Sym}_{n_p}$ on $\A^{n_z+n_p}\smallsetminus\Delta$ through which the coordinate map from the previous Lemma factors. The point stabilisers are all clearly still trivial, therefore the quotient possesses a unique natural complex manifold structure.
\begin{oss}
    Notice that the coefficients $a_0,\dots, a_{n_z-1}$ and $b_0,\dots, b_{n_p-1}$ of each monomial in the numerator and denominator of \cref{EAW:eq:lambdaextended} are respectively the elementary symmetric functions in the zeros and poles of $\lambda$, i.e. they are the $W(\AN_{n_z-1})$ (resp. $W(\AN_{n_p-1})$) invariant basic Fourier polynomials \cite{dubrovin_zhang_1998}*{eq.s 1.4, 1.6}, which generate the ring of invariants with respect to the action of the corresponding affine-Weyl group \cites{dubrovin_zhang_1998, bourbaki2007groupes}.
\end{oss}

Now, to compare such a Frobenius manifold structure with the one on the orbit space of the EAW group, polynomial division ensures that we can uniquely write any point in the Hurwitz space with the following representation:
\begin{align}\label{EAW:eq:quotient}
    \lambda(w)&= w^{\ell+1}+\sigma_{\ell+r}w^{\ell} +\dots +\sigma_r+\tfrac1w \sigma_{r-1}+\dots+\tfrac{1}{w^r}\sigma_0+\frac{\rho_1}{w-e^{\beta_1}}+\dots+\frac{\rho_{n_p}}{w-e^{\beta_{n_p}}},
\end{align}
where $\ell:=n_z-(n_p+r)-1$.

As in the previous case, there is, however, still a residual action of a subgroup of $\widetilde{\Sym}_{n_z}\times\widetilde{\Sym}_{n_p}$ through which such coordinate map factors:
\begin{prop}
    The coefficients $\sigma_{r},\dots,\sigma_{\ell+r}$ of the positive powers of the unknown $w$ in \cref{EAW:eq:quotient} are symmetric polynomials in the zeros and poles of $\lambda$:
    \begin{align*}
 \sigma_{r},\dots,\sigma_{\ell+r}&\in \CC[e^{\varphi_1},\dots,e^{\varphi_{n_z}}]^{\Sym_{n_z}}\otimes\CC[e^{\psi_1},\dots, e^{\psi_{n_p}}]^{\Sym_{n_p}}.      
    \end{align*}
    
    On the other hand, the coefficients $\sigma_{0},\dots,\sigma_{r-1}$ are symmetric polynomials in the zeros of $\lambda$ and symmetric rational functions of its poles:
      \begin{align*}
\sigma_{0},\dots,\sigma_{r-1}&\in \CC[e^{\varphi_1},\dots,e^{\varphi_{n_z}}]^{\Sym_{n_z}}\otimes\CC[e^{\psi_1},\dots, e^{\psi_{n_p}},e^{-\psi_1},\dots, e^{-\psi_{n_p}}]^{\Sym_{n_p}}.         
      \end{align*}
The action of the symmetric group $\Sym_{n_p}$ on the latter generating set is the one induced by permutations of the $\psi$ coordinates (i.e. diagonal permutation of the first $n_p$ and the last $n_p$ variables).\\

Moreover, the coefficients $\rho_1,\dots, \rho_{n_p}$ are also symmetric polynomials in the zeros and are rational in the poles, but are permuted diagonally with the coefficients $\beta_1,\dots, \beta_{n_p}$ when acted upon by a permutation $\pi\in\Sym_{n_p}$, i.e. the coordinates $(\bm{\sigma},\bm{\rho},\bm{\beta})$ and $(\bm{\sigma},\pi.\bm{\rho},\pi.\bm{\beta})$ determine the same rational function. Finally, the infinite-cyclic group $\Z^{n_p}$ acts by shifts of each of the $\bm{\beta}$ coordinates by integral multiples of $i2\pi$.
\end{prop}
\begin{proof}
    For a fixed point $(\varphi_1,\dots, \varphi_{n_z },\psi_1,\dots, \psi_{n_p })\in\A^{n_z+n_p}\smallsetminus\Delta$, the coordinates $\beta_1,\dots, \beta_{n_p}$ clearly lie in the same $\widetilde{\Sym}_{n_p}$-orbit of $(\psi_1,\dots, \psi_{n_p })$. Without loss of generality, we set $\beta_\mu=\psi_\mu$ for any $\mu=1,\dots, n_p$. Consequently, $\rho_\mu$ is the residue of $\lambda$ at $e^{\psi_\mu}$:
    \[
    \rho_\mu=\lim_{w\to e^{\psi_\mu}}(w-e^{\psi_\mu})\,\lambda(w)=e^{-r\psi_\mu}\,\frac{(e^{\psi_\mu}-e^{\varphi_1})\dots (e^{\psi_\mu}-e^{\varphi_{n_z}})}{\prod_{\nu\neq\mu}(e^{\psi_\mu}-e^{\psi_\nu})}.
    \]
    This is a $\Sym_{n_z}$-invariant, exponential polynomial in $\varphi_1,\dots, \varphi_{n_z}$, whereas it is an exponential rational function of $\psi_1,\dots, \psi_{n_p}$. With respect to the latter dependence, it is clear that the transposition in $\Sym_{n_p}$ swapping $\psi_\mu$ with $\psi_\nu$ leaves all the $\rho$-s invariant but $\rho_\mu$ and $\rho_\nu$, which are mapped into one another, while it simultaneously swaps $\beta_\mu$ with $\beta_\nu$.\\

    For the $\sigma$ coefficients, the coordinate transformation can be computed by equating the two expressions for $\lambda$ from \cref{EAW:eq:lambdaextended} and \cref{EAW:eq:quotient}:
    \[
    \begin{aligned}
        w^{n_z}+a_{n_z-1}w^{n_z-1}+\dots+a_0&=\bigl(w^{n_p}+\dots+b_0\bigr)\bigl(w^{n_z-n_p}+\sigma_{\ell+r}w^{n_z-n_p-1} +\dots +\sigma_0\bigr)+\\
        &\quad+w^r(w-e^{\beta_1})\dots(w-e^{\beta_{n_p}})\bigl(\tfrac{\rho_1}{w-e^{\beta_1}}+\dots+\tfrac{\rho_{n_p}}{w-e^{\beta_{n_p}}}\bigr).
    \end{aligned}
    \]
    The degree of the polynomial on the second line is $n_p+r-1$; as a consequence, the coefficients of the monomials $w^{n_z},w^{n_z-1},\dots, w^{n_p+r}$ of the right-hand side will only come from the polynomial on the first line, and they are:
    \[
    \begin{matrix}
        [w^{n_z}]:&& 1,\\
        [w^{n_z-1}]:& & \sigma_{\ell+r}+b_{n_p-1},\\
        [w^{n_z-2}]:&&\sigma_{\ell+r-1}+\sigma_{\ell+r}b_{n_p-1}+b_{n_p-2},\\
        \vdots &&\vdots\\
        [w^{n_p+r}]:&& \sigma_{r}+\sigma_{r+1}b_{n_p-1}+\dots.
        \end{matrix}
    \]
    In particular, the coefficients of $w^k$ will be the sum of all the admissible monomials of the form $\sigma_{k_1}b_{k_2}$ for $k_1=r,\dots, \ell+r+1$ and $k_2=0,\dots, n_p$ such that $(k_1,k_2)$ is an ordered partition of $k$, with the prescription that $\sigma_{\ell+r+1}=b_{n_p}:=1$. Since, for any $k=n_p+r,\dots, n_z$ there is an admissible $k_1$ such that $k_1+n_p=k$, it follows that equating with the corresponding coefficients on the left-hand side gives rise to a closed, upper-triangular linear system for the unknowns $\sigma_r,\dots, \sigma_{\ell+r}$ with ones on the diagonals. Hence, these coefficients are polynomial in $a_0,\dots, a_{n_z-1},b_0,\dots, b_{n_p-1}$, therefore they are exponential symmetric polynomials in $\varphi_1,\dots, \varphi_{n_z }$ and $\psi_1,\dots, \psi_{n_p }$.
    
    Finally, since the polynomial on the second line in the right-hand side is divisible by $w^r$, it follows that the coefficients of the monomials $w^0,\dots, w^{r-1}$ will only be given by the ones of the polynomial on the first line. These will again be the sum of monomials of the form $\sigma_{k_1}b_{k_2}$ for $(k_1,k_2)$ an ordered partition of $k$. This ensures that we are going to get a square, upper-triangular linear system for the remaining unknwons $\sigma_0,\dots, \sigma_{r-1}$ with, however, $b_0$ on the diagonal. Since this is not a constant, the unique solution to the system will inevitably be rational in the $b$ variables. Since, on the contrary, the diagonal entries do not depend on any of the $a$-s, the solution will still only depend polynomially on them.
\end{proof}
Now, the main difference between this case and the previous one is that we have an orbit space description for \emph{some} of the Hurwitz spaces we are now considering. Namely, as anticipated, we do know from \cite{zuo2020frobenius} that the Frobenius manifold structure on the Hurwitz spaces $\Hw_{0,\ell+r+2}(\ell,r-1,0)$ with primary differential $\omega$ -- i.e. all the cases with a single additional simple pole $n_p=1$ -- is locally isomorphic to the on the orbit space of the extended affine-Weyl group $\widetilde{W}(\AN_{\ell+r})^{(r,r+1)}$ with two marked roots, as defined in \cref{eq:MZEAW}. Notice that this is a very special case as the permutation group of the poles is trivial, and the additional subgroup of the monodromy group is $\widetilde{\Sym}_1\cong\Z$ generated by a shift in the logarithm of the finite pole by $i2\pi$.\\

We next turn to considering functions on the Hurwitz space. In particular, our ultimate goal is to write down the prepotential for the Frobenius manifold structure and possibly relate it to the one on the corresponding Dubrovin-Zhang orbit space, as described by \cref{EAW:eq:quotient}. In order to do this, we require a system of flat coordinates for the Frobenius pairing, which we have not discussed yet. In the previous case we were considering, these were simply given by a completion of the appropriate Saito polynomials with the position of the finite simple poles and their residues \cref{eq:rationalquotient}. 

Now, one can still complete the Dubrovin-Zhang flat coordinates to a set of flat coordinates for the larger Hurwitz space, however the additional coordinates will not be exactly the same as before:
\begin{lem}[Flat coordinates for the Frobenius pairing]\label{lem:flatEAWdeform}
    Let $\ell ,r,n_p \in\Z_{\geq 1}$ and consider the Hurwitz spaces $\Hw_{0,\ell+r+1}(\ell,r-1)$ and $\Hw_{0,\ell+r+n_p +1}(\ell,r-1,\underbrace{0,\dots,0}_{n_p  \text{ times}})$, equipped with a primary differential $\omega$ being the third-kind Abelian differential with simple poles at $0$ and $\infty$ and residues $-1$ and $+1$ respectively.
    
    If $\bm{t}=(t_1,\dots, t_{\ell+r+1})$ is the set of flat coordinates for the Frobenius pairing on the former given in \cite{dubrovin_zhang_1998}*{Proposition 3.1}, then the coordinates $(\bm{t},\alpha_1,\dots, \alpha_{n_p },\beta_1,\dots, \beta_{n_p })$, defined by
    \begin{align}\label{eq::EAWdeformation}
  \lambda(w)&= w^{\ell+1}+w^\ell\sigma_{\ell+r}(\bm{t})+\dots+\sigma_r(\bm{t})+\tfrac1w \sigma_{r-1}(\bm{t})+\dots+\tfrac{1}{w^r}\sigma_0(\bm{t})+\\\notag&\qquad +\frac{w\alpha_1}{w-e^{\beta_1}}+\dots+\frac{w\alpha_{n_p }}{w-e^{\beta_{n_p }}},              
    \end{align}
 are flat for the Frobenius pairing on the latter.
\end{lem}
\begin{proof}
    Again, this follows from explicit calculation of the Frobenius metric in the given coordinate system. The result is that the matrix representation splits in the following block-diagonal form:
    \[
    \eta=\eta_{\AN_{\ell+r}^{(r)}}\oplus\begin{bmatrix}
        0 & \id_{n_p }\\\id_{n_p }&0
    \end{bmatrix},
    \]
    which proves the statement. For instance, we again provide the explicit computation for the uppermost diagonal block in the coordinates $(\bm{\sigma},\bm{\alpha},\bm{\beta})$:
    \[
    \begin{aligned}
        \eta_{\sigma_p\sigma_q}&=\sum_{x\in\Gamma_\lambda}\Res_x\biggl\{\frac{w^{p+q-2r}}{\lambda_{\ell+r,r}'(w)-\sum_\mu\tfrac{\alpha_\mu e^{\beta_\mu}}{(w-e^{\beta_\mu})^2}}\frac{\dd{w}}{w^2}\biggr\}\\
        &=-\Res_\infty\biggl\{\frac{w^{p+q-2(r+1)}}{\lambda_{\ell+r,r}'-\sum_\mu\tfrac{\alpha_\mu e^{\beta_\mu}}{(w-e^{\beta_\mu})^2}}\dd{w}\biggr\}-\Res_0\biggl\{\frac{w^{p+q-2(r+1)}}{\lambda_{\ell+r,r}'-\sum_\mu\tfrac{\alpha_\mu e^{\beta_\mu}}{(w-e^{\beta_\mu})^2}}\dd{w}\biggr\}\,.
    \end{aligned}
    \]
    The residue at $\infty$ is:
    \[
    \begin{aligned}
        \Res_\infty\biggl\{\frac{w^{p+q-2(r+1)}}{\lambda_{\ell+r,r}'-\sum_\mu\tfrac{\alpha_\mu e^{\beta_\mu}}{(w-e^{\beta_\mu})^2}}\dd{w}\biggr\}&=-\Res_0\biggl\{\frac{w^{2r-(p+q)}}{\lambda_{\ell+r,r}'(\tfrac1w)-w^2\sum_\mu\tfrac{\alpha_\mu e^{\beta_\mu}}{(1-we^{\beta_\mu})^2}}\dd{w}\biggr\}\,.
    \end{aligned}
    \]
    Now, since $\lambda_{\ell+r,r}'(\tfrac1w)$ has a pole of order $\ell$ at $0$, it follows that $Q(w):=w^\ell \lambda_{\ell+r,r}'(\tfrac1w)$ is holomorphic at $0$. As a consequence:
    \[
    \begin{aligned}
        \Res_\infty\biggl\{\frac{w^{p+q-2(r+1)}}{\lambda_{\ell+r,r}'-\sum_\mu\tfrac{\alpha_\mu e^{\beta_\mu}}{(w-e^{\beta_\mu})^2}}\dd{w}\biggr\}&=\Res_0\biggl\{\frac{w^{\ell+2r-(p+q)}}{Q(w)-w^{\ell+2}\sum_\mu\tfrac{\alpha_\mu e^{\beta_\mu}}{(1-we^{\beta_\mu})^2}}\dd{w}\biggr\}\\
        &=\Res_0\biggl\{\frac{w^{\ell+2r-(p+q)}}{Q(w)}\bigl[1+\order{w^{\ell+2}}\bigr]\dd{w}\biggr\}\,.
    \end{aligned}
    \]
    Since $p+q\leq\ell+r$, $2(\ell+r)+2-(p+q)\geq2>0$. As a consequence, only the first summand in the series expansion might be singular at zero:
    \[
    \Res_\infty\biggl\{\frac{w^{p+q-2(r+1)}}{\lambda_{\ell+r,r}'(w)-\sum_\mu\tfrac{\alpha_\mu e^{\beta_\mu}}{(w-e^{\beta_\mu})^2}}\dd{w}\biggr\}=\Res_\infty\biggl\{\frac{w^{p+q-2r}}{\lambda'_{\ell+r,r}(w)}\frac{\dd{w}}{w^2}\biggr\}\,.
    \]
    Similarly, for the residue at zero we consider $P(w):=w^{r+1}\lambda_{\ell+r,r}'(w)$, which is holomorphic at zero, so that:
    \[
    \begin{aligned}
\Res_0\biggl\{\frac{w^{p+q-2(r+1)}}{\lambda_{\ell+r,r}'-\sum_\mu\tfrac{\alpha_\mu e^{\beta_\mu}}{(w-e^{\beta_\mu})^2}}\dd{w}\biggr\}&=\Res_0\biggl\{\frac{w^{p+q-(r+1)}}{P(w)-w^{r+1}\sum_\mu\tfrac{\alpha_\mu e^{\beta_\mu}}{(w-e^{\beta_\mu})^2}}\dd{w}\biggr\} \\
&=\Res_0\biggl\{\frac{w^{p+q-(r+1)}}{P(w)}\bigl[1+\order{w^{r+1}}\bigr]\dd{w}\biggr\}\\
&=\Res_0\biggl\{\frac{w^{p+q-2r}}{\lambda'_{\ell+r,r}(w)}\frac{\dd{w}}{w^2}\biggr\}\,.
\end{aligned}
    \]
    This proves that:
   \[
    \begin{aligned}
        \eta_{\sigma_p\sigma_q}&=-\Res_\infty\biggl\{\frac{w^{p+q-2r}}{\lambda'_{\ell+r,r}(w)}\frac{\dd{w}}{w^2}\biggr\}-\Res_0\biggl\{\frac{w^{p+q-2r}}{\lambda'_{\ell+r,r}(w)}\frac{\dd{w}}{w^2}\biggr\}=\bigl(\eta_{\AN_{\ell+r}^{(r)}}\bigr)_{\sigma_p\sigma_q}\,.
   \end{aligned}
    \]
\end{proof}
Despite the additional flat coordinates not being quite the finite, non-zero poles and their residues, they will still exhibit a similar behaviour under the action of the extended symmetric group $\widetilde{\Sym}_{n_p}$. As a matter of fact, the relation between the residue coordinates $\rho_1,\dots,\rho_{n_p}$ and $\alpha_1,\dots, \alpha_{n_p}$ is simply $\alpha_\mu=e^{-\beta_\mu}\rho_\mu$ for any $\mu=1,\dots, n_p$. Consequently, $\Sym_{n_p}$ still acts diagonally on the $2n_p$-tuple $(\bm{\alpha},\bm{\beta})$ in such a way that the coordinate map factors through it, and the same applies to the action of $\Z^{n_p}$ shifting the $\bm{\beta}$-coordinates while leaving the $\bm{\alpha}$-s invariant.\\ 

Keeping this in mind, then, one could say that a function of the flat coordinates on the Hurwitz space we are considering can be written as a function of the coordinates $(\bm{t},\bm{\alpha},\bm{\beta})$ (on some open subset), which is invariant under the diagonal action of the finitely-generated group $\widetilde{\Sym}_{n_p}$. If we were to only focus on polynomial functions, then there would not be much more to say; the requirement that the corresponding polynomial be invariant under individual shift in any of the $\bm{\beta}$-coordinates basically forces it not to depend on $\bm{\beta}$ at all. Just as in the Dubrovin-Zhang case, the presence of the translations suggests that we rather ought to be looking at exponential functions in the $\bm{\beta}$ variables.

In particular, if we introduce the shorthand $e^{\bm{\beta}}:=(e^{\beta_{1}},\dots, e^{\beta_{n_p}})$, then the most general translation-invariant algebraic function in these variables will be an element of $\CC[\bm{t},\bm{\alpha},e^{\bm{\beta}},e^{-\bm{\beta}}]$. In this ring, we are, then, only left to realise the invariance with respect to the diagonal action of the symmetric group $\Sym_{n_p}$ simultaneously permuting the $\bm{\alpha}$ and $\bm{\beta}$ coordinates. As an easy generalisation of \cref{prop:polarisedpowersums}, we have:
\begin{prop}\label{prop:exppowersums}
    Let $R$ be a ring. The ring of invariants $R[\bm{x},e^{\bm{y}},e^{-\bm{y}}]^{\Sym_n}$ with respect to the diagonal action of the group $\Sym_n$ by concurrent permutation of the $2n$ variables $\bm{x}=(x_1,\dots, x_n)$ and $\bm{y}=(y_1,\dots, y_n)$ is generated (over $R$) by the \emph{exponential power sums}:
    \begin{align}
        \widetilde{\mathscr{P}}_{p,q}(\bm{x},\bm{y})&= x_1^pe^{qy_1}+\dots+x_n^pe^{qy_n},\qquad (p,q)\in\Z_{\geq 0}\times\Z.
    \end{align}
\end{prop}
\begin{proof}
Notice that $R[\bm{x},e^{\bm{y}},e^{-\bm{y}}]\cong R[\bm{X},\bm{Y},\bm{Z}]\,\,\slash < Y_1Z_1-1,\dots, Y_nZ_n-1>$. The resulting action of the symmetric group $\Sym_n$ is again by diagonal permutation in the three sets of variables $\bm{X},\bm{Y},\bm{Z}$, therefore the invariant ring is generated by the polarised power sums $X_1^pY_1^qZ_1^s+\dots+X_n^pY_n^qZ_n^s$ for $p,q,s$ positive integers. The set of relations is closed under such action, therefore one gets a generating set for the quotient ring of invariants simply by applying the relations to the generating set for the original invariant ring. Clearly, this amounts to allowing negative powers of the $\bm{Y}$ (or, equivalently, of the $\bm{Z}$) variables.
\end{proof}

One can also prove a slight modification of \cref{prop:diaginvariantspoly} to the case of exponential polarised power sums:
\begin{prop}
      For any $k\in\Z_{\geq 0}$, let us denote: \[
    \widetilde{\Theta}_k(\bm{\alpha},\bm{\beta})=\widetilde{\mathscr{P}}_{1,k}(\bm{\alpha},\bm{\beta})\equiv \alpha_1e^{k \beta_1}+\dots +\alpha_{n_p}e^{k \beta_{n_p}}.
    \] 
   Then, as a function of $\bm{\varphi},\bm{\psi}$ as in \cref{lem:EAW:flatcoordsIntForm}, $e^{(r+1)(\psi_1+\dots+\psi_{n_p})}\widetilde{\Theta}_k$ is a symmetric polynomial of the coordinates $e^{\bm{\varphi}}$, $e^{\bm{\psi}}$:
    \begin{align*}
  e^{(r+1)(\psi_1+\dots+\psi_{n_p})}\,\widetilde{\Theta}_k(\bm{\varphi},\bm{\psi})&\in \CC[e^{\varphi_1},\dots,e^{\varphi_{n_z}}]^{\Sym_{n_z}}\otimes\CC[e^{\psi_1},\dots, e^{\psi_{n_p}}]^{\Sym_{n_p}}. 
    \end{align*}
    In particular, it is a homogeneous polynomial of degree $n_z-n_p+r+k+2$ in such coordinates.
\end{prop}
\begin{proof}
    Since, as already noted, $\alpha_\mu=e^{-\psi_\mu}\rho_\mu$ for $\mu=1,\dots, n_p$, it follows, from very similar computations as the ones in the proof of \cref{prop:diaginvariantspoly}, that the following holds:
    \[
    \widetilde{\Theta}_k(\bm{\varphi},\bm{\psi})=\frac{e^{-(r+1)(\psi_1+\dots+\psi_{n_p})}}{V_{n_p}(e^{\bm{\psi}})}\sum_{\rho=1}^{n_p}(-1)^{\rho-1}\,N(e^{\psi_\rho})e^{k\psi_\rho}\biggl[\prod_{\mu\neq \rho}e^{(r+1)\psi_\mu}\biggr]\biggl[\prod_{\substack{\mu<\nu\\\mu,\nu\neq \rho}}(e^{\psi_\mu}-e^{\psi_\nu})\biggr].
    \]
    It is very easy to then see that the numerator of $e^{(r+1)(\psi_1+\dots+\psi_{n_p})}\widetilde{\Theta}_k$ is again divisible by the Vandermonde determinant and that it picks up a minus sign when acted upon by any permutation of the pole variables. This proves that the first part of the proposition.
    
    The degree of homogeneity is also proven very similarly, as the residues exhibit the same behaviour as homogeneous rational functions in the poles.
\end{proof}
\begin{oss}
    Notice that homogeneity in the exponential variables actually amounts to translations in their logarithms.
\end{oss}

On the other hand, however, one should not forget that the prepotential is actually defined up to polynomials of degree at most two in the flat coordinates. Taking this into account, in fact, forces one to consider an extension of the ring of exponential diagonal invariants, as \emph{some} polynomials in $\bm{\beta}$ will be invariant under translation up to quadratics, despite not exhibiting the same property when considered as bare polynomials. As a matter of fact, one has $(x+1)^\nu\equiv x^\nu$ up to quadratics for $\nu=0,1,2,3$. As a consequence, when one quotients out $\CC[\bm{t},\bm{\alpha},\bm{\beta}]$ by the vector subspace of polynomials of degree at most two and looks at translation-invariant polynomials under individual shifts in each of the $\bm{\beta}$-coordinates in the quotient vector space, one gets a strictly bigger space than $\CC[\bm{t},\bm{\alpha}]^{(\geq 3)}$. Namely, any degree-three polynomial -- even the ones containing a positive power of one of the $\bm{\beta}$ variables -- will not change when acted upon by a shift. 

Insisting on diagonal invariance, then, leads one to consider $\CC[\bm{t}][\bm{\alpha}]^{\Sym_{n_p}}/\CC[\bm{t},\bm{\alpha}]^{(\leq 2)}$ on the one hand, whereas, on the other, it constraints the dependence on the $\bm{\alpha}$ and $\bm{\beta}$ variables of the degree-three polynomial to be encoded into an ordinary polarised power sum as the ones from \cref{prop:polarisedpowersums}. In particular, if $n$ is the number of $\bm{t}$ variables, the latter is going to be the complex vector space generated by  the following $3(n+1)$ polynomials:
\begin{equation}
   \begin{matrix}
         t_1\,\mathscr{P}_{1,1}(\bm{\alpha},\bm{\beta}),&&&\dots,&&&t_n\, \mathscr{P}_{1,1}(\bm{\alpha},\bm{\beta}),\\
       t_1\,\mathscr{P}_{0,2}(\bm{\alpha},\bm{\beta}),&&& \dots,&&& t_n\,\mathscr{P}_{0,2}(\bm{\alpha},\bm{\beta}),\\
       t_1^2\,\mathscr{P}_{0,1}(\bm{\alpha},\bm{\beta}),&&&\dots, &&&t_n^2\,\mathscr{P}_{0,1}(\bm{\alpha},\bm{\beta}),\\
       \mathscr{P}_{2,1}(\bm{\alpha},\bm{\beta}), &&& \mathscr{P}_{1,2}(\bm{\alpha},\bm{\beta}), &&& \mathscr{P}_{0,3}(\bm{\alpha},\bm{\beta}).
   \end{matrix}
     \end{equation}
We shall denote this by $\widetilde{\mathfrak{I}}^{(3)}$.\\

Thus we obtain:
\begin{prop}
    We have the following decomposition as vector spaces:
    \begin{align*}
  \bigl(\CC[\bm{t}][\bm{\alpha},\bm{\beta}]/\CC[\bm{t},\bm{\alpha},\bm{\beta}]^{(\leq 2)}\bigr)^{\widetilde{\Sym}_{n_p}}\cong \bigl(\CC[\bm{t}][\bm{\alpha}]^{\Sym_{n_p}}/\CC[\bm{t},\bm{\alpha}]^{(\leq 2)}\bigr)\oplus \widetilde{\mathfrak{I}}^{(3)}.
    \end{align*}
\end{prop}

\begin{oss}
    If one does not want to restrict to polynomial functions, but still wants to look at the quotient of the space of, say, holomorphic functions with respect to polynomials of degree smaller than two, then the first summand will be changed to invariant functions of $(\bm{t},\bm{\alpha})$ under the action of the symmetric group permuting the latter set of variables modulo polynomials of degree smaller than two, whereas $\widetilde{\mathfrak{I}}^{(3)}$ will be unaltered.
\end{oss}

\subsection{Explicit description of the Frobenius manifold structure}
We start with the following remark: as claimed in \cite{Arsie_2022}*{Remark 4.1}, one can get the prepotentials there constructed from the Weyl groups $W(\mathscr{B}_2)$, $W(\mathscr{B}_3)$ and $W(\mathscr{B}_4)$ by applying a Legendre transformation to the orbit space of a proper Dubrovin-Zhang extended affine-Weyl group of type $\AN$. This is clear from the Hurwitz space point of view. The B-model for the Arsie-Lorenzoni-Mencattini-Moroni construction on the space of orbits of $W(\mathscr{B}_{\ell+2})$ is, in fact, $\Hw_{0,\ell+2}(\ell,0)$ with primary form being the second-kind Abelian differential $\phi$ on $\RS$ having a pole at $\infty$ \cite{MaZuoBn}. In the coordinates such that any point in the Hurwitz space is represented as in \cref{eq:EAWsuperpotential} (with $r=1$), it is $\phi=\dd{w}$. It is, then, clear that, when one Legendre-transforms the Frobenius manifold structure in the direction of the residue at the finite simple pole, the primary differential is changed to $\omega=-\tfrac1w \dd{w}$. As just discussed, on the other hand, $\Hw^{\omega}_{0,\ell+2}(\ell,0)$ is a B-model for the Frobenius manifold on the orbit space of the extended affine-Weyl group $\widetilde{W}^{(1)}(\AN_{\ell+1})$.

This of course holds true in more generality in the following sense: starting from the Hurwitz space $\Hw^\phi_{0,\ell+r+1}(\ell,r-1)$ and applying a Legendre transformation in the direction of the residue at the order-$r$ pole, one gets the Frobenius manifold structure on the space of orbits of $\widetilde{W}^{(r)}(\AN_{\ell+r})$.\\

We are, now, ready to state the following result:
\begin{thm}[Expression for the prepotential]\label{thm:prepotentialDZ}
In the coordinates $(\bm{t},\bm{\alpha},\bm{\beta})$ from \cref{lem:flatEAWdeform}, the prepotential for the Frobenius manifold $\Hw^\omega_{0,\ell+r+n_p +1}(\ell,r-1,\bm{0})$ can be written as:
\begin{align}
    \Fr(\bm{t},\bm{\alpha},\bm{\beta})&= \Fr_{\AN_{\ell+r}^{(r)}}(\bm{t})+\Fr_{\mathrm{tail}}(\bm{\alpha},\bm{\beta})+\Fr_{\mathrm{int.}}(\bm{t},\bm{\alpha},\bm{\beta}),
\end{align}
where:
\begin{itemize}
    \item $\Fr_{\AN_{\ell+r}^{(r)}}(\bm{t})\in\CC[\bm{t},e^{t_\bullet}]$ is the prepotential on the orbit space of the Dubrovin-Zhang EAW group $\widetilde{W}^{(r)}(\AN_{\ell+r})$.
    \item $\Fr_{\mathrm{tail}}(\bm{\alpha},\bm{\beta})$ encodes the $\widetilde{\Sym}_{n_p }$-invariant non-polynomial part of the free energy:
    \begin{align}\label{eq:tailtermEAW}
        \Fr_{\mathrm{tail}}(\bm{\alpha},\bm{\beta})&= \tfrac12\sum_{\mu=1}^{n_p }\alpha_\mu^2\log\alpha_\mu+\sum_{1\leq \mu<\nu\leq n_p }\alpha_\mu\alpha_\nu\log(e^{\beta_\mu}-e^{\beta_\nu}).
    \end{align}
    \item $\Fr_{\mathrm{int.}}(\bm{t},\bm{\alpha},\bm{\beta})\in\bigoplus_{k=-r}^{\ell+1}\CC[\bm{t}]\,\widetilde{\Theta}_k(\bm{\alpha},\bm{\beta})\oplus \widetilde{\mathfrak{I}}^{(3)}\subseteq\CC[\bm{t}][\bm{\alpha},e^{\bm{\beta}},e^{-\bm{\beta}}]^{\widetilde{\Sym}_{n_p}}\oplus\widetilde{\mathfrak{I}}^{(3)} $ is the sum of an invariant exponential and a degree-three invariant polynomial.\\
    Explicitly:
    \begin{align}\label{EAW:eq:Fint}
    \begin{split}
            \Fr_{\mathrm{int.}}(\bm{t},\bm{\alpha},\bm{\beta})&= \tfrac{1}{\ell+1}\widetilde{\Theta}_{\ell+1}(\bm{\alpha},\bm{\beta})+\tfrac{1}{\ell}\sigma_{\ell+r}(\bm{t})\,\widetilde{\Theta}_\ell(\bm{\alpha},\bm{\beta})+\dots+\sigma_{r+1}(\bm{t})\,\widetilde{\Theta}_1(\bm{\alpha},\bm{\beta})+\\
         &\qquad+f(\bm{t})\,\widetilde{\Theta}_0(\bm{\alpha},\bm{\beta})+\tfrac{1}{-1}\sigma_{r-1}(\bm{t})\,\widetilde{\Theta}_{-1}(\bm{\alpha},\bm{\beta})+\dots+\tfrac{1}{-r}\sigma_0(\bm{t})\,\widetilde{\Theta}_{-r}(\bm{\alpha},\bm{\beta})+\\
          &\qquad+\sigma_r(\bm{t})\,\Theta_1(\bm{\alpha},\bm{\beta})+\tfrac12\mathscr{P}_{2,1}(\bm{\alpha},\bm{\beta})\,,   
    \end{split}   
    \end{align}
    where $f(\bm{t})\in\CC[\bm{t},e^{t_\bullet}]/\CC[\bm{t}]^{(\leq 1)}$ is a solution to the system of PDEs:
     \begin{align}\label{EAW:eq:systemf}
                 \pdv[2]{f}{t_a}{t_b}&= \sum_{p,q=0}^{\ell+r}\pdv{\sigma_p}{t_a}\pdv{\sigma_q}{t_b}\Res_0\biggl\{w^{\ell+2r-(p+q)}\frac{\dd w}{Q(w)}\biggr\}\,,\qquad 1\leq a\leq b\leq \ell+r+1\,.
    \end{align}
    Here, $t_{\ell+r+1}:=t_\bullet$ and $Q(w):=w^\ell\lambda'_{\ell+r,r}\bigl(\tfrac{1}{w}\bigr)\in\CC[\bm{t},e^{t_\bullet}][w]$ is the polynomial controlling the singular behaviour of $\lambda'_{\ell+r,r}$ at $\infty$, where $'\equiv\partial_w$.
\end{itemize}   
\end{thm}

\begin{proof}
    We mimic the proof of \cref{thm:prepotAN}, with the three-point function components coming from \cref{eq::EAWdeformation}. Hence, we start by computing the following three-point function components in the coordinates $(\bm{\sigma},\bm{\alpha},\bm{\beta})$:
    \[
\begin{aligned}
    c_{\sigma_p\sigma_q\sigma_s}&=\sum_{x\in\Gamma_\lambda}\Res_x\biggl\{\frac{w^{p+q+s-3r}}{\lambda'_{\ell+r,r}(w)-\sum_\mu\tfrac{\alpha_\mu e^{\beta_\mu}}{(w-e^{\beta_\mu})^2}}\frac{\dd{w}}{w^2}\biggr\}\\
    &=-\Res_\infty\biggl\{\frac{w^{p+q+s-3r}}{\lambda'_{\ell+r,r}-\sum_\mu\tfrac{\alpha_\mu e^{\beta_\mu}}{(w-e^{\beta_\mu})^2}}\frac{\dd{w}}{w^2}\biggr\}-\Res_0\biggl\{\frac{w^{p+q+s-3r}}{\lambda'_{\ell+r,r}-\sum_\mu\tfrac{\alpha_\mu e^{\beta_\mu}}{(w-e^{\beta_\mu})^2}}\frac{\dd{w}}{w^2}\biggr\}\\
    &\equiv R_\infty+R_0\,.
\end{aligned}
    \]
    For the residue at $\infty$ we let, like in the proof of \cref{lem:flatEAWdeform}, $Q(w):=w^\ell\lambda'_{\ell+r,r}\bigl(\tfrac{1}{w}\bigr)\,$, so that:
    \[
    \begin{aligned}
        R_\infty&=\Res_0\biggl\{\frac{1}{\lambda'_{\ell+r,r}(\tfrac{1}{w})-w^2\sum_\mu\tfrac{\alpha_\mu e^{\beta_\mu}}{(1-we^{\beta_\mu})^2}}\frac{\dd{w}}{w^{p+q+s-3r}}\biggr\}\\
        &=\Res_0\biggl\{\frac{w^{\ell+3r-(p+q+s)}}{Q(w)-w^{\ell+2}\sum_\mu\tfrac{\alpha_\mu e^{\beta_\mu}}{(1-w e^{\beta_\mu})^2}}\dd{w}\biggr\}\\
        &=\Res_0\biggl\{\frac{w^{\ell+3r-(p+q+s)}}{Q(w)}\biggl[1+\frac{w^{\ell+2}}{Q(w)}\sum_\mu\frac{\alpha_\mu e^{\beta_\mu}}{(1-we^{\beta_\mu})^2}+\order{w^{2(\ell+2)}}\biggr]\dd{w}\biggr\}\\
        &=-\Res_\infty\biggl\{\frac{w^{p+q+s-3r}}{\lambda'_{\ell+r,r}(w)}\frac{\dd{w}}{w^2}\biggr\}+\\&\quad+\sum_{\mu=1}^{n_p}\alpha_\mu\sum_{n\geq 0}(n+1)e^{(n+1)\beta_\mu}\Res_0\biggl\{w^{2\ell+3r+n+2-(p+q+s)}\frac{\dd{w}}{Q(w)^2}\biggr\}\,,
    \end{aligned}
    \]
    since $3(\ell+r)+2-(p+q+s)>2>0$. 
    
    As for the residue at zero, we let $P(w):=w^{r+1}\lambda'_{\ell+r,r}(w)$, so that:
    \[
    \begin{aligned}
        R_0&=-\Res_0\biggl\{\frac{w^{p+q+s-(2r+1)}}{P(w)-w^{r+1}\sum_\mu\tfrac{\alpha_\mu e^{\beta_\mu}}{(w-e^{\beta_\mu})^2}}\dd{w}\biggr\}\\
        &=-\Res_0\biggl\{\frac{w^{p+q+s-(2r+1)}}{P(w)}\biggl[1+\frac{w^{r+1}}{P(w)}\sum_{\mu}\frac{\alpha_\mu e^{\beta_\mu}}{(w-e^{\beta_\mu})^2}+\order{w^{2(r+1)}}\biggr]\dd{w}\biggr\}\\
        &=-\Res_0\biggl\{\frac{w^{p+q+s-3r}}{\lambda'_{\ell+r,r}(w)}\frac{\dd{w}}{w^2}\biggr\}-\sum_\mu\alpha_\mu\sum_{n\geq 0}(n+1)e^{-(n+2)\beta_\mu}\Res_0\biggl\{w^{p+q+s+n-r}\frac{\dd{w}}{P(w)^2}\biggr\}\,.
    \end{aligned}
    \]
    This proves that there are coefficients $k_{n pqs}(\bm{t})\in\CC[\bm{t},e^{t_\bullet}]$ such that, for any fixed choice of $p,q,s$, only finitely many of the polynomials $k_{npqs}$ are non-zero, and such that:
    \[
c_{\sigma_p\sigma_q\sigma_s}=R_\infty+R_0=c^{\widetilde{W}^{(r)}(\AN_{\ell+r})}_{\sigma_p\sigma_q\sigma_s}+\sum_{n\in\Z}k_{n pqs}(\bm{t})\,\widetilde{\Theta}_n(\bm{\alpha},\bm{\beta})\, .
    \]
 This, then, gives the following PDEs for $\Fr$ in terms of $\Fr_{\AN_{\ell+r}^{(r)}}$:
\[
\sum_{a,b,c=1}^{\ell+r+1}\pdv{t_a}{\sigma_p}\pdv{t_b}{\sigma_q}\pdv{t_c}{\sigma_s}\frac{\partial^3\bigl(\Fr-\Fr_{\AN_{\ell+r}^{(r)}}\bigr)}{\partial t_a\partial t_b\partial t_c}=\sum_{n\in\Z}k_{n pqs}(\bm{t})\,\widetilde{\Theta}_n(\bm{\alpha},\bm{\beta}).
\]
Since the Dubrovin-Zhang prepotential is independent of $\bm{\alpha}$ and $\bm{\beta}$, a solution to the system will be given by $\Fr=\Fr_{\AN_{\ell+r}^{(r)}}+\Phi$ for some $\Phi(\bm{t}, \bm{\alpha},\bm{\beta})$ whose third-derivatives with respect to the $\bm{t}$ variables are polynomial in $\bm{t},e^{t_\bullet},\bm{\alpha}$ and rational in $e^{\bm{\beta}}$. 

We fix the tail term by looking at the following components:
\[
\begin{aligned}
    c_{\alpha_\mu\alpha_\nu\alpha_\rho}&=\delta_{\mu\nu}\delta_{\nu\rho}\tfrac{1}{\alpha_\mu}, &&& c_{\alpha_\mu\alpha_\nu\beta_\rho}&=\frac{1}{e^{\beta_\mu}-e^{\beta_\nu}}\bigl(e^{\beta_\mu}\delta_{\mu\rho}-e^{\beta_\nu}\delta_{\nu\rho}\bigr).
\end{aligned}
\]
Integrating them up leads precisely to \cref{eq:tailtermEAW} plus the term $\tfrac12\mathscr{P}_{2,1}$ from the interaction contribution, up to some function $\Psi$ whose third derivatives $\Psi_{\alpha_\mu\alpha_\nu\alpha_\rho}$ and $\Psi_{\alpha_\mu\alpha_\nu\beta_\rho}$ all vanish. The difference from the previous case here come from the fact that the components $c_{\alpha_\mu\alpha_\mu\beta_\mu}$ are no longer zero.\\

For the cross-term, we start by looking at:
\[
c_{\alpha_\mu\beta_\nu\beta_\rho}=\begin{cases}
    \alpha_\rho e^{\beta_\mu+\beta_\rho}\tfrac{1}{(e^{\beta_\mu}-e^{\beta_\rho})^2} & \mu=\nu\neq \rho\\
    \alpha_\nu e^{\beta_\mu+\beta_\nu}\tfrac{1}{(e^{\beta_\mu}-e^{\beta_\nu})^2} & \mu=\rho\neq \nu\\
    \alpha_\nu e^{\beta_\nu}\tfrac{1}{e^{\beta_\nu}-e^{\beta_\mu}}\bigl[1-e^{\beta_\nu}\tfrac{1}{e^{\beta_\nu}-e^{\beta_\mu}}\bigr] & \mu\neq\nu= \rho\\
    e^{\beta_\mu}\bigl[\lambda'_{\ell+r,r}(e^{\beta_\mu})-\sum_{\gamma\neq \mu}\alpha_\gamma e^{\beta_\gamma}\tfrac{1}{(e^{\beta_\gamma}-e^{\beta_\mu})^2}\bigr] & \mu=\nu=\rho\\
    0&\text{otherwise}
\end{cases}.
\]
All the linear terms in $\bm{\alpha} $ come from differentiating the tail term, therefore the only condition we are to set on $\Psi$ is:
\[
\begin{aligned}
\Psi_{\alpha_\mu\beta_\mu\beta_\mu}&=e^{\beta_\mu}\lambda'_{\ell+r,r}(e^{\beta_\mu})\\&=(\ell+1)e^{(\ell+1)\beta_\mu}+\ell\sigma_{\ell+r}e^{\ell\beta_\mu}+\dots+\sigma_{r+1}e^{\beta_\mu}+\sigma_{r-1}e^{-\beta_\mu}+\dots+\sigma_0e^{-r\beta_\mu},
\end{aligned}
\]
whereas any other $\alpha\beta\beta$-derivative vanishes. This gives:
\[
\begin{aligned}
    \Psi(\bm{t},\bm{\alpha},\bm{\beta})&=\tfrac{1}{\ell+1}\widetilde{\Theta}_{\ell+1}(\bm{\alpha},\bm{\beta})+\tfrac{1}{\ell}\sigma_{\ell+r}(\bm{t})\widetilde{\Theta}_\ell(\bm{\alpha},\bm{\beta})+\dots+\sigma_{r+1}(\bm{t})\widetilde{\Theta}_1(\bm{\alpha},\bm{\beta})+\\&\quad-\sigma_{r-1}(\bm{t})\widetilde{\Theta}_{-1}(\bm{\alpha},\bm{\beta})+\dots-\tfrac{1}{r}\sigma_0(\bm{t})\widetilde{\Theta}_{-r}(\bm{\alpha},\bm{\beta})+\\
    &\quad+c_1(\bm{t})\alpha_1\beta_1+\dots+c_{n_p}(\bm{t})\alpha_{n_p}\beta_{n_p}+f_1(\bm{t})\alpha_1+\dots+f_{n_p}(\bm{t})\alpha_{n_p}.
\end{aligned}
\]
For some unknown exponential polynomials $c_1,\dots, c_{n_p}\in\CC[\bm{t},e^{t_\bullet}]/\CC$, $f_1,\dots, f_{n_p}\in\CC[\bm{t},e^{t_\bullet}]/\CC[\bm{t}]^{(\leq 1)}$. If we prove that there are actually only two such polynomials $c$ and $f$ such that $c_\mu=c$ and $f_\mu=f$ for any $\mu=1,\dots, n_p$, then we have the statement.\\

The former set of polynomials can be fixed by looking at:
\[
\begin{aligned}
    \Psi_{t_a\alpha_\mu\beta_\nu}&=\sum_{p=0}^{\ell+r}\pdv{\sigma_p}{t_a}\,c_{\sigma_p\alpha_\mu\beta_\nu}\\&=\delta_{\mu\nu}\biggl[\pdv{\sigma_{\ell+r}}{t_a}e^{\ell \beta_\mu}+\dots+\pdv{\sigma_{r+1}}{t_a}e^{\beta_\mu}+\pdv{\sigma_r}{t_a}+\pdv{\sigma_{r-1}}{t_a}e^{-\beta_\mu}+\dots+\pdv{\sigma_0}{t_a}e^{-r\beta_\mu}\biggr],
\end{aligned}
\]
which, according to the expression above, should be equal to:
\[
\Psi_{t_a\alpha_\mu\beta_\nu}=\delta_{\mu\nu}\biggl[\pdv{\sigma_{\ell+r}}{t_a}e^{\ell \beta_\mu}+\dots+\pdv{\sigma_{r+1}}{t_a}e^{\beta_\mu}+\pdv{c_\mu}{t_a}+\pdv{\sigma_{r-1}}{t_a}e^{-\beta_\mu}+\dots+\pdv{\sigma_0}{t_a}e^{-r\beta_\mu}\biggr].
\]
Therefore, in $\CC[\bm{t},e^{t_\bullet}]/\CC$, it actually is $c_1=\dots=c_{n_p}=\sigma_r$.\\

As for $f_1,\dots, f_{n_p}$, we consider the following three-point function components:
\[
\begin{aligned}
    \Psi_{t_at_b\alpha_\mu}&=\sum_{p,q=0}^{\ell+r}\pdv{\sigma_p}{t_a}\pdv{\sigma_q}{t_b}\sum_{n\geq 0}\biggl[e^{n\beta_\mu}\,\Gamma_{npq}+e^{-(n+1)\beta_\mu}\,\Xi_{npq}\biggr],
\end{aligned}
\]
where:
\[
\begin{aligned}
    \Gamma_{npq}(\bm{t})&:=\Res_0\biggl\{w^{\ell+2r+n-(p+q)}\frac{\dd w}{Q(w)}\biggr\}, &&&\Xi_{npq}(\bm{t})&:=\Res_0\biggl\{w^{p+q+n-r}\frac{\dd w}{P(w)}\biggr\}.
\end{aligned}
\]
Here, $P(w):=w^{r+1}\lambda_{\ell+r,r}'(w)\in\CC[\bm{t},e^{t_\bullet}][w]$ is the polynomial controlling the singular behaviour of $\lambda_{\ell+r,r}'$ at zero. Clearly, only finitely many of these coefficients will be non-zero.

On the other hand, computing the derivative directly from the expression above gives:
\[
\begin{aligned}
    \Psi_{t_at_b\alpha_\mu}&=\tfrac{1}{\ell}\pdv[2]{\sigma_{\ell+r}}{t_a}{t_b}e^{\ell\beta_\mu}+\dots+\pdv[2]{\sigma_{r+1}}{t_a}{t_b}e^{\beta_\mu}+\pdv[2]{\sigma_r}{t_a}{t_b}\beta_\mu+\pdv[2]{f_\mu}{t_a}{t_b}+\\&\quad-\pdv[2]{\sigma_{r-1}}{t_a}{t_b}e^{-\beta_\mu}+\dots-\tfrac{1}{r}\pdv[2]{\sigma_0}{t_a}{t_b}e^{-r\beta_\mu}.
\end{aligned}
\]
This, again, produces a system of PDEs for the Dubrovin-Zhang coordinates $\bm{\sigma}(\bm{t})$, and, most remarkably, the following system of PDEs for $f_\mu$:
\[
\pdv[2]{f_\mu}{t_a}{t_b}=\sum_{p,q=0}^{\ell+r}\pdv{\sigma_p}{t_a}\pdv{\sigma_q}{t_b}\Res_0\biggl\{w^{\ell+2r-(p+q)}\frac{\dd w}{Q(w)}\biggr\}.
\]
In particular, this is the very same system of equations for any choice of $\mu$, and it is easy to see that two solutions can only differ by a linear polynomial in the $\bm{t}$ variables. As a consequence, we can, without loss of generality, take $f_1=\dots=f_{n_p}=:f\in\CC[\bm{t},e^{t_\bullet}]/\CC[\bm{t}]^{(\leq 1)}$.\\

As a final remark, it is important to notice that one of the PDEs one gets for the exponential polynomials $\bm{\sigma}(\bm{t})$ is actually very simple, namely it has to be $\partial_{t_a}\partial_{t_b}\sigma_r=0$ for any choice of $a$ and $b$. In other words, $\sigma_r$ must be an affine function of the flat coordinates. This is necessary for \cref{EAW:eq:Fint} to make sense, for otherwise the second-to-last term would not be $\widetilde{\Sym}_{n_p}$-invariant up to quadratics.
\end{proof}

For the discussion that will follow, it is worthwhile to now give a more precise description of the kind of Frobenius manifold structures we are presently considering. In particular, we are going to explicitly work out the unity and the Euler vector field.

\begin{prop}[Unity vector field]
    Let $\ell ,r,n_p \in\Z_{\geq 1}$ and consider $\Hw_{0,\ell+r+n_p +1}(\ell,r-1,\bm{0})$ equipped the third-kind  Abelian differential $\omega$ with simple poles at $0$ and $\infty$ and residues $-1$ and $+1$ respectively.
    
    In the coordinates $\{a_0,\dots, a_{\ell+r+n_p },b_0,\dots, b_{n_p -1}\}$ defined so that, in each equivalence class, one can find a representative of the form:
    \begin{equation}\label{eq:EAWdeformationtogether}
\lambda(w)=\frac{w^{\ell+r+n_p +1}+a_{\ell+r+n_p }w^{\ell+r+n_p }+\dots+a_1w+a_0}{w^r\bigl(w^{n_p }+b_{n_p -1}w^{n_p -1}+\dots+b_0\bigr)},
    \end{equation}
    the unity vector field is:
    \begin{equation}\label{eq:EAW:identity}
        e=\pdv{a_{r+n_p }}+b_{n_p -1}\pdv{a_{r+n_p -1}}+\dots+b_0\pdv{a_r}.
    \end{equation}
\end{prop}
\begin{proof}
    It is clear from \cref{eq:HurwitzFrobeniusstructure} that the unity vector field is characterised by the property $\Lie_e\lambda=1$ \cite{ma2023frobenius}. It is, then, a simple calculation to show that:
    \[
    \begin{aligned}
        \Lie_e\lambda&=\frac{1}{w^r\bigl(w^{n_p }+b_{n_p -1}w^{n_p -1}+\dots+b_0\bigr)}\bigl[w^{r+n_p }+b_{n_{\mathrm{P}-1}}w^{r+n_p -1}+\dots+b_0w^r\bigr]=1.
    \end{aligned}
    \]
\end{proof}
For consistency, we now write down the unity vector field in our system of flat coordinates for the Frobenius pairing, and check that it is indeed covariantly constant.
\begin{prop}[Unity in flat coordinates]\label{lem:identityflatcoords}
The following hold:
\begin{itemize}
    \item     One can always, without loss of generality, fix the Dubrovin-Zhang system of flat coordinates $\{t_1,\dots, t_{\ell+r+1}\}$ on the Hurwitz Frobenius manifold $\Hw_{0,\ell+r+1}^\omega(\ell,r-1)$ such that $t_{\star}=\sigma_r$ for some $\star=1,\dots, \ell+r+1$.
\item In the coordinate system $(\bm{t},\bm{\alpha},\bm{\beta})$ as in \cref{lem:flatEAWdeform}, with the $\bm{t}$-coordinates as in the previous point of this Proposition, we have, on the Hurwitz Frobenius manifold $\Hw_{0,\ell+r+n_p +1}^\omega(\ell,r-1,\bm{0})$,
    \[
    e=\pdv{t_\star}.
    \]
\end{itemize}
\end{prop}
\begin{proof}
    As pointed out in the proof of \cref{thm:prepotentialDZ}, $\sigma_r$ can depend at most linearly on the flat coordinates $\bm{t}$. Therefore, there is an Euclidean transformation mapping any set of flat coordinates of the Frobenius pairing into one satisfying the required property.
    
For the second point, elementary polynomial manipulations of \cref{eq:EAWdeformationtogether} and \cref{eq::EAWdeformation} yield:
\[
\begin{matrix}
     a_r=\sigma_0b_r+\sigma_1b_{r-1}+\sigma_2b_{r-2}+\dots+\sigma_rb_0,\\[0.5em]
    \vdots\\[0.5em]
    a_{r+n_p }=\sigma_0b_{r+n_p }+\sigma_1b_{r+n_p -1}+\dots+\sigma_{r+n_p }b_0,
\end{matrix}
\]
where it is understood that $\sigma_{\ell+r+1}=b_{n_p }:=1$ and $\sigma_\mu=c_\nu:=0$ for $\mu>\ell+r+1$ and $\nu>n_p $. Moreover, these are the only $a$-coordinates that do depend on $\sigma_r$. It follows that:
\[
\pdv{\sigma_r}=\sum_{\alpha=0}^{\ell+r+n_p }\pdv{a_\alpha}{\sigma_r}\pdv{a_\alpha}=\sum_{\alpha=0}^{n_p }b_{\alpha}\pdv{a_{r+\alpha}}=e.
\]
\end{proof}
\begin{oss}
    Notice that the previous result is perhaps not surprising, as it is known that $\partial_{t_r}$ is the unity vector field for DZ Frobenius manifold structure on the orbit space of $\widetilde{W}^{(r)}$ \cite{dubrovin_zhang_1998}*{Theorem 2.1}.
\end{oss}

The Euler vector field may also be expressed in the flat coordinates: 

\begin{prop}[Euler vector field] In the coordinate systems defined in \cref{lem:EAW:flatcoordsIntForm}, \cref{eq:EAWdeformationtogether} and \cref{lem:flatEAWdeform}, the Euler vector field $E$ for the Frobenius manifold $\Hw^\omega_{0,\ell+r+n_p +1}(\ell,r-1,\bm{0})$ reads:
  \[
\begin{aligned}
    E&=\tfrac{1}{\ell+1}\bigl[\partial_{\varphi_1}+\dots+\partial_{\varphi_{n_z }}+\partial_{\psi_1}+\dots+\partial_{\psi_{n_p }}\bigr]\\&=\tfrac{1}{\ell+1}\bigl[n_z \,a_0\partial_{a_0}+(n_z -1)a_1\partial_{a_1}+\dots+a_{n_z -1}\partial_{a_{n_z-1 }}+n_p \,b_0\partial_{b_0}+\dots+b_{n_p -1}\partial_{b_{n_p -1}}\bigr]\\
    &=\sum_{p=0}^{\ell+r}\bigl(1+\tfrac{r-p}{\ell+1}\bigr)\sigma_p\partial_{\sigma_p}+\tfrac{1}{\ell+1}\bigl[\partial_{\beta_1}+\dots+\partial_{\beta_{n_p }}\bigr]+\alpha_1\partial_{\alpha_1}+\dots+\alpha_{n_p }\partial_{\alpha_{n_p }},
\end{aligned}
\]
where $n_z =\ell+1+r+n_p $.
\end{prop}
\begin{proof}
   Just as in the purely rational case one has, as a consequence of Euler's Theorem for homogeneous functions, that the Euler vector field is, up to an overall normalisation factor, characterised by the property
$   \bigl(\Lie_E\lambda\bigr)(w)=(\ell+1)\,\lambda(w)-w\,\lambda'(w)$, \cite{ma2023frobenius}. It is clear that the vector field $\partial_{\varphi_1}+\dots+\partial_{\varphi_{n_z }}+\partial_{\psi_1}+\dots+\partial_{\psi_{n_p }}$ satisfies such an equation. The normalisation is chosen so that the degree of the component in the direction of the identity is one.

According to the previous Proposition, this is equivalent to requiring that $E(\sigma_r)=\sigma_r$. In particular, we have the relations:
\[
\begin{aligned}
a_{r+\mu}&=\sigma_0b_{r+\mu}+\dots+\sigma_rb_\mu+\dots+\sigma_{r+\mu}b_0,  & \mu=0,\dots, n_p ,
\end{aligned}
\]
as in the proof of \cref{lem:identityflatcoords}. Since $a_{r+\mu}\in\CC[e^{\varphi_1},\dots, e^{\varphi_{n_z }},e^{\psi_1},\dots, e^{\psi_{n_p }}]$ is homogeneous of degree $n_z -(r+\mu)=\ell+1+n_p -\mu$, and $b_\mu$ is similarly homogeneous of degree $n_p -\mu$, it follows that $\sigma_r$ must be a homogeneous polynomial of degree $\ell+1$. Hence, by Euler's Theorem:
\[
\bigl[\partial_{\varphi_1}+\dots+\partial_{\varphi_{n_z }}+\partial_{\psi_1}+\dots+\partial_{\psi_{n_p }}\bigr]\sigma_r=(\ell+1)\,\sigma_r.
\]
This fixes the overall normalisation and proves the second expression for $E$ in the $(\bm{a},\bm{b})$ coordinates.\\
A similar argument shows that the degree of $\sigma_p$ as a homogeneous polynomial in $e^{\varphi_1},\dots, e^{\varphi_{n_z }},$ $e^{\psi_1},\dots, e^{\psi_{n_p }}$ is $n_z-n_P-p=\ell+1+r-p$, and that the $\alpha$-s are homogeneous of degree $\ell+1$.
\end{proof}
\begin{oss}
    Again, since the prepotential must be quasi-homogeneous of degree $3-d=2$ with respect to the Euler vector field, it follows that $f$ is also a quasi-homogeneous polynomial of degree one.
\end{oss}

\subsection{Examples}
We provide some explicit examples of the family of structures we have just described. Notice that, in the case of a single simple pole, the Frobenius manifold structure on the corresponding Hurwitz space is known to be isomorphic to the one on the orbit space of some extended affine-Weyl group of type $\AN_\ell$ with two marked roots, as described in \cite{zuo2020frobenius}. More precisely, we have the local Frobenius manifold isomorphism between $\Hw^\omega_{0,\ell+r+2}(\ell,r-1,0)$ and $\CC^{\ell+r+3}/\widetilde{W}^{(r,r+1)}(\AN_{\ell+r+1})$.

We will check that the prepotentials agree in the two cases (up to an Euclidean transformation of the flat coordinates).

\begin{ex}[Orbit space of the group $\widetilde{W}^{(1)}(\AN_1)$]
We consider the Frobenius manifold structure on the Hurwitz space $\Hw_{0,3}(0,0,0)$ of equivalence classes of meromorphic functions on $\RS$ of the form:
\begin{align*}
    \lambda(w)&= w+\sigma_1+\frac{\sigma_0}{w}+\frac{w\alpha }{w-e^\beta}\,,
\end{align*}
with primary differential $-\tfrac1w\dd{w}$.

This is a single-pole deformation of the Frobenius manifold on the orbit space $\CC^2/\widetilde{W}^{(1)}(\AN_1)$, which is well-known to be isomorphic to the quantum cohomology $QH^\bullet(\RS)$ of the projective line. The corresponding prepotential is given by \cite{dubrovin_zhang_1998}*{Example 2.1}:
\[
\Fr_{QH^\bullet(\RS)}(t_1,t_2)= \tfrac12 t_1^2t_2+e^{t_2},
\]
where the flat coordinates are simply $\sigma_1=:t_1$, $\sigma_0=:e^{t_2}$.\\

Clearly here $f=0$, therefore the prepotential, identity and Euler vector field are going to be:
\begin{align*}
    \Fr(t_1,t_2,\alpha,\beta)&=\tfrac12 t_1^2t_2+e^{t_2}+\tfrac12\alpha^2\log\alpha+\tfrac12 \alpha^2\beta+\alpha e^{\beta}-\alpha e^{t_2-\beta}+t_1\alpha\beta,\\
    e&=\partial_{\sigma_1}=\partial_{t_1},\\
    E&= t_1\partial_{t_1}+2\partial_{t_2}+\partial_\beta+\alpha\partial_\alpha.
\end{align*}
According to \cite{zuo2020frobenius}, the one under consideration is the superpotential description of the Frobenius manifold structure on the orbit space $\CC^4/\widetilde{W}^{(1,2)}(\AN_2)$, therefore there must be an affine transformation mapping the function above to the one in \cite{zuo2020frobenius}*{Example 4.6}. One easily check that the following reshuffling does the job:
\[
\begin{aligned}
    s_1&:=-t_1, &&& s_2&:=-\alpha, &&& s_3&:=\beta, &&& s_4&:=t_2.
\end{aligned}
\]

Notice, finally, that once the flat coordinates and the corresponding prepotential for the Frobenius manifold structure on $QH^\bullet(\RS)$ have been fixed, it is a simple exercise to write down the prepotential for the structure on any simple-pole deformation of the corresponding LG superpotential, with as many poles as demanded.
For instance, on the Hurwitz space $\Hw_{0,4}(0,0,0,0)$, whose points can be represented as:
\[
\widehat{\lambda}(w)= w+t_1+\tfrac1w e^{t_2}+\frac{w\alpha_1}{w-e^{\beta_1}}+\frac{w\alpha_2}{w-e^{\beta_2}},
\]
the free energy differs from the one above due only to the presence of two additional logarithmic terms and the generalised polarised power sums being “longer”:
\begin{align*}
    \widehat{\Fr}(\bm{t},\bm{\alpha},\bm{\beta})&=\tfrac12 t_1^2t_2+e^{t_2}+\tfrac12 \alpha_1^2\log\alpha_1+\tfrac12\alpha_2^2\log\alpha_2+\alpha_1\alpha_2\log(e^{\beta_1}-e^{\beta_2})+\tfrac12\alpha_1^2\beta_1+\\&\qquad+\tfrac12 \alpha_2^2\beta_2
+t_1\bigl[\alpha_1\beta_1+\alpha_2\beta_2\bigr]+\alpha_1e^{\beta_1}+\alpha_2e^{\beta_2}-e^{t_2}\bigl[\alpha_1e^{-\beta_1}+\alpha_2e^{-\beta_2}\bigr].
\end{align*}
The identity is the same as before, while the Euler vector field is simply modified by adding the terms corresponding to the second pole:
\[
\widehat{E}= t_1\partial_{t_1}+2\partial_{t_2}+\partial_{\beta_1}+\partial_{\beta_2}+\alpha_1\partial_{\alpha_1}+\alpha_2\partial_{\alpha_2}\,.
\]
In particular, we notice that $\widehat{\Fr}$ corresponds to the solutions $\mathcal{F}_{\phi_j}$ given in \cite[Theorem 6.2]{rejeb2023new} by identifying:
\[
\begin{aligned}
    t_1&=x_{j,2m+1-j}\,,&&& t_2&=x_{j,j}\,,\\\alpha_\mu&=x_{j,2m+1-\mu}\,&&& \beta_\mu&=x_{j,\mu}\,,&&& \mu=1,\dots, m\,,
\end{aligned}
\]
where $m=2$, but the same map works for any number of simple poles. As a matter of fact, the prepotential $\mathcal{F}_{\phi_j}$ corresponds to the Frobenius manifold structure on the Hurwitz space whose superpotential is given in \cite[eq. 6.1]{rejeb2023new} -- i.e. $\Hw_{0,m+1}(0,\dots, 0)$ -- with the choices of the primary differentials $\phi_j$ as given in \cite[eq. 6.2]{rejeb2023new}. In particular, $\phi_j$ is the third-kind Abelian differential having two simple poles at $\infty_0=0$ and $\infty_j=a_j$ with residues $\mp 1$ respectively. Since all the poles have the same order, these Frobenius manifold structures will all be isomorphic -- as the diagonal action of the symmetric group maps these primary differentials into one another while leaving the superpotential invariant. In fact, the prepotentials $\mathcal{F}_{\phi_j}$ are all the same up to relabelling the coordinates. The same applies to the identity $\mathbf{1}_{\phi_j}$ and the Euler vector field $\mathbf{E}_{\phi_j}$ as given in \cite[Proposition 6.4]{rejeb2023new}.
\end{ex}

\begin{ex}[Orbit space of the group $\widetilde{W}^{(1)}(\AN_2)$] We now consider the Frobenius manifold structure on the Hurwitz space $\Hw_{0,4}(1,0,0)$ of equivalence classes of meromorphic functions on $\RS$ of the form:
\begin{equation}\label{eq:superpotMZA3}
    \lambda(w)= w^2+\sigma_2w+\sigma_1+\frac{\sigma_0}{w}+\frac{w\alpha}{w-e^\beta},
\end{equation}
    with primary differential $\omega=-\tfrac{1}{w}\dd w$.\\

    According to \cite{dubrovin_zhang_1998}, $\Hw_{0,3}(1,0)$ is the B-mirror model for the orbit space $\CC^3/\widetilde{W}^{(1)}(\AN_2)$. Flat coordinates $\{t_1,t_2,t_3\}$ are now given by:
    \[
    \begin{aligned}
        \sigma_0&=e^{t_3}, &&& \sigma_1&=t_1, &&& \sigma_2&=t_2.
    \end{aligned}
    \]
    The corresponding free energy is the same as the one in \cite{dubrovin_zhang_1998}*{Example 2.2}:
    \[
\Fr_{\AN_2^{(1)}}(t_1,t_2,t_3)=\tfrac12 t_1^2t_3+\tfrac14 t_2^2 t_1 +t_2 e^{t_3}-\tfrac{1}{96}t_2^4.
\]

Here, we have precisely one non-vanishing second derivative of $f$, namely $f_{t_2t_2}=\tfrac12$. Therefore, we can take $f(\bm{t})=\tfrac14 t_2^2$. \\As a consequence, the prepotential, identity and Euler vector field are for $\Hw^\omega_{0,4}(1,0,0)$ are respectively given by:
\[
\begin{aligned}
     \Fr(\bm{t},\alpha,\beta)&=\Fr_{\AN_2^{(1)}}(\bm{t})+\tfrac12 \alpha^2\log\alpha+\tfrac12\alpha^2\beta+\tfrac12 \alpha e^{2\beta}-\alpha e^{t_3-\beta}+t_2\alpha e^\beta+t_1\alpha\beta+\tfrac14 t_2^2\alpha,\\
     e&=\partial_{\sigma_1}=\partial_{t_1},\\
     E&= t_1\partial_{t_1}+\tfrac12 t_2\partial_{t_2}+\tfrac32 \partial_{t_3}+\alpha\partial_\alpha+\tfrac12 \partial_\beta.
\end{aligned}
\]
This is the same as the prepotential for the Frobenius manifold structure on $\CC^5/\widetilde{W}^{(1,2)}(\AN_3)$ from \cite{zuo2020frobenius}*{Example 4.7} if we change flat coordinates in the following fashion:
\[
\begin{aligned}
    t_1&=:s_1,&&& t_2&=:s_3,&&&t_3&=:s_5,&&&\alpha&=:s_2,&&&\beta+i\pi&=:s_4.
\end{aligned}
\]

Moreover, we notice that, from the obvious reflection symmetry about the centre of the $\AN_3$ diagram, we expect such a structure to be isomorphic to the one the orbit space $\CC^5/\widetilde{W}^{(2,3)}(\AN_3)$. The corresponding prepotential is given in \cite{zuo2020frobenius}*{Example 4.8} as a separate example, but one could also compute it using the result from \cref{thm:prepotentialDZ}. In particular, we should now look at LG superpotentials of the form:
\[
\begin{aligned}
    \widehat{\lambda}(z)&= z+v_2+\tfrac1z v_1 e^{v_3}+\tfrac{1}{z^2}e^{2v_3}+\frac{\gamma z}{z-e^{\delta}},
\end{aligned}
\]
where $v_1,v_2,v_3$ are flat coordinates for the Frobenius metric on the orbit space $\CC^3/\widetilde{W}^{(2)}(\AN_2)$. The corresponding prepotential, according to \cref{thm:prepotentialDZ}, is:
\[
    \begin{aligned}
        \widehat{\Fr}(\bm{v},\gamma,\delta)&=\Fr_{\AN_2^{(2)}}(\bm{v})+\tfrac12 \gamma^2\log\gamma+\tfrac12\gamma^2\delta+\gamma e^\delta-\tfrac12\gamma e^{2(v_3-\delta)}+v_2\gamma e^{v_3-\delta}+v_1\gamma\delta,
    \end{aligned}
\]
where $\Fr_{\AN_2^{(2)}}$ is again given in \cite{dubrovin_zhang_1998}*{Example 2.2}  (with $t_1$ swapped with $t_2$).\\

In particular, it is clear that the previous superpotential is projective equivalent to \cref{eq:superpotMZA3} via the M\"{o}bius transformation $z:=\tfrac{1}{w}e^{t_3}$. Notice that such a transformation swaps the poles of the primary differential (and the corresponding residues). Carrying out the substitution explicitly in \cref{eq:superpotMZA3} yields:
\[
\lambda(z)= z+t_1+\alpha+\tfrac{1}{z}t_2e^{t_3}+\tfrac{1}{z^2}e^{2t_3}+\frac{-\alpha z}{z-e^{t_3-\beta}}.
\]
Comparing this with the previous expression suggests the following coordinate transformation:
   \[
   \begin{aligned}
       v_1&=t_2 ,&&&v_2&=t_1+\alpha,&&&v_3&=t_3, &&& \gamma&=-\alpha ,&&&\delta&=t_3-\beta.
   \end{aligned}
   \]
   One can finally check that, up to quadratics, carrying out the substitution in $\widehat{\Fr}$ actually gives $\Fr$. As anticipated, since the transformation is affine, one can conclude that the two structures are actually isomorphic.
    \end{ex}
    
\begin{ex}[Orbit space of the group $\widetilde{W}^{(2)}(\AN_3)$] We consider the Frobenius manifold structure on the Hurwitz space $\Hw_{0,6}(1,1,0,0)$, whose points can be uniquely represented as:
\[
\lambda(w)= w^2+\sigma_3w+\sigma_2+\sigma_1\tfrac1w+\sigma_0\tfrac{1}{w^2}+\frac{w\alpha_1}{w-e^{\beta_1}}+\frac{w\alpha_2}{w-e^{\beta_2}},
\]
given by the choice of the primary differential $\omega=-\tfrac1w\dd{w}$.\\

In order to write down its prepotential in the given system of flat coordinates, we start by giving flat coordinates on the Hurwitz space $\Hw_{0,4}^\omega(1,1)$, which is a B-model for the orbit space $\CC^4/\widetilde{W}^{(2)}(\AN_3)$. These are given in \cite{dubrovin_zhang_1998}*{Example 2.6} as functions of the invariant exponential polynomials. A simple calculation, then, shows that:
\[
\begin{aligned}
    \sigma_3&=t_1, &&& \sigma_2&=t_2, &&& \sigma_1&=e^{t_4}t_3, &&& \sigma_0&=e^{2t_4}.
\end{aligned}
\]
The prepotential in these coordinates is:
\[
\begin{aligned}
    \Fr_{\AN_{3}^{(2)}}(t_1,t_2,t_3,t_4)= \tfrac14 t_1^2t_2+\tfrac12 t_2^2t_4+\tfrac14 t_2t_3^2-\tfrac{1}{96}t_1^4-\tfrac{1}{96}t_3^4+t_1t_3e^{t_4}+\tfrac12 e^{2t_4}.
\end{aligned}
\]

The only non-vanishing second derivative of $f$ as given by \cref{EAW:eq:systemf} is $f_{t_1t_1}=\tfrac12$. Hence, we can take $f(\bm{t})=\tfrac14 t_1^2$. Therefore, the Frobenius manifold structure on the Hurwitz space $\Hw^\omega_{0,6}(1,1,0,0)$ can be described as follows:
    \begin{align*}
        \Fr(\bm{t},\alpha_1,\alpha_2,\beta_1,\beta_2)&= \Fr_{\AN_{3}^{(2)}}(\bm{t})+\tfrac12 \alpha_1^2\log\alpha_1+\tfrac12\alpha_2^2\log\alpha_2+\alpha_1\alpha_2\log(e^{\beta_1}-e^{\beta_2})+\\
&\qquad+\tfrac12\bigl[\alpha_1^2\beta_1+\alpha_2^2\beta_2\bigr]+t_2\bigl[\alpha_1\beta_1+\alpha_2\beta_2\bigr]+\tfrac12\bigl[\alpha_1e^{2\beta_1}+\alpha_2e^{2\beta_2}\bigr]+\\
&\qquad+t_1\bigl[\alpha_1e^{\beta_1}+\alpha_2e^{\beta_2}\bigr]+\tfrac14 t_1^2\bigl[\alpha_1+\alpha_2\bigr]+\\&\qquad -e^{t_4}t_3\bigl[\alpha_1e^{-\beta_1}+\alpha_2e^{-\beta_2}\bigr]-\tfrac12 e^{2t_4}\bigl[\alpha_1e^{-2\beta_1}+\alpha_2e^{-2\beta_2}\bigr],\\
e&=\partial_{\sigma_2}=\partial_{t_2},\\
E&=\tfrac12 t_1\partial_{t_1}+t_2\partial_{t_2}+\tfrac12 t_3\partial_{t_3}+\partial_{t_4}+\tfrac12\bigl(\partial_{\beta_1}+\partial_{\beta_2}\bigr)+\alpha_1\partial_{\alpha_1}+\alpha_2\partial_{\alpha_2}.
    \end{align*}
    
    We again point out that adding more simple poles to the tail only modifies $\Fr_{\mathrm{int.}}$ because the exponential power sums contain more terms, the coefficients being the very same polynomials in $\bm{t}$.
\end{ex}
\documentclass[main.tex]{subfile}

\subsection{Monodromy group}\label{EAW:sec:monodromy}
As anticipated, it is interesting to probe the relation between the monodromy group of the Hurwitz Frobenius manifold -- as defined in \cite{Dubrovin1996}*{Appendix G} -- and some proper extension of some reflection group. The reason for doing that is that, as already noted in \cites{zuo2020frobenius,ma2023frobenius}, the group whose orbit space represents a C-model for the Hurwitz Frobenius manifold sits naturally inside the fundamental group of the complement of the discriminant. Therefore, studying subgroups of the monodromy group that come from extensions of reflection groups à la Dubrovin-Zhang gives natural candidates in the quest for an orbit space description. \\

Now, in order to describe the monodromy group, one ought to be looking at analytic continuation of systems of flat coordinates for the intersection form around the discriminant locus. We do possess a system such coordinates on open subsets non-intersecting the discriminant locus, namely the one given in \cref{lem:EAW:flatcoordsIntForm}. In particular, as already discussed, since the discriminant is the locus where any two of the zeros or of the poles coincide, the monodromy group will act on these set of coordinates as the group $\widetilde{\Sym}_{n_z}\times\widetilde{\Sym}_{n_p}$ permuting the zeros and poles separately and shifting their logarithms by integral multiples of $i2\pi$.\\

In order to relate these to flat coordinates for the intersection form on some orbit space of an extended reflection group, we recall that, following the Saito construction, these are given by coordinates with respect to a basis of coroots. In particular, since we are conjecturally dealing with some extension of affine-Weyl groups of type-A, we firstly let:
\begin{align}
    \begin{split}
        \bm{\xi}&:=(x_1,x_2-x_1,\dots, x_{n_z -1}-x_{n_z -2},-x_{n_z -1})\in\CC^{n_z },\\
        \bm{\zeta}&:=(y_1,y_2-y_1,\dots, y_{n_p -1}-y_{n_p -2},-y_{n_p -1})\in\CC^{n_p },
    \end{split}
\end{align}
be vectors in the hyperplanes of $\CC^{n_z }$ and $\CC^{n_p }$ respectively whose coordinates sum up to zero. Then, we consider the following coordinate transformation on an open subset of $\CC^{n_z +n_p }$:
\begin{align}
    \begin{split}
        \varphi_a&=: i2\pi \bigl(\xi_a+\tfrac{A}{n_z } x_{n_z }+ \tfrac{B}{n_z } y_{n_p }\bigr)\,, \qquad a=1,\dots, n_z\, ;\\
        \psi_\alpha&=: i2\pi \bigl(\zeta_\alpha+\tfrac{C}{n_p }x_{n_z }+\tfrac{D}{n_p }y_{n_p }\bigr)\,,\qquad \alpha=1,\dots, n_p \,,
    \end{split}
\end{align}
for some choice of $A,B,C,D\in\CC$. These are easily inverted to give the $x$-$y$ coordinates in terms of the $\varphi$-$\psi$ ones:
\begin{align}\label{eq:coveringinverse}
    \begin{split}
        x_a&=\tfrac{1}{i2\pi}\bigl[\bigl(1-\tfrac{a}{n_z }\bigr)(\varphi_1+\dots+\varphi_a)-\tfrac{a}{n_z }(\varphi_{a+1}+\dots+\varphi_{n_z })\bigr]\,,\qquad  a=1,\dots, n_z -1\,;\\
        y_\alpha&=\tfrac{1}{i2\pi}\bigl[\bigl(1-\tfrac{\alpha}{n_p }\bigr)(\psi_1+\dots+\psi_\alpha)-\tfrac{\alpha}{n_p }(\psi_{\alpha+1}+\dots+\psi_{n_p })\bigr]\,,\qquad \alpha=1,\dots, n_p -1\,;\\
        x_{n_z }&=\tfrac{1}{i2\pi}\tfrac{1}{\Delta}\bigl[D(\varphi_1+\dots+\varphi_{n_z })-B(\psi_1+\dots+\psi_{n_p })\bigr]\,,\\
        y_{n_p }&=\tfrac{1}{i2\pi}\tfrac{1}{\Delta}\bigl[A(\psi_1+\dots+\psi_{n_p })-C(\varphi_1+\dots+\varphi_{n_z })\bigr]\,,
    \end{split}
\end{align}
where $\Delta:=AD-BC$.\\

It is, then, straightforward to check that:
\begin{lem}
    The action of the monodromy group of $\Hw^\omega_{0,n_z }(n_z -(n_p +r)-1,r-1,\bm{0})$ on an open subset of $\CC^{n_z +n_p }$ in the $x-y$ coordinates is generated by the following transformations:
    \begin{align}
        \begin{split}
                        t_a\begin{pmatrix}
                \bm{\xi}\\\bm{\zeta}\\x_{n_z }\\y_{n_p }
            \end{pmatrix}&:=\begin{pmatrix}
                \bm{\xi}+\bm{e}_a-\tfrac{1}{n_z }\bigl(\bm{e}_1+\dots+\bm{e}_{n_z}\bigr)\\
                \bm{\zeta}\\
                x_{n_z }+\tfrac{1}{\Delta}D\\
                y_{n_p }-\tfrac{1}{\Delta }C
            \end{pmatrix}\,, \qquad a=1,\dots,n_z\,; \\
            s_\alpha\begin{pmatrix}
                \bm{\xi}\\\bm{\zeta}\\x_{n_z }\\y_{n_p }
            \end{pmatrix}&:=\begin{pmatrix}
                \bm{\xi}\\
                \bm{\zeta}+\bm{e}_\alpha-\tfrac{1}{n_p }\bigl(\bm{e}_1+\dots+\bm{e}_{n_p}\bigr)\\
                x_{n_z }-\tfrac{1}{\Delta}B\\
                y_{n_p }+\tfrac{1}{\Delta }A
            \end{pmatrix}\,, \qquad \alpha=1,\dots, n_p \,;
               \end{split}
    \end{align}
   and by permutations in the components of $\bm{\xi}$ and $\bm{\zeta}$ separately.
\end{lem}
\begin{proof} It is a straightforward computation.

As an example, we look at the action of a generator of the $\Z^{n_z +n_p }$ subgroup, namely, the transformation shifting $\varphi_a$ by $i2\pi$ for some $a=1,\dots, n_z $ and leaving all the other coordinates invariant. It is clear by looking at \cref{eq:coveringinverse} that this does the following to the $x-y$ coordinates:
    \[
    \begin{aligned}
        x_b&\mapsto x_b-\tfrac{b}{n_z }, &&& b&=1,\dots, a-1,\\
        x_b&\mapsto x_b+1-\tfrac{b}{n_z }, &&&b&=a,\dots, n_z -1,\\
        x_{n_z }&\mapsto x_{n_z }+\tfrac{1}{\Delta}D,\\
        y_{n_p }&\mapsto y_{n_p }-\tfrac{1}{\Delta}C,        
    \end{aligned}
    \]
    whereas the remaining $y$-coordinates are unchanged. Therefore, the components of $\bm{\zeta}$ are unchanged, while any of the coordinates of $\bm{\xi}$ -- with the exception of the $a^{\mathrm{th}}$ one -- is shifted by $\xi_b\mapsto \xi_b-\tfrac{b}{n_z }+\tfrac{b-1}{n_z }=\xi_b-\tfrac{1}{n_z }$. On the other hand, on the $a^{\mathrm{th}}$ one, we have $\xi_a\mapsto \xi_a+1-\tfrac{1}{n_z }$.
    
    Similarly, one can show that a shift in the $\alpha^{\mathrm{th}}$ $\psi$-coordinates results in $s_\alpha$, and finally that a transposition of two $\varphi$ (resp. $\psi$) coordinates yields the corresponding transposition in the components of $\bm{\xi}$ (resp. $\bm{\zeta}$). 
\end{proof}

It is, now, convenient to introduce the following notation; if $W_1$ and $W_2$ are two Weyl groups, we denote by:
\begin{equation}\label{eq:hattimes1}
    \widetilde{W}_1^{(r,s)}\,\widehat{\boxtimes} \,\,\widetilde{W}_2^{(p,q)}:=\bigl(\widetilde{W}_1\times\widetilde{W}_2\bigr)\rtimes\Z^2
\end{equation}
the extension of the product of the affine-Weyl groups $\widetilde{W}_1\times\widetilde{W}_2$ corresponding to the choices of two marked nodes on each diagram whose action on $\mathfrak{h}_{\R}^{(1)}\oplus \mathfrak{h}_{\R}^{(2)}\oplus\R^2$ is defined as follows:
\begin{equation}\label{eq:hattimes2}
    (g_1,g_2,n,m).(v,w,x,y):=\bigl(g_1.v+n\,\omega_r+m\,\omega_s,\,g_2.w+n\,\rho_{p}+m\,\rho_q,\,x-n,\,y-m\bigr),
\end{equation}
for $g_1\in\widetilde{W}_1$, $g_2\in\widetilde{W}_2$, $n,m\in\Z$ and $\omega_1,\dots,\omega_{\rank W_1}$ denote the fundamental weights of $W_1$, whereas $\rho_1,\dots,\rho_{\rank W_2}$ are the fundamental weights of $W_2$.

\begin{prop}\label{prop:subgroupreflection}
    For any choice of $n_1,n_2=1,\dots, n_z -1$ and $\nu_1,\nu_2=1,\dots, n_p -1$, there exists a normal subgroup of the monodromy group of $\Hw^\omega_{0,n_z }(n_z -(n_p +r)-1,r-1,\bm{0})$ isomorphic to $\widetilde{W}^{(n_1,n_2)}(\AN_{n_z -1})\,\widehat{\boxtimes}\,\,\widetilde{W}^{(\nu_1,\nu_2)}(\AN_{n_p -1})$.
\end{prop}
\begin{proof}
We first show that the affine-Weyl groups $\widetilde{W}(\AN_{n_z -1})$ and $\widetilde{W}(\AN_{n_p -1})$ are indeed subgroups of the monodromy group.

    To begin with, the Weyl group $W(\AN_\ell)$ acts as permutations in $\Sym_{\ell+1}$ on the coordinates of vectors in $\R^{\ell+1}$ adding up to zero with respect to the canonical basis. These transformations are clearly a subgroup of the monodromy group, according to the description provided in the previous Lemma.\\
As for the affine transformations, for any fixed $k=1,\dots,n_z -1$, we can consider the following transformation:
\[
(t_{k+1}^{-1}\circ t_k)\begin{pmatrix}
                \bm{\xi}\\\bm{\zeta}\\x_{n_z }\\y_{n_p }
            \end{pmatrix}=\begin{pmatrix}
                \bm{\xi}+\bm{e}_k-\bm{e}_{k+1}\\
                \bm{\zeta}\\
                x_{n_z }\\
                y_{n_p }
            \end{pmatrix}=\begin{pmatrix}
                \bm{\xi}+\alpha_k^\vee\\\bm{\zeta}\\x_{n_z }\\y_{n_p }
            \end{pmatrix},
\]
which generate a subgroup isomorphic to the coroot lattice of $\AN_{n_z -1}$.\\
Similarly, for $\mu=1,\dots, n_p -1$, the transformations $s_{\mu+1}^{-1}s_\mu$ generate a subgroup isomorphic to the coroot lattice of $\AN_{n_p -1}$.\\

As for the extended transformations, if one recalls the expansion of the fundamental weights in the simple root basis \cref{eq:fundweight}, then it is natural to consider, for $n=1,\dots, n_z -1$:
\[
T_n\begin{pmatrix}
                \bm{\xi}\\\bm{\zeta}\\x_{n_z }\\y_{n_p }
            \end{pmatrix}:=(t_n\circ\dots\circ t_1)\begin{pmatrix}
                \bm{\xi}\\\bm{\zeta}\\x_{n_z }\\y_{n_p }
            \end{pmatrix}=\begin{pmatrix}
                \bm{\xi}+\sum_{a=1}^n\bm{e}_a-\tfrac{n}{n_z }\sum_{a=1}^{n_z }\bm{e}_a\\\bm{\zeta}\\x_{n_z }+\tfrac{n}{\Delta}D\\y_{n_p }-\tfrac{n}{\Delta}C
            \end{pmatrix}=\begin{pmatrix}
                \bm{\xi}+\omega_n\\\bm{\zeta}\\x_{n_z }+\tfrac{n}{\Delta}D\\y_{n_p }-\tfrac{n}{\Delta}C
            \end{pmatrix}.
\]
Similarly, for $\mu=1,\dots, n_p -1$:
\[
S_\mu\begin{pmatrix}
                \bm{\xi}\\\bm{\zeta}\\x_{n_z }\\y_{n_p }
            \end{pmatrix}:=(s_\mu\circ\dots\circ s_1)\begin{pmatrix}
                \bm{\xi}\\\bm{\zeta}\\x_{n_z }\\y_{n_p }
            \end{pmatrix}=\begin{pmatrix}
                \bm{\xi}\\\bm{\zeta}+\sum_{\nu=1}^{\mu}\bm{e}_\nu-\tfrac{\mu}{n_p }\sum_{\nu=1}^{n_p }\bm{e}_\nu\\
                x_{n_z }-\tfrac{\mu}{\Delta}B\\y_{n_p }+\tfrac{\mu}{\Delta}A
            \end{pmatrix}=\begin{pmatrix}
                \bm{\xi}\\\bm{\zeta}+\rho_\mu\\
                x_{n_z }-\tfrac{\mu}{\Delta}B\\y_{n_p }+\tfrac{\mu}{\Delta}A
            \end{pmatrix}.
\]
It, then, follows, that if we fix our parameters in the coordinate transformation so that $A=-n_1$, $B=-n_2$, $C=-\nu_1$ and $D=-\nu_2$, then the transformations $T_{n_1}\circ S_{\nu_1}$ and $T_{n_2}\circ S_{\nu_2}$ will look like:
\[
\begin{aligned}
    (T_{n_1}\circ S_{\nu_1})\begin{pmatrix}
                \bm{\xi}\\\bm{\zeta}\\x_{n_z }\\y_{n_p }
            \end{pmatrix}&=\begin{pmatrix}
                \bm{\xi}+\omega_{n_1}\\\bm{\zeta}+\rho_{\nu_1}\\x_{n_z }-1\\y_{n_p }
            \end{pmatrix}, &&&  (T_{n_2}\circ S_{\nu_2})\begin{pmatrix}
                \bm{\xi}\\\bm{\zeta}\\x_{n_z }\\y_{n_p }
            \end{pmatrix}&=\begin{pmatrix}
                \bm{\xi}+\omega_{n_2}\\\bm{\zeta}+\rho_{\nu_2}\\x_{n_z }\\y_{n_p }-1
            \end{pmatrix}.
\end{aligned}
\]
These are precisely a set of generators for the extended transformation in the extended affine-Weyl group $\widetilde{W}^{(n_1,n_2)}(\AN_{n_z -1})\,\widehat{\boxtimes}\,\,\widetilde{W}^{(\nu_1,\nu_2)}(\AN_{n_p -1})$.\\

    Hence, for any choice of integers $n_1,n_2,\nu_1,\nu_2$ in the appropriate ranges, we can construct a subgroup of the monodromy group of the Hurwitz space isomorphic to the corresponding product of type-$\AN$ extended affine-Weyl group with two marked nodes -- namely the one generated by permutations in the coordinates of $\bm{\xi}$ and $\bm{\zeta}$, the simple-coroot affine translations $\{t_{k+1}^{-1}\circ t_k\}_{k=1}^{n_z -1}$ and $\{s_{\mu+1}^{-1}\circ s_\mu\}_{\mu=1}^{n_p -1}$, and the extended transformations $T_{n_1} \circ S_{\nu_1}$ and $T_{n_2}\circ S_{\nu_2}$.
    
    Proving that each of these subgroups is normal is then a simple check. In particular, the subgroup generated by the transformations $t_1,\dots, t_{n_z}$ and $s_1,\dots, s_{n_p}$ is abelian, hence one can only work it out on the permutations.
    \end{proof}
Thanks to the Galois correspondence between covering spaces and subgroups of the fundamental group \cite{hatcher2002algebraic}*{Theorem 1.38}, we have:
\begin{thm}\label{thm:coveringEAW}
     For any choice of $n_1,n_2=1,\dots, n_z -1$ and $\nu_1,\nu_2=1,\dots, n_p -1$, the Hurwitz space $\Hw^\omega_{0,n_z }(n_z -(n_p +r)-1,r-1,\bm{0})$ can be normally covered by the orbit space of the extended affine-Weyl group $\widetilde{W}^{(n_1,n_2)}(\AN_{n_z -1})\,\widehat{\boxtimes}\,\,\widetilde{W}^{(\nu_1,\nu_2)}(\AN_{n_p -1})$.\\
     In particular, the covering map is locally given by:
     \begin{align}\label{EAW:eq:covering}
    \begin{split}
        \varphi_a&= i2\pi \bigl(\xi_a-\tfrac{n_1}{n_z } x_{n_z }- \tfrac{n_2}{n_z } y_{n_p }\bigr)\,, \qquad a=1,\dots, n_z\, ;\\
        \psi_\alpha&= i2\pi \bigl(\zeta_\alpha-\tfrac{\nu_1}{n_p }x_{n_z }-\tfrac{\nu_2}{n_p }y_{n_p }\bigr)\,, \qquad \alpha=1,\dots, n_p \,.
    \end{split}
\end{align}
\end{thm}
This is a generalisation \cite{zuo2020frobenius}*{Theorem 5.1}. In particular, as already discussed, that result only deals with the case $n_p=1$, but is able to further establish what choice of $n_1$ and $n_2$ leads to a local isomorphism of the Hurwitz Frobenius manifolds $\Hw^\omega_{0,n_z}(n_z-r-1,r-1,0)$ and the orbit space of the Ma-Zuo extended affine-Weyl group $\widetilde{W}^{(n_1,n_2)}(\AN_{n_z-1})$. Namely, it turns out that that is the case when $n_1=r$ and $n_2=r+1$.\\

Now, the question is whether one can endow the orbit space of these new kind of extended affine-Weyl groups $\widetilde{W}^{(n_1,n_2)}(\AN_{\ell_1})\,\widehat{\boxtimes}\,\,\widetilde{W}^{(m_1,m_2)}(\AN_{\ell_2})$ with a Frobenius manifold structure in such a way that, for some (possibly unique) choice of $n_1,n_2,m_1,m_2$, the covering map \cref{EAW:eq:covering} turns into a local isomorphism. This is, at present, still an open question. In order to answer such a question, one should first find a generating set for the ring of invariants with respect to the action of such groups and to endowing the representation space with a suitable metric so that, when one takes the Lie derivative of such a metric in the direction of the unity, one gets a flat metric, and its flat coordinates can then be constructed as polynomials in the generators of the ring of invariants. 

Now, the former task can be accomplished quite straightforwardly by an obvious generalisation of the results regarding the ring of invariants in \cites{zuo2020frobenius,ma2023frobenius}. Namely, if we denote $\bm{X}:=(x_1,\dots, x_{\ell_1})$ and $x:=x_{\ell_1+1}$ and similarly for the $y$-coordinates, the “basic” $\widetilde{W}^{(n_1,n_2)}(\AN_{\ell_1})\,\widehat{\boxtimes}\,\,\widetilde{W}^{(m_1,m_2)}(\AN_{\ell_2})$-invariant polynomials ought to be:
\begin{align*}
    u_a(\bm{X},\bm{Y},x,y)&:= e^{i2\pi (\inp{\omega_a}{\omega_{n_1}}x+\inp{\omega_a}{\omega_{n_2}}y)}\,\widetilde{u}_a(\bm{X})\,, && a=1,\dots, \ell_1\,,\\
    v_\alpha(\bm{X},\bm{Y},x,y)&:= e^{i2\pi (\inp{\rho_\alpha}{\rho_{m_1}}x+\inp{\rho_\alpha}{\rho_{m_2}}y)}\,\widetilde{v}_\alpha(\bm{Y})\,, && \alpha=1,\dots, \ell_2\,,\\
    u&:= i2\pi x\,,\\
    v&:= i2\pi y\,,
\end{align*}
where $\tilde{u}_1,\dots,\tilde{u}_{\ell_1}$ and $\tilde{v}_1,\dots, \tilde{v}_{\ell_2}$ are the Fourier polynomials, invariant with respect to the respective affine-Weyl group \cite{dubrovin_zhang_1998}:
\begin{align*}
    \widetilde{u}_a(\bm{X})&:=\frac{1}{n_a}\sum_{w\in W(\AN_{\ell_1})}e^{i2\pi \inp{w(\omega_a)}{\bm{X}}}=\sum_{1\leq \mu_1<\dots<\mu_a\leq \ell_1+1}e^{i2\pi (\xi_{\mu_1}+\dots+\xi_{\mu_a})},\\
        \widetilde{v}_\alpha(\bm{Y})&:=\frac{1}{n_\alpha}\sum_{w\in W(\AN_{\ell_2})}e^{i2\pi \inp{w(\omega_\alpha)}{\bm{Y}}}=\sum_{1\leq \mu_1<\dots<\mu_\alpha\leq \ell_2+1}e^{i2\pi (\zeta_{\mu_1}+\dots+\zeta_{\mu_\alpha})}.
\end{align*}

Furthermore, what the unity and the Euler vector field are has been already established. Therefore, the only residual ambiguity that we still have comes from fixing the intersection form in such a way that it matches with the above requirements. In particular, on the two subspaces $\mathfrak{h}^{(1)}_{\R}$ and $\mathfrak{h}^{(2)}_{\R}$ isomorphic to the two real Cartan subalgebras, the intersection form is, up to a constant, the Killing form of the corresponding Weyl group, and the vector-space direct sum turns into an orthogonal direct sum. As a consequence, we only need to fix it on the additional two-dimensional subspace upon which the affine-Weyl groups both act trivially.\\

On the other hand, we do know what the intersection form on the Hurwitz space is.
 Hence, if we conjecturally assume that there exists an orbit space description for any such structure, we can find out what the intersection form we are to fix is supposed to look like by working out the pull-back with respect to the covering map \cref{EAW:eq:covering}.
 
 The following two computational Lemmas deal with this problem:
\begin{lem}
    In its system of flat coordinates given in \cref{lem:EAW:flatcoordsIntForm}, the intersection form on the Hurwitz space $\Hw^\omega_{0,n_z }(n_z -(n_p +r)-1,r-1,\bm{0})$ is given by:
    \begin{align}
    \begin{split}\label{EAW:eq:intformHurwitz}
                g&=(\dd\phi_1)^2+\dots+(\dd \phi_{n_z })^2-(\dd\psi_1)^2-\dots-(\dd\psi_{n_p })^2+\\&\qquad -\tfrac1r\bigl(\dd\phi_1+\dots+\dd \phi_{n_z }-\dd\psi_1-\dots-\dd\psi_{n_p }\bigr)^2.
    \end{split}
    \end{align}
\end{lem}
\begin{proof}
    According to \cite{Dub04}*{Corollary 3.2}, if $\overset{\star}{\Fr}$ denotes the free energy for the almost-dual structure in a system of flat coordinates $\bm{z}:=(z_1,\dots, z_n)$ for the intersection form on some Frobenius manifold $M$, then:
    \[
    \Lie_E\overset{\star}{\Fr}= 2\,\overset{\star}{\Fr}+\tfrac{1}{1-d}g(\bm{z},\bm{z}).
    \]
    The almost-dual prepotential on the Hurwitz space under consideration is given in \citelist{\cite{Riley2007}*{Proposition 9} \cite{RileyPhD}*{Theorem 4.4}}. The statement of the Lemma is therefore a simple calculation.
\end{proof}
\begin{lem}
    The pull-back of the intersection form \cref{EAW:eq:intformHurwitz} via the covering map $\Phi$ given locally by \cref{EAW:eq:covering} splits in the following block-diagonal form:
    \begin{align}
            \Phi^*g&= (i2\pi)^2\bigl[\mathscr{C}^{(n_z -1)}\oplus -\mathscr{C}^{(n_p -1)}\oplus K\bigr],
    \end{align}
    where $\mathscr{C}^{(\ell)}$ denotes the Cartan matrix of $\AN_\ell$ and $K$ is the $2\times 2$ $x_{n_z }-y_{n_p }$ block:
    \begin{align*}
            K&:=\mqty[\tfrac{1}{n_z }n_1^2-\tfrac{1}{n_p }\nu_1^2-\tfrac{1}{r}(n_1-\nu_1)^2 & \tfrac{1}{n_z }n_1n_2-\tfrac{1}{n_p }\nu_1\nu_2-\tfrac1r (n_1-\nu_1)(n_2-\nu_2)\\
     \tfrac{1}{n_z }n_1n_2-\tfrac{1}{n_p }\nu_1\nu_2-\tfrac1r (n_1-\nu_1)(n_2-\nu_2) &\tfrac{1}{n_z }n_2^2-\tfrac{1}{n_p }\nu_2^2-\tfrac{1}{r}(n_2-\nu_2)^2  ].
    \end{align*}
\end{lem}
What is now missing is a representation-theoretic interpretation of these expressions, which would enable to define a sensible metric on the representation space that matches with the one above when the marked nodes $n_1,n_2,\nu_1$ and $\nu_2$ are chosen suitably.

\section{Conclusion and Outlook}
\documentclass[main.tex]{subfile}

The main component that is missing in these results is an analogue of a Chevalley-type theorem. In the Coxeter and extended-affine-Weyl constructions, and in the work \cites{Arsie_2022, zuo2020frobenius} with superpotentials with a single movable pole, such theorems are central to the construction of the orbit spaces on which the Frobenius structures live. The fact that the ring of invariant polynomials is finitely generated is then used, for example, to prove the smoothness of the orbit space. Here, the diagonal invariants $\Theta_r({\bm \alpha},{\bm \beta}) = \mathscr{P}_{1,r}({\bm \alpha},{\bm \beta})$ play an analogous role: they are polynomial functions of {\sl both} the zeros and poles of the superpotential are are invariant under the full symmetry group $\Sym_{n_z} \times \Sym_{n_p}$. However, the ring of such diagonal invariants is poorly understood (unlike the ring of coinvariants \cite{Gordon_2003}) and is certainly not freely generated. Nonetheless, the superpotential construction enables the Frobenius manifold to be constructed without such a theorem.\\

We end by describing an extension of the theory developed here to the $\mathscr{B}_\ell$-case, together with some open questions.
\begin{enumerate}
\item One could easily extend theses results by considering superpotentials with an additional $\mathbb{Z}/2\Z$-symmetry; that is, to superpotential of the form:
\[
\lambda(w) = \frac{1}{w^{2r}} \frac{\prod_{i=1}^{n_z} (w^2 - z_i^2) }{ \prod_{a=1}^{n_p} (w^2 - p_a^2) }.
\]
This would generalise the results in the works \cites{ZuoB,MaZuoBn,Dubrovin:2015wdx} to cases with multiple poles. This would be a routine computational exercise, generalising the results presented above.
    \item There is an important difference between zeros and poles of a rational function from the Hurwitz-space point of view. Namely, while allowing two zeros to coincide only takes you on the discriminant locus of the very same Hurwitz space, a non simple pole will modify the ramification profile at $\infty$, thus yielding a rational function in a different Hurwitz space. This could also be seen by looking at the Jacobian in \cref{eq:jacobian}. More concretely, let us fix $\ell,r,n_p,n\in\Z_{\geq 0}$, with $\ell>r+n_p$, and let $\bm{k}\in\Z^n_{\geq 0}$ such that $(n,\bm{k})\vdash n_p$. Then there is a natural surjection of $\Hw_{0,\ell+n_p+1}(\ell,\bm{0})$ onto $\Hw_{0,\ell+n_p+1}(\ell,\bm{k})$ (or, in the exponential case, of $\Hw_{0,\ell+r+n_p+1}(\ell,r-1,\bm{0})$ onto $\Hw_{0,\ell+r+n_p+1}(\ell,r-1,\bm{k})$)\footnote{Notice that the dimension of the target space as a quasi-projective variety is indeed, in both cases, smaller than the one of the source Hurwitz space. E.g.:
    \[
   \dim \Hw_{0,\ell+n_p+1}(\ell,\bm{k})=\ell+n_p+n\leq\ell+2n_p=\dim \Hw_{0,\ell+n_p+1}(\ell,\bm{0}).
    \]}. 
    
    When one endows both spaces with the same primary form, a natural question would therefore be whether one can say something about the Frobenius manifold structure on the target Hurwitz space based on the local description we have provided in the former. Notice that this would presumably involve some non-trivial coalescence limiting procedure yielding flat coordinates on the latter, presumably in the direction undertaken e.g. in \cites{Strachan_2004,FEIGIN2007143}, as the na\"ive approach clearly does not work at all. Notice that, in the exponential case and for $n=1$ (i.e. all the simple pole coalesce to a single order-$n_p$ pole), the Frobenius manifold structure on the target is known to be isomorphic to the one on the space of orbits of the extended affine-Weyl group $\widetilde{W}^{(r,r+n_p)}(\AN_{\ell+r+n_p})$ with two non-adjacent marked roots \cite{ma2023frobenius}.
    \item In both the two extremal exponential cases, i.e. whenever there is just one not-necessarily-simple pole on $\RS\smallsetminus\{0,\infty\}$, an orbit space description of the structure is known, as discussed, in terms of extended affine-Weyl groups of type $\AN$ with two marked roots. The question is, then, whether this can also be done for the “intermediate” cases.
    Of course, it is natural to start from the purely simple-pole case.
    
    To this end, we have shown in \cref{prop:subgroupreflection} that the monodromy group of the Hurwitz Dubrovin-Zhang Frobenius manifold $\Hw^\omega_{0,\ell+r+n_p+1}(\ell,r-1,\bm{0})$ contains a subgroup isomorphic to $\widetilde{W}^{(n_1,n_2)}(\AN_{\ell+r+n_p})\,\,\widehat{\boxtimes}\,\,\widetilde{W}^{(\nu_1,\nu_2)}(\AN_{n_p-1})$, as defined in \cref{eq:hattimes1,eq:hattimes2}, for any choices of the four marked roots $n_1,n_2=1,\dots, \ell+r+n_p$ and $\nu_1,\nu_2=1,\dots, n_p-1$. This determines a covering of the Hurwitz spaces by the orbit space of each of such groups. The peculiarity of the two extremal cases lies in the fact that the group on the right-hand side is trivial and therefore these subgroups will be naturally isomorphic to the one on the left, coming from extension of the symmetric group permuting the zeros. As a consequence, the Frobenius manifold structure on the space of orbits of such groups picks out exactly one of these coverings, namely the only one that is locally an isomorphism.
    
    In general, a Frobenius manifold structure on the space of orbits when neither group is trivial is still missing. The conjecture is that, when that has been constructed, even in the non-extremal cases only one of such coverings will realise a local Frobenius manifold isomorphism.
\end{enumerate}

We hope to address these problems in future work.

\phantomsection
\addcontentsline{toc}{section}{References}
\bibliography{biblio} 

\end{document}